\documentclass[11pt,letterpaper]{article}
\usepackage[margin=1in]{geometry}
\usepackage[utf8]{inputenc}
\usepackage{times}

\usepackage[UKenglish]{babel}
\usepackage[UKenglish]{isodate}

\usepackage{amsmath,amssymb,amsthm,bm}
\usepackage{hyperref, color, xcolor}
\definecolor{oxfordColor}{HTML}{002147}
\hypersetup{
    colorlinks   = true,                    
    linkcolor    = blue!70!black,
    anchorcolor  = gray,
    citecolor    = blue!70!black,
    filecolor    = red,
    menucolor    = green,
    runcolor     = red,
    urlcolor     = [rgb]{0,0.2,0.5}
}

\usepackage{graphicx}
\usepackage{float}
\usepackage[title]{appendix} 
\usepackage{caption}
\usepackage{subcaption}

\usepackage[capitalize,nameinlink]{cleveref}
\usepackage[algo2e,ruled,linesnumbered]{algorithm2e}
\SetKwComment{Comment}{/* }{ */}

\newtheorem{theorem}{Theorem}
\newtheorem{lemma}[theorem]{Lemma}

\newtheorem{proposition}[theorem]{Proposition}

\newtheorem{corollary}[theorem]{Corollary}
\newtheorem{remark}[theorem]{Remark}
\newtheorem{example}[theorem]{Example}

\newtheorem{definition}[theorem]{Definition}

\newtheorem*{thmsampling}{Theorem~\ref{thm:sampling}}
\newtheorem*{thmconnectivity}{Theorem~\ref{thm:connectivity}}
\newtheorem*{thmlooseness}{Theorem~\ref{thm:looseness}}

\newtheorem*{lemccgeneral}{Lemma~\ref{lem:cc-general}}
\newtheorem*{lemsample}{Lemma~\ref{lem:sample}}
\newtheorem*{lemmt}{Lemma~\ref{lem:mixing-time}}
\newtheorem*{lembad}{Lemma~\ref{lem:bad}}
\newtheorem*{lembadall}{Lemma~\ref{lem:bad:2}}
\newtheorem*{lemsi}{Lemma~\ref{lem:si}}
\newtheorem*{propbad}{Proposition~\ref{prop:bad}}

\newcommand{\faild}{\mathcal{F}_{\mathrm d}}
\newcommand{\failu}{\mathcal{F}_{\mathrm u}}
\newcommand{\verticesd}{\mathcal{V}_{\mathrm d}}

\let\epsilon=\varepsilon

\def\gv{\mathcal{V}_{\mathrm{good}}}
\def\bv{\mathcal{V}_{\mathrm{bad}}}
\def\gc{\mathcal{C}_{\mathrm{good}}}
\def\bc{\mathcal{C}_{\mathrm{bad}}}
\def\var{\operatorname{var}}
\def\pr{\operatorname{Pr}}
\def\density{\alpha}
\def\rh{r_0}

\def\densitydef{2^{(r_0- \delta) k} /k^3 }
\def\densitymaindef{2^{r k}}
\def\mdegdef{\lceil 2^{(r_0-\delta) k} \rceil}

\def\hd{\operatorname{HD}}
\def\sizecc{2 k^4 \lceil \log (n / \varepsilon) \rceil}
\def\mtdef{\lceil 2^{2k+3} n^\theta \log \frac{2n}{\varepsilon^2} \rceil}
\def\rvalue{0.117841}

\def\rone{0.227092}

\def\deltadef{0.00001}

\DeclareMathOperator{\TV}{TV}
\DeclareMathOperator{\mix}{\mathrm{mix}}
\DeclareMathOperator{\bad}{\mathrm{bad}}

\newcommand\restr[2]{{% we make the whole thing an ordinary symbol
  \left.\kern-\nulldelimiterspace % automatically resize the bar with \right
  #1 % the function
  \right|_{#2} % this is the delimiter
  }}

\usepackage{algorithm}
\usepackage{algorithmic}

\usepackage{xpatch}
\xpatchcmd{\algorithmic}{\setcounter}{\algorithmicfont\setcounter}{}{}
\providecommand{\algorithmicfont}{}
\providecommand{\setalgorithmicfont}[1]{\renewcommand{\algorithmicfont}{#1}}

\setalgorithmicfont{\normalsize}
\usepackage{tikz}
\usetikzlibrary{arrows.meta}
\usetikzlibrary{shapes}
\usetikzlibrary{positioning}

\title{Fast sampling of satisfying assignments from random $k$-SAT with applications to connectivity\thanks{A preliminary version of the results (with weaker bounds and different proofs) appeared in the proceedings of SODA 2023 by
a subset of the authors. For the purpose of Open Access, the
authors have applied a CC BY public copyright licence to any Author Accepted Manuscript version arising
from this submission. All data is provided in full in the results section of this paper.}}

\author{ Zongchen Chen\thanks{Department of Computer Science and Engineering, University at Buffalo, Buffalo, NY 14260, USA.}
\and Andreas Galanis \thanks{ Department of Computer Science, University of Oxford, Wolfson Building, Parks Road, Oxford, OX1~3QD, UK.}
  \and
  Leslie Ann Goldberg \footnotemark[3]
  \and Heng Guo \thanks{School of Informatics, University of Edinburgh, Informatics Forum, 10 Crichton Street, Edinburgh, EH8 9AB, UK. HG has received funding from the European Research Council (ERC) under the European Union's Horizon 2020 research and innovation programme (grant agreement No.~947778).}
\and
Andr\'es Herrera-Poyatos \thanks{Andalusian Institute of Data Science and Computational Intelligence (DaSCI), University of Granada, ESP. This author was supported by an Oxford-DeepMind Graduate Scholarship and an EPSRC Doctoral Training Partnership.}
\and Nitya Mani\thanks{Department of Mathematics, Massachusetts Institute of Technology, Cambridge, MA 20139, USA.}
\and Ankur Moitra\footnotemark[6]
 }

\date{4th August 2024}

\begin{document}
\maketitle
 \thispagestyle{empty}

\begin{abstract}
We give a nearly linear-time algorithm to approximately sample satisfying assignments in the random $k$-SAT model when the density of the formula scales exponentially with $k$. The best previously known sampling algorithm for the random $k$-SAT model applies when the density $\alpha=m/n$ of the formula  is less than $2^{k/300}$ and runs in time $n^{\exp(\Theta(k))}$ (Galanis, Goldberg, Guo and Yang, SIAM J. Comput., 2021). Here $n$ is the number of variables and $m$ is the number of clauses.
Our algorithm achieves a significantly faster running time of $n^{1 + o_k(1)}$ and samples satisfying assignments up to density $\alpha\leq 2^{0.039  k}$.

The main challenge in our setting is the presence of  many variables with unbounded degree, which causes significant correlations within the formula and impedes the application of relevant Markov chain methods from the bounded-degree setting (Feng, Guo, Yin and Zhang, J. ACM, 2021; Jain, Pham and Vuong, 2021). Our main technical contribution is a $o_k(\log n )$ bound of the sum of influences in the $k$-SAT model which turns out to be robust against the presence of high-degree variables. This allows us to apply the spectral independence framework and obtain fast mixing results  of a uniform-block Glauber dynamics on a carefully selected subset of the variables. The final key ingredient in our method is to take advantage of the sparsity of logarithmic-sized connected sets and the expansion properties of the random formula, and establish relevant connectivity properties of the set of satisfying assignments that enable the fast simulation of this Glauber dynamics. 

Our results also allow us to conclude that, with high probability, a random $k$-CNF formula with density at most $2^{0.227 k}$ has a giant component of solutions that are connected in a graph where solutions are adjacent if they have Hamming distance $O_k(\log n)$. We are also able to deduce looseness results for random $k$-CNFs in the same regime.
\end{abstract}

\newpage
\pagenumbering{arabic}

\section{Introduction} \label{sec:intro}

The random $k$-SAT model is 
a foundational model in the study of randomised algorithms.
For integers $k,n,m \ge 2$, the random formula $\Phi = \Phi(k, n, m)$ is a $k$-CNF formula chosen uniformly at random 
from the set of formulae with $n$ Boolean variables and $m$ clauses,  where each clause has $k$ literals (repetitions allowed).  Here, we consider the sparse regime where the density of the formula, $\alpha=m/n$, is bounded by an absolute constant. An important question is 
determining the probability that the random formula is satisfiable as a function of its density (in the limit $n\rightarrow \infty$). Interestingly, for all sufficiently large $k$, the probability that $\Phi$ is satisfiable drops abruptly from $1$ to $0$ when the density $\alpha$ crosses a certain threshold $\alpha_{\star}(k)$. Recently there has been  tremendous progress in establishing this phase transition, concluding that $\alpha_{\star}(k) = 2^k \log 2 -\tfrac{1}{2} (1 + \log 2) + o_k(1)$ as $k \to \infty$~\cite{ding2015, coja2016}. Despite the  good progress on pinning down  this phase transition, finding satisfying assignments for densities up to $\alpha^*$ poses severe challenges. In fact, the best known algorithm~\cite{coja2010} for finding a satisfying assignment of a random formula $\Phi$ succeeds up to densities $(1+ o_k(1)) \tfrac{2^k}{k}\log k$, and going beyond such densities is a major open problem with links to phase transitions~\cite{AC08}.

Lately there has been significant interest in the related computational problem of sampling satisfying assignments of $\Phi$ uniformly at random. This problem is closely connected to the problem of estimating the number of satisfying assignments of $\Phi$, also known as the value of the partition function of the model. From a probabilistic viewpoint, the analysis of the partition function depends on subtle properties 
of the solution set~$\Omega = \Omega_{\Phi}$ consisting of the satisfying assignments of $\Phi$~\cite{2SAT,AminBelief2, SlySunZhang, montanari2007}. 
In this direction, there has been substantial work on finding the so-called free energy of the model, i.e., the asymptotic value of the quantity $\tfrac{1}{n}\mathbf{E}[\log (1+|\Omega|)]$. 
Computing the $k$-SAT free energy is 
a difficult problem which is still open
(roughly, the difficulty comes from the asymmetry of the model and the unbounded degrees), but there have been   results for 
closely related models including  the permissive version of the model ~\cite{AminBelief2,montanari2007,CR13},  the regular $k$-SAT model~\cite{coja-oghlan_wormald_2018}, and the regular NAE-SAT model~\cite{nam2020onestep,SlySunZhang}. Very recently, a formula for the free energy of the $2$-SAT model was given in~\cite{2SAT}. 

Regarding the algorithmic problem of sampling satisfying assignments uniformly at random, in the random $k$-SAT model progress has been slower relative to other well-studied models on random graphs (such as $k$-colourings or independent sets). One of the main reasons for this is that the usual distribution properties that are typically used to obtain fast algorithms (such as correlation decay and spatial mixing) fail to hold for densities as low as $\alpha=o_k(1)$~\cite{montanari2007}.  These issues are in fact present already in the bounded-degree $k$-SAT setting, where the formulae are worst-case but every variable is constrained to have a bounded-number of occurrences. For random formulae, these issues are further aggravated by the fact that the degrees of a linear number of variables are unbounded.
Very recently, the authors of~\cite{galanis2019counting} gave an approximate counting algorithm (FPTAS) for the number of satisfying assignments of $\Phi$ when $k$ is large enough and $\alpha \lesssim 2^{k / 300}$ (where $\lesssim$ hides a polynomial factor in~$1/k$). This algorithm elevates Moitra's counting method for bounded-degree $k$-SAT~\cite{moitra19} to the random formula setting, and is the first polynomial-time approximate-counting algorithm to achieve an exponential-in-$k$ bound on $\alpha$. However, its running time is  $n^{\exp(\Theta(k))}$ because the algorithm repeatedly has to enumerate local structures (including solving LPs as a subroutine), which does not scale well with $k$. Hence, the problem of  
finding a \emph{fast} algorithm for sampling the satisfying assignments in  the random $k$-SAT model has remained open.

In this work we give a fast algorithm that in time $n^{1+o_k(1)}$ approximately samples satisfying assignments of a random $k$-SAT formula of density $\alpha \leq 2^{0.039 k}$, within arbitrarily small polynomial error.
Our work also delves into the connections between the solution space geometry of $k$-CNF $\Phi$ and algorithms for efficiently sampling from the solutions of $\Phi$.

A unifying theme of previous approaches to counting and sampling CSP solutions is a tool called \textit{marking}, first introduced in~\cite{moitra19}, which finds a set of ``marked" variables such that the set of satisfying assignments projected on these variables is connected. Marking is also an essential step in the developing of our sampling algorithm. Our algorithm first runs a Markov chain to sample assignments of a judiciously-chosen subset of marked variables of $\Phi$ (from the relevant marginal distribution), and subsequently extending this random assignment to all the variables. This has the advantage that it avoids the enumeration of local structures, and in fact achieves a nearly-linear running time. We give a high-level overview of the techniques developed in our proofs in Section~\ref{sec:intro:techniques}. Roughly, our Markov chain is a uniform-block Glauber dynamics which, interestingly, mixes quickly despite the presence of high-degree variables in the random formula. The main point of departure from similar approaches that have been applied to the bounded-degree setting is that we completely circumvent sophisticated coupling arguments that have been used there and which are unfortunately severely constricted by the unbounded degrees in our setting (and made inapplicable). Instead, our main technical contribution is to show that the stationary distribution of our chain is $(c^{k} \log n)$-spectrally independent for some constant $c \in (0,1)$, allowing us to apply recently-developed tools in the analysis of Markov chains. Unlike most applications of spectral independence, our proof does not rely on correlation decay (which, as we mentioned, fails to hold for densities exponential in $k$). We show our spectral-independence bounds by relating the probabilistic properties of the solution space with the structure of the formula using coupling techniques, so that we can exploit local sparsity properties of random $k$-SAT.

To formally state our main result, we say that an event $\mathcal{E}$ regarding the choice of the random formula $\Phi$ holds \emph{with high probability} (abbreviated w.h.p.) if $\pr(\mathcal{E}) = 1 - o(1)$ as $n \to \infty$.
The total variation distance between  two probability distributions $\mu$ and $\nu$ over the same space $\Omega$ is given by $\tfrac{1}{2} \sum_{x \in \Omega} \lvert \mu(x) - \nu(x) \rvert$ and is denoted by $d_{\TV}(\mu, \nu)$. Our main result can now be stated as follows.

\newcommand{\statethmsampling}{
For any real $\theta \in (0,1)$, there is $k_0 \ge 3$ with $k_0 = O(\log(1/\theta) )$ such that, for any integers $k\ge k_0$ and $ \xi \ge 1$, and for any positive real $\alpha \le 2^{0.039 k}$, the following holds. 

There is an efficient algorithm to sample from the satisfying assignments of a random $k$-CNF formula $\Phi = \Phi(k, n, \lfloor \alpha n \rfloor)$ within $n^{-\xi}$ total variation distance of the uniform distribution. The algorithm  runs in time $O( n^{1+\theta} )$, and succeeds w.h.p. over the choice of $\Phi$.
} 
\begin{theorem} \label{thm:sampling}
	\statethmsampling
\end{theorem}

Using standard techniques from the literature, this $O( n^{1+\theta} )$ uniform sampling algorithm can be used to obtain a  randomised approximation scheme for counting satisfying  assignments of $\Phi$ in time $O(n^{2+\theta}/\varepsilon^2)$, where $\varepsilon$ is the multiplicative error, see~\cite[Section~7]{feng2020} and Remark~\ref{rem:counting} for details. 

Our results can be applied to analyse the solution space geometry of random $k$-CNF formulae for the densities under consideration. Many involved heuristics in statistical physics make predictions about the geometry of the solution space of a random $k$-CNF instance, often depicted in diagrams like Figure~\ref{fig:phases}. Some phases and transitions in this diagram are precisely understood. For example, as mentioned above, the satisfiability threshold (pictured in the transition to the rightmost image in Figure~\ref{fig:phases}) was determined by~\cite{ding2015}. Another transition of interest is the clustering threshold, above which the solution space of a random $k$-CNF shatters into exponentially many linearly separated connected components, each of which contains an exponentially small fraction of the satisfying assignments of the formula, as rigorously understood in~\cite{CP12, ACR11, MMZ05, MMZ205}.

\begin{figure}[H]
    \centering
    \includegraphics[width=.9\textwidth]{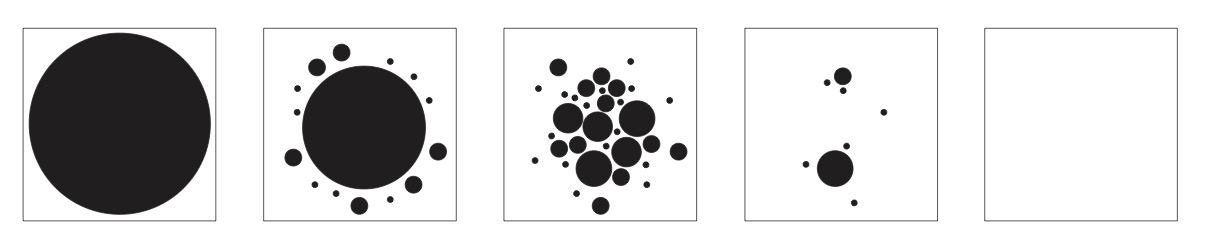}
    \caption{Heuristic phase diagrams such as above~\cite{KMR07} depict the predicted evolution of the structure of the solution space of a random $k$-CNF as the density $\alpha$ of the formula increases from left to right. We primarily study the leftmost regime.}
    \label{fig:phases}
\end{figure}

In the lower-density regime, the solution space geometry of random $k$-CNFs appears poorly understood. It is widely believed that beneath a critical clause density, the solution space of a random $k$-CNF is “connected.” However, from the literature, it is not even clear what “connected” means. Connectivity is sometimes used in the statistical physics literature as a characterization of the entropy or energy profile of the solution space of a random $k$-CNF formula as in~\cite{Zde08}. In such settings, connectivity is often characterized by an absence of clustering behavior, leaving somewhat of a mystery as to the graphical properties of the solution space of a low density random $k$-CNF.

Conjectures about connectivity take different forms, and different notions of what connectivity might mean are articulated in~\cite{Zde08, KMR07, CP12}. The most common precise notion of connectivity is with respect to Hamming distance, i.e. understanding connectivity properties of the graph of solutions to a random $k$-CNF, where solutions are $f(n)$-connected if their Hamming distance is at most $f(n)$. At lower densities, random $k$-CNFs still can have isolated solutions far in Hamming distance from other satisfying assignments. However, the prevailing belief is that below some threshold, the overwhelming majority of solutions to a random $k$-CNF lie in a giant component that is $o(n)$-connected.

Much more is known about related notions and local versions of connectivity, like looseness, which characterises how rigid a particular satisfying assignment is. Roughly speaking, a satisfying assignment to a formula is $f(n)$-loose if any variable can be flipped to yield a new satisfying assignment by changing at most $f(n)$ additional variable assignments. In~\cite{AC08}, the authors showed $o(n)$-looseness holds in the connectivity regime for related, simpler random models, random $q$-coloring, and hypergraph $2$-coloring, conjecturing that $o(n)$-looseness holds for random $k$-CNF instances below the clustering threshold. This conjecture was partially resolved in~\cite{CP12}, where in an analysis of the decimation process for random $k$-SAT, the authors observed that with high probability over formulae and satisfying assignments, at least $99\%$ of the variables were $O(\log n)$-loose. Looseness, however, is a local notion, not a global one. The set of elements in $\{0,1\}^n$ that have Hamming weight at least $2n/3$ or at most $n/3$ is $1$-loose, but $\Omega(n)$-connected.

We will concern ourselves with the following precise notion of connectivity. 

\begin{definition}[$D$-Connectivity]
Let $\Phi = (\mathcal{V}, \mathcal{C})$ be a $k$-CNF formula. For any assignment $\Lambda \colon \mathcal{V} \to \{\mathsf{F}, \mathsf{T}\}$, let $\| \Lambda \|_1$ be the number of variables $\Lambda$ assigns to be $\mathsf{T}$. Throughout, we implicitly consider variable assignments in $\mathbb{F}_2^n$, so $\|\cdot\|_1$ encodes Hamming weight and $\|\Lambda_1 - \Lambda_2\|_1$ encodes Hamming distance.

We say a sequence of satisfying assignments $\zeta_0 \leftrightarrow \zeta_1 \leftrightarrow \cdots \leftrightarrow \zeta_{\ell}$ of $\Phi$ is a \emph{$D$-path} if $\|\zeta_{i} - \zeta_{i-1}\|_1 \le D$ for each $i \in [\ell]$. 
We say two satisfying assignments of $\Phi$, $\Lambda, \Lambda' \in \Omega$, are \textit{$D$-connected} if there exists a $D$-path connecting $\Lambda$ and $\Lambda'$ (that is, $\zeta_0 = \Lambda$ and $\zeta_\ell = \Lambda'$).
\end{definition}

Marking-based deterministic and MCMC algorithms are mysterious at first glance, as they enable counting and sampling of $k$-CNF solutions even in regimes where the solution space is disconnected (i.e. not $1$-connected). In this work, we leverage the idea of marking in a novel way to construct paths that certify global connectivity properties of the solution space of $k$-CNFs at densities close to where counting algorithms are known.

%Density here is 0.227092
\newcommand{\statethmconnectivity}{
    There is $k_0 \ge 3$ and a polynomial $p(k)$ with non-negative integer coefficients such that, for any integer $k\ge k_0$, and for any positive real $\alpha \le 2^{0.227 k}$, the following claim holds with high probability over the choice of a random $k$-CNF formula $\Phi = \Phi(k, n, \lfloor \alpha n \rfloor)$. Two satisfying assignments chosen uniformly at random are $p(k) \log(n)$-connected with probability at least $1 - 1/n$. 
}
\begin{theorem} \label{thm:connectivity} \label{t:robustconn}
  \statethmconnectivity
\end{theorem}

In fact, we show it suffices to take $p(k) = 2k^5$. Our new applications of marking also have implications for other, more local, structural properties of the $k$-CNF solution space, like looseness.

\begin{definition}\label{d:loose}
Given a $k$-CNF formula $\Phi = (\mathcal{V}, \mathcal{C})$ and a satisfying assignment $\Lambda$, a variable $v \in \mathcal{V}$ is $f(n)$-\emph{loose} with respect to $\Lambda$ if there exists a satisfying assignment to $\Phi$, $\tau \in \Omega$, with $\tau(v) \neq \Lambda(v)$ and $\|\Lambda - \tau\|_1 \le f(n)$.
 
For a random $k$-CNF formula $\Phi = \Phi(k, n, m)$ and a satisfying assignment  $\Lambda$ chosen uniformly at random, we say that $\Phi$ is \emph{$f(n)$-loose} if with high probability over $(\Phi, \Lambda)$, all variables $v \in V$ are $f(n)$-loose with respect to $\Lambda$.
\end{definition}

We observed earlier that looseness does not imply connectivity; in fact, the other direction of implication is also false as looseness is an incomparable goal to connectivity. Looseness requires that \textit{locally}, we are able to flip \textit{any} variable and get to a nearby solution rather than merely the existence of a path away from a solution. Nonetheless, we are able to deduce some nontrivial results about the looseness of the solution space of  random $k$-CNFs.

\newcommand{\statethmlooseness}{
    There is $k_0 \ge 3$ such that, for any integer $k\ge k_0$, and for any positive real $\alpha \le 2^{0.227 k}$, the random $k$-CNF formula $\Phi(k, n, \lfloor \alpha n \rfloor)$ is $\operatorname{poly}(k)\log(n)$-loose. 
}
\begin{theorem} \label{thm:looseness} \label{t:robustloose}
    \statethmlooseness
\end{theorem}

We note here that, independently of this work, He, Wu, and Wang~\cite{Kun1} also obtained sampling algorithms for random $k$-CNF formulae. The approach of~\cite{Kun1} is based on bounding chains following the recursive sampler method developed in~\cite{Anand2022, Kun2, Kun3}. 
Their algorithm works up to densities roughly equal to $2^{k/3}$ and samples satisfying assignments within $\varepsilon$ total variation distance of the uniform distribution in time $(n/\varepsilon)^{1 + O(k^{-5})}$. 

\section{Proof outline} \label{sec:intro:techniques} \label{sec:proof-outline} \label{sec:po}

Our  nearly linear-time sampling algorithm is based on running a Markov chain; this is a standard technique in approximate counting, where typically one runs a Markov chain on the whole state space that converges to the desired distribution. The twist in $k$-SAT is that the state space of the Markov chain needs to be carefully selected in order to avoid certain bottleneck phenomena that impede fast convergence. This approach has been recently applied to bounded-degree $k$-CNF formulae~\cite{feng2020, feng2021,   jain2021sampling} building on the work of Moitra~\cite{moitra19}  and using the Markov chain known as single-site Glauber dynamics. The main difficulties in all of these works are that the usual distribution properties that are typically used to obtain fast algorithms (such as correlation decay and spatial mixing) fail on the set of all SAT solutions, and in fact even ensuring a connected state space is a major problem. Working around this is one of the main challenges for us too, and in the random $k$-SAT setting it is further aggravated  by the fact that a linear number of variables have degrees much higher than average. In fact, w.h.p., a good portion of vertices have degrees depending on $n$. with the maximum degree of the formula scaling as $\log n / \log \log n$.

This poses several new challenges for the Markov chain approach to work in our setting. First of all, we have to ensure that the set of satisfying assignments that our Markov chain considers has good connectivity properties. We address this problem in Section~\ref{sec:po:marking} of this proof outline, where we find a suitable subset of marked variables where we can run the Glauber dynamics; this part is inspired by Moitra's ``marking'' approach, though here we need to add an extra layer of marking to facilitate later the analysis of the Markov chain. Second and more importantly, state-of-the-art arguments for bounding the mixing time of the single-site Glauber dynamics on $k$-CNF formulae, such as~\cite{feng2020}, break under the presence of high-degree variables. We focus on this in Section~\ref{sec:po:si}, where we outline a novel argument that analyses the mixing time of the uniform-block Glauber dynamics using recent advances in spectral independence~\cite{AlevLau,KO20,alg20, Chen2020SI}. This is the first application of the spectral-independence framework for $k$-CNF formulae, where the absence of correlation decay limits the application of standard techniques (based on self-avoiding walk trees~\cite{alg20,Chen2020SI}). To obtain our spectral-independence bounds we need to combine the probabilistic structure of satisfying assignments with the local sparsity properties of the random formula. The third challenge in our approach is simulating the individual steps of the uniform-block Glauber dynamics since they involve updating a linear number of variables, making the computation of the transition probabilities more challenging. To this end, we need to initialise our block Glauber dynamics to random values (instead of an arbitrary assignment that is typically used as initialisation), and show that the formula breaks into small tree-like connected components that allows us to do the relevant computations throughout the algorithm's execution (cf. Section~\ref{sec:po:structure}).  Based on these pieces, the full algorithm is presented in Section~\ref{sec:po:alg}. 

The fact that the formula breaks into small tree-like connected components when marked variables are assigned random values will also allow us to analyse the geometry of the space of satisfying assignment of the random formula, and we will delve into this connection in Section~\ref{sec:po:structure}.

\subsection{Marking variables in the random $k$-SAT model} \label{sec:po:marking}

In order to ensure good connectivity properties which are essential for fast convergence of the relevant Markov chain, our algorithm runs  Glauber dynamics on a large subset $\mathcal{V}_{\mathrm m}$ of so-called ``marked" variables  of the random formula, leaving the rest of the variables unassigned. The variables in $\mathcal{V}_{\mathrm m}$ are chosen in a way that ensures that their marginals are near~$1/2$, which is important for ensuring rapid mixing. Moitra~\cite{moitra19} introduced a random ``marking'' procedure to identify such a subset of variables in the bounded-degree case. The presence of high-degree variables impedes a direct application of this technique in the random-formula setting, but in~\cite{galanis2019counting} the authors show that by temporarily removing a small linear number of ``bad" clauses that contain high-degree variables, one can also achieve marginals near~$1/2$ for an appropriate set of variables in the random $k$-SAT model. Here, we further refine these arguments, as we need more control over the high-degree variables of the formula in order to conclude rapid mixing of the Glauber dynamics.  Recall that the degree of a variable $v$ is the number of occurrences of literals involving the variable $v$ in $\Phi$ and that the maximum degree of the formula $\Phi$ is the maximum degree among its variables. The following important definitions will be used throughout the paper. We usually use $\mathcal{V}$ to denote the set of variables and $\mathcal{C}$ to denote the set of clauses of a $k$-CNF formula $\Phi$. For any $c \in \mathcal{C}$ we denote by $\var(c)$ the set of variables appearing in $c$, and for any $S \subseteq \mathcal{C}$ we denote $\var(S) = \bigcup_{c \in S} \var(c)$.

\begin{definition} [high-degree, $\Delta_r$]
\label{def:degree}
  Let $r \in (0,1)$ and let $k \ge 3$ be an integer. Let $\Phi = (\mathcal{V}, \mathcal{C})$ be a $k$-CNF formula. We say that a variable $v \in \mathcal{V}$ is \emph{high-degree} if the degree of $v$ is at least $\Delta_r := \lceil 2^{rk} \rceil$. 
\end{definition}

We refer to Section~\ref{sec:bad} for details on our procedure to determine the bad variables/clauses of the formula~$\Phi$. Roughly, bad variables consist of high-degree variables (as in Definition~\ref{def:degree}), plus those variables that appear in a clause with at least two other bad variables (recursively); bad clauses are those clauses that contain at least three bad variables. We use $\bv(r)$ and $\bc(r)$ to denote the sets of bad variables and clauses. We use $\gv(r) = \mathcal{V} \setminus \bv(r)$ to denote the set of \emph{good variables}, and  $\gc(r) = \mathcal{C} \setminus \bc(r)$ to denote the set of \emph{good clauses}.  The following proposition,  proved in Section~\ref{sec:bad}, summarises the main properties of the above sets.
  
\newcommand{\statepropbad}{
  Let $\Phi = (\mathcal{V}, \mathcal{C})$ be a $k$-CNF formula. For any $c \in \gc(r)$, we have $\lvert \var(c) \cap \bv(r) \rvert \le 2$, and for any $c \in \bc(r)$, we have $\lvert \var(c) \cap \gv(r) \rvert = 0$. Moreover, every good variable has degree less than $\Delta_r$. There is a procedure to determine $\bc$ that runs in time $O(n + m k)$, where $n$ is the number of variables of $\Phi$ and $m$ is the number of clauses of $\Phi$. 
}

\begin{proposition} \label{prop:bad}
\statepropbad
\end{proposition}

It turns out that, w.h.p. over the choice of $\Phi$, most clauses (and variables) in the random formula $\Phi$ are good, see Lemma~\ref{lem:bad} for a precise statement. At this stage, it would be natural to try to   rework the Markov chain approach of~\cite{feng2020}.  To do this, we would split the set of good variables into \emph{marked variables} and \emph{control variables} in such a way that marked variables have marginals close to $1/2$. Then we  run the Glauber dynamics on the set of marked variables. However, as we explain in Section~\ref{sec:po:si}, the state-of-the-art techniques used to analyse the mixing time of the single-site Glauber dynamics on bounded-degree formulae do not generalise to the random $k$-SAT setting; the main reason for this is that they fail to capture the effect that the high-degree variables have on the marginal probabilities of other variables. Therefore, we need to develop an alternative approach that is robust against the presence of high-degree variables. Our main contribution is an argument to apply the spectral independence framework~\cite{Chen2020SI, Chen2021} to the random $k$-SAT model that leads to nearly linear sampling algorithms. To do this, it is important to introduce a third type of good variables, which we call the \emph{auxiliary variables}.  This motivates the following definition of marking.

 \begin{definition}[$\rho$-distributed, $(r, r_{\mathrm m}, r_{\mathrm a}, r_{\mathrm c})$-marking, $r_0$, $r_1$, $\delta$] \label{def:distributed-marking}
  Let $r\in(0,1)$. Let $\Phi = (\mathcal{V}, \mathcal{C})$ be a $k$-CNF formula and let $V$ be a subset of $\gv(r)$. We say that $V$ is \emph{$\rho$-distributed} if for each $c \in \gc(r)$ we have $\lvert \var(c) \cap V \rvert \ge \rho (k-3)$.  An \emph{$(r, r_{\mathrm m}, r_{\mathrm a}, r_{\mathrm c})$-marking} of $\Phi$ is a partition $(\mathcal{V}_{\mathrm m}, \mathcal{V}_{\mathrm a}, \mathcal{V}_{\mathrm c})$ of the variables of $\Phi$ such that 
  \begin{enumerate}
      \item \label{item:def-marking:2} the set of good variables $\mathcal{V}_{\mathrm m}$ is $r_{\mathrm m}$-distributed;
      \item \label{item:def-marking:3} the set of good variables $\mathcal{V}_{\mathrm a}$ is $r_{\mathrm a}$-distributed.
      \item \label{item:def-marking:1} $\mathcal{V}_{\mathrm c}$ contains all the bad variables and the set $\mathcal{V}_{\mathrm c} \setminus \bv(r)$ is $r_{\mathrm c}$-distributed;
  \end{enumerate}
  The variables in $\mathcal{V}_{\mathrm m}$ are called \emph{marked} variables, the variables in $\mathcal{V}_{\mathrm a}$ are called \emph{auxiliary} variables, and the variables in $\mathcal{V}_{\mathrm c}$ are called \emph{control} variables.
  
  In our sampling algorithm we work with $r = r_0 - \delta$ for $r_0 := \rvalue$ and $\delta := \deltadef$, and work with an $(r, r_0, r_0, 2r_0)$-marking. In our connectivity results (Theorems~\ref{thm:connectivity} and~\ref{thm:looseness}) we choose $r = r_1 - \delta$ for $r_1 := 0.227092$ and work with an $(r, r_1, 0, r_1)$-marking in order to achieve the larger density threshold.
\end{definition}

In Section~\ref{sec:marking} we show that random $k$-CNF formulae have $(r_0-\delta, r_0, r_0, 2r_0)$-markings when the density $\alpha$ is below the threshold $\densitydef$, and that the marginals of good variables are close to $1/2$; this is where the value of $r_0$ becomes important in the argument. We also show that random $k$-CNF formulae have $(r_1-\delta, r_1, 0, r_1)$-markings when the density $\alpha$ is below the threshold $2^{(r_1-\delta) k}/k^3$. We state this result for $r_0$ in Proposition~\ref{prop:marginals} below; first we give some relevant definitions.
 
\begin{definition}[$\Omega^*$, $\mu_A$, $\Omega$, $\Phi^\Lambda$, $\mathcal{C}^{\Lambda}$, $\mathcal{V}^{\Lambda}$, $\Omega^\Lambda$] \label{def:mu} \label{def:subformula}
  Let $\Phi = (\mathcal{V}, \mathcal{C})$ be a $k$-CNF formula. Let $\Omega^*$ be the set of all assignments $\mathcal{V} \to \{\mathsf{F}, \mathsf{T}\}$. Given any subset $A \subseteq \Omega^*$, let $\mu_A$ be the uniform distribution on $A$. Let $\Omega$ be the set of satisfying assignments of $\Phi$. For any partial assignment $\Lambda$ we denote by $\Phi^\Lambda$ the formula obtained by simplifying $\Phi$ under $\Lambda$, i.e., removing the clauses which are already satisfied by $\Lambda$, and removing false literals  from the remaining clauses. We denote by $\mathcal{C}^{\Lambda}$ and $\mathcal{V}^{\Lambda}$ the sets of clauses and variables of $\Phi^\Lambda$. Moreover, we denote by $\Omega^\Lambda$ the set of satisfying assignments of $\Phi^\Lambda$. 
\end{definition}
 
 \newcommand{\statepropmarginals}{
 There is an integer $k_0$ such that for any $k \ge k_0$ and any density $\alpha$  with $\density \le \densitydef$ the following holds w.h.p. over the choice of the random $k$-CNF formula $\Phi = \Phi(k, n, \lfloor \alpha n\rfloor)$. There exists an $(r_0-\delta, r_0, r_0, 2r_0)$-marking $(\mathcal{V}_{\mathrm m}, \mathcal{V}_{\mathrm a}, \mathcal{V}_{\mathrm c})$ of $\Phi$. Moreover, for any such marking, for any $v \in \gv(r_0 - \delta)$, any $V \subseteq \mathcal{V}_{\mathrm m} \cup \mathcal{V}_{\mathrm a}$ with $v \not \in V$, and any $\Lambda \colon V \to \{\mathsf{F}, \mathsf{T}\}$, we have
\begin{equation*}
    \max \left\{ \pr_{\mu_{\Omega^\Lambda}}\left( v \mapsto \mathsf{F} \right), \pr_{\mu_{\Omega^\Lambda}}\left( v \mapsto \mathsf{T}\right)  \right\} \le \frac{1}{2} \exp\left(\frac{1}{k 2^{r_0 k}}\right).
\end{equation*}
}
\begin{proposition} \label{prop:marginals}
  \statepropmarginals  
\end{proposition}
\begin{proof} \label{prop:marginals:proof}
  This follows directly by combining Lemmas~\ref{lem:marking} and~\ref{lem:marginals}, which are stated and proved in Section~\ref{sec:marking}. 
\end{proof}

We note that the density threshold of Theorem~\ref{thm:sampling} is $2^{0.039 k}$, which is significantly smaller than the threshold $\densitydef$ in Proposition~\ref{prop:marginals}. The bottleneck for the threshold Theorem~\ref{thm:sampling} comes from our mixing time results, see Section~\ref{sec:po:si}.

The bound given in Proposition~\ref{prop:marginals} on the marginal probabilities of the marked and auxiliary variables is exploited several times in this work, and we will explain some of these applications in this proof outline. We remark that the bound on the marginals of good variables holds for \emph{any} pinning of \emph{any} subset of marked and auxiliary variables,  which will be relevant in the spectral independence argument.

 \begin{definition}[$\restr{\mu}{V}$] \label{def:marginal}
    Let $\mathcal{V}$ be a finite set and let $\Omega \subseteq \{\mathsf{F}, \mathsf{T}\}^{\mathcal{V}}$. Let $\mu$ be a distribution over $\Omega$. For a set $V \subseteq \mathcal{V}$, we denote by $\restr{\mu}{V}$ the marginal distribution of $\mu$ on $V$.
\end{definition}

Proposition~\ref{prop:marginals} implies that the distribution $\restr{\mu_{\Omega}}{\mathcal{V}_{\mathrm m} \cup \mathcal{V}_{\mathrm a}}$ is very close to the uniform distribution over all assignments $\mathcal{V}_{\mathrm m} \cup \mathcal{V}_{\mathrm a}  \to \{\mathsf{F}, \mathsf{T}\}$. This concept is formalised in the following definition.

\begin{definition}[$\varepsilon$-uniform] \label{def:uniform}
  Let $V$ be a set of variables and $\mu$ be a probability distribution over the assignments $V \to \{\mathsf{F}, \mathsf{T}\}$. Let $\Lambda \colon S \to \{\mathsf{F}, \mathsf{T}\}$ be an assignment of some subset of variables $S \subseteq V$. We denote by $\pr_\mu(\Lambda)$ the probability under $\mu$ of the event that the variables in $S$ are assigned values according to $\Lambda$, and by $\pr_\mu(\cdot \vert \Lambda)$ the corresponding conditional distribution of $\mu$.
  
  For $\varepsilon \in (0,1)$, we say that the distribution $\mu$ is \emph{$\varepsilon$-uniform} if for any variable $v \in V$ and any partial assignment $\Lambda \colon V \setminus \{v\} \to \{\mathsf{F}, \mathsf{T}\}$, we have
  \begin{equation*}
    \max \left\{ \pr_{\mu}\left( \left. v \mapsto \mathsf{F} \right| \Lambda\right), \pr_{\mu}\left( \left. v \mapsto \mathsf{T} \right| \Lambda \right)  \right\} \le \frac{1}{2}e^{\varepsilon}.
  \end{equation*}
\end{definition}

From Proposition~\ref{prop:marginals}, it follows that the distribution $\restr{\mu_{\Omega}}{\mathcal{V}_{\mathrm m}}$  is
$\epsilon$-uniform for $\epsilon = (2^{-r_0 k}/k)$, so 
for any $\Lambda \colon \mathcal{V}_{\mathrm m}  \to \{\mathsf{F}, \mathsf{T}\}$,
the 
probability that the assignment of the marked variables is~$\Lambda$ is at least $(1-e^\varepsilon/2)^{\lvert \mathcal{V}_{\mathrm m} \rvert}$. The $\epsilon$-uniform property also (trivially) guarantees that the space of assignments $\Lambda \colon \mathcal{V}_{\mathrm m}  \to \{\mathsf{F}, \mathsf{T}\}$  with $\pr_{\mu_\Omega}(\Lambda) > 0$ is connected via single-variable updates, so we can indeed consider the Glauber dynamics over $\mathcal{V}_{\mathrm m}$. This leads to the main challenge of this work: does this chain mix rapidly?
 
\subsection{Mixing time of the Glauber dynamics on the marked variables} \label{sec:po:si}
 
Recently, there has been significant progress in showing that the single-variable Glauber dynamics on appropriately chosen subsets of variables mixes quickly for $k$-CNF formulae with bounded degree~\cite{feng2020, jain2021sampling}. These approaches carefully execute a union bound over paths of clauses connecting marked variables in order to bound the coupling time between two copies of the chain. 
However, these union bound arguments break under the presence of high-degree variables that are present in random $k$-SAT; this is because the number of paths connecting marked variables is very sensitive to the max degree of the formula and in particular grows too fast in our setting. We give a more detailed discussion in Section~\ref{sec:po:si:pw}. 

Instead, we apply the spectral independence framework to show rapid mixing of a uniform-block Glauber dynamics, which we review briefly below. Applications of spectral independence usually exploit decay of correlations to show that the spectral independence condition holds, see~\cite{alg20,Chen2020SI, bez2021} for examples. As we have mentioned in the introduction, correlation decay fails to hold for densities exponential in $k$ in the random $k$-SAT model~\cite{montanari2007} and therefore, we have to develop a different approach to conclude that the spectral-independence condition holds in our setting. This is our main contribution in this work; we show that the marginal distribution on the marked variables, i.e., $\restr{\mu_\Omega}{\mathcal{V}_{\mathrm m}}$, is $(\epsilon \log n)$-spectrally independent for some $\epsilon>0$ that can be made arbitrarily small for sufficiently  large $k$. Our argument builds on the coupling idea of Moitra~\cite{moitra19} (as  refined in~\cite{galanis2019counting} for random $k$-SAT) and relates the spectral independence condition to the expected number of failed clauses in this coupling process. This allows us to exploit the local sparsity properties of the random $k$-SAT model to analyse the mixing time of the Glauber dynamics. 

A caveat here is that the spectral independence of $\restr{\mu_\Omega}{\mathcal{V}_{\mathrm m}}$ is not enough on its own to conclude fast mixing of the single-site Glauber dynamics. The most direct way to  work around  this is to analyse instead the so-called $\rho$-uniform-block Glauber dynamics that updates $\rho$ vertices at a time for some $\rho$ that scales linearly in $n$; the main missing ingredient there is to show that the modified chain can be implemented efficiently which we discuss in Section~\ref{sec:po:structure}. We next give a quick overview of the relevant ingredients of the spectral-independence literature that we will need.

 \subsubsection{The $\rho$-uniform-block Glauber dynamics, spectral independence, and the mixing time} \label{sec:po:si:block}
 
 Let $V$ be a finite set of size $M$ and $\mu$ be a distribution over the assignments $V \to \{\mathsf{F}, \mathsf{T}\}$.  Let $\Omega$ be the set of assignments $V \to \{\mathsf{F}, \mathsf{T}\}$ with positive probability under $\mu$. For an integer $\rho \in \{1,2,\ldots,\lvert V \rvert \}$, the $\rho$-uniform-block Glauber dynamics for $\mu$ is a Markov chain $X_t$ where $X_0 \in \Omega$ is an arbitrary configuration and, for $t \ge 1$, $X_t$ is obtained from $X_{t-1}$ by first picking a subset $S \subseteq V$ of size $\rho$ uniformly at random, letting $\Lambda_t$ be the restriction of $X_t$ to $V\setminus S$, and updating the configuration on $S$ according to the probability distribution $\mu ( \cdot \vert \Lambda_{t})$. This chain satisfies the detailed balance equation for $\mu$. Hence, when the chain is irreducible, for $\varepsilon > 0$, we can consider its mixing time $T_{\mix}(\rho, \varepsilon) = \max_{\sigma \in \Omega} \min \{t :  d_{\TV}(X_t, \mu) \le \varepsilon \mid X_0 = \sigma\}$. We say that $\mu$ is \emph{$b$-marginally bounded} if for all $v \in V$, $S \subseteq V \setminus \{v\}$, $\Lambda \colon S \to \{\mathsf{F}, \mathsf{T}\}$ with $\pr_{\mu}(\Lambda) > 0$, and $\omega \in \{\mathsf{F}, \mathsf{T}\}$, it either holds that $\pr_\mu(v \mapsto \omega \vert \Lambda ) = 0$ or $\pr_\mu(v \mapsto \omega \vert \Lambda) \ge b$. Spectral independence  results have recently been used in the $b$-marginally bounded setting to obtain fast mixing time of the uniform-block Glauber dynamics~\cite{Blanca2022, Chen2021}. For $S \subset V$, $\Lambda \colon S \to \{\mathsf{F}, \mathsf{T}\}$ with $\pr_{\mu}(\Lambda) > 0$, and $u,v \in V$ with $u,v \not \in S$ and $0 < \pr_{\mu}(u \mapsto \mathsf{T} \vert \Lambda) < 1$, the \emph{influence} of $u$ on $v$ (under $\mu$ and $\Lambda$) is defined as 
\begin{equation} \label{eq:influence}
  \mathcal{I}^{\Lambda}(u \to v) = \pr_\mu\left( v \mapsto \mathsf{T} \vert u \mapsto \mathsf{T}, \Lambda\right) - \pr_\mu\left( v \mapsto \mathsf{T}  \vert u \mapsto \mathsf{F}, \Lambda\right).
\end{equation}
The \emph{influence matrix conditioned on $\Lambda$} is the (two-dimensional) matrix whose entries consist of  ${\mathcal{I}}^\Lambda(u \to v)$ over all relevant $u$ and $v$. We denote by ${\mathcal{I}}^\Lambda$ the matrix and by $\lambda_1({\mathcal{I}}^\Lambda)$ the largest absolute value of its eigenvalues. For a real $\eta>0$, we say that $\mu$ is \emph{$\eta$-spectrally independent} if for all $S \subset V$ and $\Lambda \colon S \to \{\mathsf{F}, \mathsf{T}\}$ with $\pr_{\mu}(\Lambda) > 0$ we have $\lambda_1({\mathcal{I}}^\Lambda) \le \eta$.  From the results of~\cite{Chen2021}, one can conclude the following bound for the mixing time of the uniform-block Glauber dynamics, see Appendix~\ref{sec:ap:mt} for details.

 \begin{lemma} \label{lem:block-glauber}
  The following holds for any reals $b, \eta > 0$,  any $\kappa \in (0,1)$ and any integer $M$ with $M \ge \frac{2}{\kappa} (4\eta / b^2 + 1)$. Let $V$ be a set of size $M$, let $\mu$ be a distribution over the assignments $V \to \{\mathsf{F}, \mathsf{T}\}$, let $\Omega = \{\Lambda \colon V \to \{\mathsf{F}, \mathsf{T}\} :  \mu(\Lambda) > 0\}$  and let $\mu_{\min} = \min_{\Lambda \in \Omega} \mu(\Lambda)$. If $\mu$ is $b$-marginally bounded and $\eta$-spectrally independent, then, for $\rho = \lceil \kappa M \rceil$ and $C_\rho = (2/\kappa)^{4 \eta / b^2 + 1}$, we have
    \begin{equation*}
      T_{\mix}(\rho, \varepsilon) \le \left\lceil C_\rho \frac{M}{\rho} \left( \log \log \frac{1}{\mu_{\min}} + \log \frac{1}{2 \varepsilon^2} \right) \right\rceil.
  \end{equation*}
 \end{lemma}
 We are going to consider the uniform-block Glauber dynamics on the marked variables of $\Phi$, so $V = \mathcal{V}_{\mathrm m}$, and the set of states coincides with the set of assignments $\mathcal{V}_{\mathrm m} \to \{\mathsf{F}, \mathsf{T} \}$ as all of them have positive probability. In this setting, the target distribution is $\restr{\mu_\Omega}{\mathcal{V}_{\mathrm m}}$. The distribution $\restr{\mu_\Omega}{\mathcal{V}_{\mathrm m}}$ is $(1/e)$-marginally-bounded as a straightforward consequence of the fact that it is $(1/k)$-uniform, see Remark~\ref{rem:marginally-bounded} for details. Hence, in order to conclude rapid mixing it remains to establish spectral independence. For this, we are going to use the well-known fact (see for instance~\cite{Chen2020SI}) that, for $S \subset V$ and $\Lambda \colon S \to \{\mathsf{F}, \mathsf{T}\}$, we have 
\begin{equation} \label{eq:si}
    \lambda_1({\mathcal{I}}^\Lambda) \le \max_{u \in V \setminus S} \sum_{v \in V \setminus S} \lvert \mathcal{I}^{\Lambda}(u \to v) \rvert.
\end{equation}

\subsubsection{Spectral independence in the random $k$-SAT model} \label{sec:po:si:results}

In this section we state our spectral independence results in the random $k$-SAT model. The results stated in this section are proved in Section~\ref{sec:mixing-time}. Our main technical result is the following.

\newcommand{\statelemsi}{
There is an integer $k_0 \ge 3$ such that  for any integer $k \ge k_0$ and any density $\alpha$ with $\density \le 2^{ r_0  k / 3} / k^3$ the following holds. W.h.p. over the choice of the random $k$-CNF formula $\Phi = \Phi(k, n, \lfloor \alpha n \rfloor)$, for any $(r_0 - \delta, r_0, r_0, 2r_0)$-marking $(\mathcal{V}_{\mathrm m}, \mathcal{V}_{\mathrm a}, \mathcal{V}_{\mathrm c})$ of $\Phi$, the  distribution $\restr{\mu_{\Omega}}{\mathcal{V}_{\mathrm m}}$ is $(2^{-(r_0-\delta) k}\log n)$-spectrally independent.}
\begin{lemma} \label{lem:si}
	\statelemsi
\end{lemma}

We are going to describe some of the ideas behind the proof of Lemma~\ref{lem:si}. First, we highlight the fact that, due to the presence of high-degree variables (which form logarithmic-sized connected components), current techniques seem unable to conclude $\eta$-spectral independence with $\eta = O(1)$. This has also been the case in recent work on $2$-spin systems on random graphs~\cite{bez2021}, where instead  correlation decay is exploited to prove $\eta$-spectral independence for some $\eta=o(\log n)$. Here, our $\eta$-spectral independence bound for  $\eta = o_k(\log n)$ will be  based on an appropriate coupling. Note, in light of Lemma~\ref{lem:block-glauber}, $\eta = O(\log n)$ is good enough for proving polynomial mixing time of the uniform-block Glauber dynamics, but we need the improved bound of Lemma~\ref{lem:si} in order to conclude the following fast mixing-time result from  Lemma~\ref{lem:block-glauber} (as illustrated Section~\ref{sec:mixing-time}).
 
 \newcommand{\statelemmt}{
 There is a function $k_0(\theta) = \Theta(\log(1/\theta))$ such that, for any $\theta \in (0,1)$, for any integer $k \ge k_0(\theta)$ and any density $\alpha$ with $\density \le 2^{0.039  k}$ the following holds. W.h.p. over the choice of the random $k$-CNF formula $\Phi = \Phi(k, n, \lfloor \alpha n \rfloor)$, for any $(r_0-\delta, r_0, r_0, 2r_0)$-marking $(\mathcal{V}_{\mathrm m}, \mathcal{V}_{\mathrm a}, \mathcal{V}_{\mathrm c})$ of $\Phi$ and for $\rho = \lceil 2^{-k-1} \lvert \mathcal{V}_{\mathrm m} \rvert \rceil$, the $\rho$-uniform-block Glauber dynamics for updating the marked variables has mixing time $T_{\mix}(\rho, \varepsilon/2) \le T := \mtdef$.}
    
\begin{lemma} \label{lem:mixing-time}
  \statelemmt
\end{lemma}

Lemma~\ref{lem:mixing-time} is stated for the block size $\rho = \lceil 2^{-k-1} \lvert \mathcal{V}_{\mathrm m} \rvert \rceil$, but it could be proved more generally when $\rho = c \lvert \mathcal{V}_{\mathrm m} \rvert$ and $c \in (0,1)$. The fact that $\rho \le \lvert \mathcal{V}_{\mathrm m} \rvert / 2^k$ in the statement will be relevant in implementing efficiently the dynamics, discussed in  Section~\ref{sec:po:structure}. 

We remark that the more restrictive density threshold $\density \le 2^{ r_0  k / 3} / k^3$ in the statement of Lemma~\ref{lem:si} arises in the union bound given in the proof of this lemma, and that for large enough $k$ we have $2^{0.039 k} \le 2^{ r_0  k / 3} / k^3$, the former being the density threshold given in Lemma~\ref{lem:mixing-time} and Theorem~\ref{thm:sampling}.

Our approach to prove $\eta$-spectral independence significantly differs from those that in two-spin systems, where it is enough to study sum of influences over trees (thanks to the tree of self-avoiding walks) and exploit decay of correlations in this setting (very roughly, the further away two vertices are in the tree, the smaller the influence that one vertex has in the other). Here we relate influences to the structure of the dependency graph $G_\Phi$ by running a coupling process on the auxiliary variables, and we state this connection in the upcoming  Lemma~\ref{lem:expectation}. First we define more formally the dependency graph $G_\Phi$.

\begin{definition}[$G_\Phi$] \label{def:graph-phi} 
Let $\Phi = (\mathcal{V}, \mathcal{C})$ be a $k$-CNF formula. We define the graph $G_\Phi$ as follows. The vertex set of $G_\Phi$  is $\mathcal{C}$ and two clauses $c_1$ and $c_2$ are adjacent if and only if $\var(c_1) \cap \var(c_2) \ne \emptyset$. A set $C \subseteq \mathcal{C}$ is connected if $C$ is connected in the graph $G_\Phi$. We say that two variables $u$ and $v$ are connected in $\Phi$ if there is a path $c_1, c_2, \ldots, c_\ell$ in $G_{\Phi }$ with $u \in \var(c_1)$ and $v \in \var(c_\ell)$.
\end{definition}

Let $u \in \mathcal{V}_{\mathrm m}$, $S \subset \mathcal{V}_{\mathrm m}$ and $\Lambda \colon S \to \{\mathsf{F}, \mathsf{T}\}$. The aim of the coupling process is bounding the sum $\sum_{v \in \mathcal{V}_{\mathrm m} \setminus (S \cup \{u\})}\lvert \mathcal{I}^{\Lambda}(u \to v) \rvert$ in terms of the expected size of a connected set of \emph{failed} clauses, where the expectation is over the choices made in the coupling process. We refer to  Section~\ref{sec:mixing-time} for a definition of failed clauses, as it is not relevant in this discussion. Here we give a brief overview of how the coupling process on the auxiliary variables works. First, we start with two assignments $X = \Lambda \cup (u \mapsto \mathsf{T})$ and $Y = \Lambda \cup (u \mapsto \mathsf{F})$, where $\Lambda \cup (u \mapsto \omega)$ denotes the assignment defined on $S \cup \{u\}$ that agrees with $\Lambda$ on $S$ and sends $u$ to $\omega$. The process progressively extends $X$ and $Y$ on some auxiliary variables $v_1, v_2, \ldots$ following the optimal coupling between the marginals $\pr_{\mu_{\Omega}}(v \mapsto \cdot \vert X)$ and $\pr_{\mu_{\Omega}}(v \mapsto \cdot \vert Y)$, see Section~\ref{sec:mixing-time} for the definition of optimal coupling. The main property of this process is that with high probability over the choices made, at some point the graphs $G_{\Phi^X}$ and $G_{\Phi^Y}$ factorise in small  connected components in spite of the presence of bad variables and, on top of that, $\Phi^X$ and $\Phi^Y$ share most of these connected components. Then we can bound influences between marked variables by analysing the connected components where $\Phi^X$ and $\Phi^Y$ differ, which turn out to be $\operatorname{poly}(k) \log n$ in size after enough steps of the process.

One of the key ideas behind our analysis is exploiting the fact that, in the random $k$-SAT model, w.h.p. over the choice of the random formula $\Phi$, any logarithmic-sized set of clauses $Z$ that is connected in $G_\Phi$ has constant tree-excess, that is, the number of edges connecting a pair of clauses in $Z$ is $\lvert Z \rvert + O(1)$. This saves a factor of~$\Delta_{r_0-\delta}$ in the spectral independence bound by ensuring that there is a large independent set of clauses in the set of failed clauses. We also obtain improved analysis by restricting the coupling process to auxiliary variables. This enables us to get exponentially small bounds (in~$k$) on the influences between marked variables, which leads to our $(2^{-(r_0-\delta) k}\log n)$-spectral independence result.

\subsection{Analysis of the connected components of $\Phi^\Lambda$. Applications to connectivity and looseness} \label{sec:po:structure}

In this section we deal with the third challenge mentioned at the beginning of Section~\ref{sec:po}: can we determine the transition probabilities of the Glauber dynamics so that we can actually simulate this Markov chain? In fact, simulating the single-site Glauber dynamics on the marked variables was one of the main challenges even in the bounded-degree case. In that case this was resolved using a method that is restricted to the bounded-degree setting (and whose bottleneck is the analysis of a rejection sampling procedure). A different procedure is required for the random $k$-SAT setting.

One of the key ideas to simulate this chain is starting the chain on an assignment $X_0 \colon \mathcal{V}_{\mathrm m} \to \{\mathsf{F}, \mathsf{T}\}$ drawn from the uniform distribution over all assignments of $\mathcal{V}_{\mathrm m}$. Since the distribution $\restr{\mu_\Omega}{\mathcal{V}_{\mathrm m}}$ is $(1/k)$-uniform (Proposition~\ref{prop:marginals}), the transition probabilities of the Glauber dynamics are close to uniform. This allows us to show that the probability distribution of the assignment $X_t$ that is output by the uniform-block Glauber dynamics after $t$ steps is also $(1/k)$-uniform (Corollary~\ref{cor:marginals:glauber}), 
which will be important in what follows.

In order to run the $\rho$-uniform-block Glauber dynamics we need to be able to sample from the distribution $\mu_{\Omega^{\Lambda}}$ for any set $S \subseteq \mathcal{V}_{\mathrm m}$ with $\lvert S \rvert = \rho$ and any assignment $\Lambda \colon \mathcal{V}_{\mathrm m} \setminus S \to \{\mathsf{F}, \mathsf{T} \}$ that arises. 
Unless we can restrict $\Lambda$, sampling from $\mu_{\Omega^{\Lambda}}$ could potentially be as hard as sampling from $\mu_{\Omega}$. Fortunately for us, the assignment $\Lambda$ is not completely arbitrary; $\Lambda$ is determined by the random choice of $S$ and the current state of the Glauber dynamics (which follows a $(1/k)$-uniform distribution as discussed above). We show that we can efficiently sample from $\mu_{\Omega^{\Lambda}}$ w.h.p.{}  over the choice of $\Lambda$. An important observation is that we can efficiently sample from $\mu_{\Omega^{\Lambda}}$ when the connected components of $G_{\Phi^\Lambda}$ are logarithmic in size, for example, by applying brute force. This raises the following question: does $G_{\Phi^\Lambda}$  break into small connected components  w.h.p.{} over the choice of $\Lambda$? Lemma~\ref{lem:cc-general} gives a positive answer when $0 \le \rho \le \lvert V\rvert / 2^k$. Here the reader can see $V$ as the set of marked variables. The proof of Lemma~\ref{lem:cc-general} exploits sparsity properties of logarithmic-sized connected sets of clauses in random formulae in conjunction with the fact that $\mu$ is $(1/k)$-uniform.  Lemma~\ref{lem:cc-general} is stated with an added layer of generality, as we will also apply it to analyse the geometry of the space of satisfying assignments of $\Phi$ with $r = r_1 - \delta$. In our sampling algorithm setting we consider $r = r_0 - \delta$. Recall that $r_0 = \rvalue$, $r_1 = 0.227092$ and $\delta = \deltadef$. The restriction $r \in (2\delta,1/(2 \log 2)]$ in the statement of Lemma~\ref{lem:cc-general} is not optimal, but it is enough for our purposes. 

\newcommand{\statelemccgeneral}{ Let $r \in (2\delta,1/(2 \log 2)]$. There is an integer $k_0 \ge 3$ such that, for any integer $k \ge k_0$, any density $\density \le 2^{(r-2\delta) k}$, and any real number $b$ with $a:= 2k^4 < b$, the following holds w.h.p. over the choice of $\Phi = \Phi(k, n, \lfloor \density n\rfloor)$. 

Let $L$ be an integer satisfying $a \log n \le L \le b \log n$. Let $V$ be a set of good variables of $\Phi$ that is $(r+\delta)$-distributed (\cref{def:distributed-marking}), let $\mu$ be a $(1/k)$-uniform distribution over the assignments $V \to \{\mathsf{F}, \mathsf{T}\}$, and let $\rho$ be an integer with $0 \le \rho \le \lvert V\rvert / 2^k$.  Consider the following  experiment. First, draw $S \subseteq V$ from the uniform distribution $\tau$ over subsets of $V$ with size $\rho$. Then,  sample an assignment $\Lambda$ from 
$\restr{\mu}{V\setminus S}$. 
Denote by $\mathcal{F}$ the event that that there is a connected set of clauses $Y$ of $\Phi$ with $\lvert Y \rvert \ge  L$  such that all clauses in $Y$ are unsatisfied by $\Lambda$. Then  $
    \pr_{S \sim \tau} \left( \pr_{\Lambda \sim \restr{\mu}{V \setminus S}} \left( \mathcal{F} \right) \le 2^{- \delta k L}  \right) \ge 1 - 2^{-\delta k L}$.
 }  
\begin{lemma} \label{lem:cc-general}
 \statelemccgeneral
\end{lemma}
\begin{proof}[Proof sketch]
  The proof is in Section~\ref{sec:errors}. For the sake of exposition, we first sketch the proof in the case $\rho = 0$, where the conclusion in the statement reads $\pr_{\Lambda \sim \restr{\mu}{V}} \left( \mathcal{F} \right) \le 2^{- \delta k L}$. At the end of this proof sketch we explain how we extend the proof to any $\rho$ with $0 \le \rho \le \lvert V \rvert / 2^k$. 
  
  The first step   is exploiting   local sparsity properties of random $k$-CNF formulae to find  many variables from~$V$ in any sufficiently large connected set of clauses. Our sparsity results hold for  connected sets of clauses with size at least $2k^4 \log n$, and let us conclude the following result (stated as Lemma~\ref{lem:bound-marked} in Section~\ref{sec:errors}): w.h.p. over the choice of $\Phi$, for every connected set of clauses $Z \subseteq \mathcal{C}$ we have
 \begin{equation} \label{eq:num-marked:intro}
     \text{if } \ 2k^4 \log(n) \le \lvert Z \rvert \le b \log(n), \ \text{ then } \  \lvert \var(Z) \cap V \rvert \ge r k \lvert Z \rvert.
 \end{equation}
  The proof of Lemma~\ref{lem:bound-marked} counts the variables from~$V$ in~$Z$ by using the fact that $Z$ does not contain many bad clauses (Lemma~\ref{lem:bad}, which gives the restriction on $r$) and the fact that there are not many edges joining clauses in $Z$. In fact, for such a set $Z$, we show that the number of edges is of order $\lvert Z \rvert + O(1)$, that is, $Z$ has constant tree-excess (Lemma~\ref{lem:tree-excess}). We also need the following result on random $k$-CNF formulae. For each clause $c \in \mathcal{C}$, let $\mathcal{Z}(c, L) = \{Z \subseteq \mathcal{C} : c \in Z, Z \text{ is connected in } G_\Phi, \lvert Z \rvert = L\}$.
 Then, w.h.p. over the choice of $\Phi$,~\cite[Lemma 40]{galanis2019counting} shows that, as long as $L \ge \log n$,
 \begin{equation} \label{eq:number-Y:intro}
     \text{for any clause }c\in\mathcal{C} \text{ we have } \left| \mathcal{Z}(c, L)\right| \le (9 k^2 \alpha)^L.
 \end{equation}
 Once we have established \eqref{eq:num-marked:intro} and \eqref{eq:number-Y:intro}, the proof exploits the fact that $\mu$ is close to the uniform distribution. First, we introduce some notation. Let $L$ be an integer with $a \log n \le L \le b \log n$. Let $S = \emptyset$ as we are dealing with the case $\rho = 0$.  For $c \in \mathcal{C}$ and $Z \in \mathcal{Z}(c, L)$, we denote by $\mathcal{E}_1(Z, S)$ the event that none of the clauses of $Z$ are satisfied by assignment $\Lambda$ (Definition~\ref{def:subformula}), where $\Lambda$ is drawn from $\restr{\mu}{V \setminus S}$, see Definition~\ref{def:marginal}.  We keep track of $S$ in the notation here as this is relevant in the general case. The first observation   is that the event $\mathcal{F}$ from the statement satisfies $\mathcal{F} = \bigcup_{c \in \mathcal{C}, Z \in \mathcal{Z}(c, L)} \mathcal{E}_1(Z, S)$. We then claim that for any $c \in \mathcal{C}$ and $Z \in \mathcal{Z}(c, L)$ we have
  \begin{equation} \label{eq:E1:intro}
     \pr_{\Lambda \sim \restr{\mu}{V \setminus S}} \left( \mathcal{E}_1(Z, S) \right) 
      \le  \frac{2^{-\delta k L}}{\lvert \mathcal{C} \rvert \cdot \lvert \mathcal{Z}(c, L) \rvert },
  \end{equation}
  so the result would follow from a union bound over $c$ and $Z$. Let us give some insight on how we prove \eqref{eq:E1:intro}. Let $c \in \mathcal{C}$ and $Z \in \mathcal{Z}(c, L)$. The main idea  is   that, if all clauses in $Z$ are unsatisfied by $\Lambda$ then, when we sampled $\Lambda \sim \restr{\mu}{V \setminus S}$, for each variable $v$ in $\var(Z)\cap(V\setminus S)$ we picked the value that does not satisfy the clauses of $Z$ containing $v$. Thus, we can bound the probability that all clauses in $Z$ are unsatisfied as a product, over the variables in $\var(Z)\cap(V\setminus S)$, of probabilities, each factor corresponding to the probability that a variable is assigned a certain value (under some careful conditioning, see the proof in Section~\ref{sec:errors} for details). Since the distribution $\mu$ is $(1/k)$-uniform, each one of these factors can be bounded by $\exp(1/k)/2$, obtaining
  \begin{equation} \label{eq:E1:step}
     \pr_{\Lambda \sim \restr{\mu}{V \setminus S}} \left(  \mathcal{E}_1(Z, S)\right)  
      \le  \left( \frac{1}{2} \exp\left(\frac{1}{k}\right) \right)^{\lvert \var(Z) \cap (V \setminus S) \rvert}.
  \end{equation}
  In \eqref{eq:num-marked:intro} we gave a lower bound on $\lvert \var(Z) \cap V \rvert$, which can be applied in conjunction with \eqref{eq:number-Y:intro} to conclude, after some calculations, that the bound given in \eqref{eq:E1:intro} holds. 
  
  The case $\rho > 0$ is more technical and one has to be more careful in these calculations. We show that \eqref{eq:E1:intro} holds when $S$ does not contain many variables in $\var(Z) \cap V$. A slightly different argument is needed when going from \eqref{eq:E1:step} to  \eqref{eq:E1:intro}; here we have to bound $\lvert \var(Z) \cap (V\setminus S) \rvert$ instead of $\lvert \var(Z) \cap V \rvert$. It turns out that, as long as the bound $\lvert \var(Z) \cap V \cap S \rvert \le \lvert \var(Z) \cap V \rvert / k$ holds, the calculations to go from \eqref{eq:E1:step} to \eqref{eq:E1:intro} also hold in this setting. Finally, we show that the probability that $\lvert \var(Z) \cap V \cap S \rvert \le \lvert \var(Z) \cap V \rvert / k$ occurs when picking $S$ is at least  $1-2^{\delta k L}$. The proof of this fact is purely combinatorial, and requires the hypothesis $\rho \le \lvert V \rvert / 2^ k$, see Section~\ref{sec:errors} for details.
\end{proof}

Once we have established 
Lemma~\ref{lem:cc-general}, we can use it to implement the $\rho$-uniform-block Glauber dynamics on the marked variables for $0 < \rho \le \lvert \mathcal{V}_{\mathrm m} \rvert$ and complete our sampling algorithm, which we explicitly state in Section~\ref{sec:po:alg}.

Before concluding this section, we mention how we apply Lemma~\ref{lem:cc-general} to analyse the geometry of the space of satisfying assignments of $\Phi$ in order to conclude the  $O(\log n)$-connectivity and $O(\log n)$-looseness results given in Theorems~\ref{thm:connectivity} and~\ref{thm:looseness}. First, we need the following definition.

\begin{definition}[$H_\Phi$] \label{def:graph-H}
    Let $\Phi = (\mathcal{V}, \mathcal{C})$ be a $k$-CNF formula. We define the graph $H_\Phi$ as follows. The vertex set of $H_\Phi$ is $\mathcal{V}$ and two variables $v_1$ and $v_2$ are adjacent in $H_\Phi$ if there is a clause $c \in \mathcal{C}$ with $v_1, v_2 \in \var(c)$.
\end{definition}

We apply Lemma~\ref{lem:cc-general} with  $r = r_1-\delta$ and a density $\alpha \le 2^{(r_1 - 3 \delta)k} / k^3$. For an $(r, r_1, 0, r_1)$-marking $(\mathcal{V}_{\mathrm m}, \emptyset, \mathcal{V}_{\mathrm c})$ of $\Phi$, we let $V = \mathcal{V}_{\mathrm m}$ and $\mu = \restr{\mu_\Omega}{\mathcal{V}_{\mathrm m}}$. In this setting, for $\rho = 0$, Lemma~\ref{lem:cc-general} allows us to conclude that, w.h.p. over the choice of $\Lambda \sim \restr{\mu_\Omega}{\mathcal{V}_{\mathrm m}}$, the graph $G_{\Phi^{\Lambda}}$ consists of connected components with size at most $O( \log n)$. Thus, the connected components of $H_{\Phi^{\Lambda}}$ have size at most $O( \log n)$ as each clause contains at most $k$ variables. This leads to the main idea behind the proof of Theorem~\ref{thm:connectivity}: we can construct $O(\log n)$-paths between satisfying assignments by progressively updating the variables in each one of the connected components of $H_{\Phi^{\Lambda}}$. As an example, let $\mathcal{E}_1, \mathcal{E}_2, \ldots, \mathcal{E}_t$ be these connected components and let $\sigma_1$ and $\sigma_2$ be two satisfying assignments that agree with $\Lambda$ on $\mathcal{V}_{\mathrm m}$. Then we can find an $O(\log n)$-path $\sigma_1 = \zeta_0 \leftrightarrow \zeta_1 \leftrightarrow \cdots \leftrightarrow \zeta_{t} = \sigma_2$ as follows: the assignment $\zeta_j$ is the satisfying assignment that agrees with $\Lambda$, agrees with $\sigma_1$ on the variables in $\mathcal{V} \setminus \left( \bigcup_{i = 1}^j \mathcal{E}_{j} \right)$ and agrees with $\sigma_2$ on the variables in $\bigcup_{i = 1}^j \mathcal{E}_{j}$. The case when $\sigma_1$ and $\sigma_2$ differ on some marked variables builds on the same idea though it is more technical and requires applying Lemma~\ref{lem:cc-general} with $\rho=1$. We refer to Section~\ref{sec:connectivity:proof} for this argument and the proof of Theorem~\ref{thm:connectivity}.

The fact that the connected components of $H_{\Phi^{\Lambda}}$ are $O(\log n)$ in size with high probability over $\Lambda \sim \restr{\mu_\Omega}{\mathcal{V}_{\mathrm m}}$ is also related to the looseness of the formula $\Phi$. Let $v \in \mathcal{V} \setminus \mathcal{V}_{\mathrm m}$. For any satisfying assignment $\sigma$ that agrees with $\Lambda$ on the marked variables, we can construct a satisfying assignment $\tau$ with $\tau(v) \ne \sigma(v)$ and $\| \sigma - \tau \|_1  = O(\log n)$ by updating the variables in the connected component of $v$ in $H_{\Phi^\Lambda}$, provided that there is a way to satisfy this connected component when giving $v$ the value $\tau(v)$. In Section~\ref{sec:looseness} we formalise this idea and give all the details of this argument to prove Theorem~\ref{thm:looseness}.

\subsection{The sampling algorithm} \label{sec:po:alg}

To complete this proof outline, we explicitly describe 
Algorithm~\ref{alg:main}, 
our algorithm for sampling satisfying assignments of $k$-CNF formulae. The algorithm uses a method  $\operatorname{Sample}(\Phi^\Lambda, S)$ to sample an assignment $\tau \colon S \to \{\mathsf{F}, \mathsf{T}\}$ from the distribution $\restr{\mu_{\Omega^\Lambda}}{S}$. This method exploits the fact that logarithmic-sized connected set of clauses have constant tree-excess, which does not hold in the bounded-degree case. This tree-like property enables us to efficiently sample satisfying assignments on the connected components of $\Phi^\Lambda$ by a standard dynamic programming argument, see Section~\ref{sec:sample}. Lemma~\ref{lem:sample} is our main result on  $\operatorname{Sample}(\Phi^\Lambda, S)$.
\newcommand{\statelemsampleold}{
There is an integer $k_0 \ge 3$ such that, for any integers $k \ge k_0$
and $\xi \ge 1$ and any density $\density \le \densitymaindef$, the following holds w.h.p. over the choice of $\Phi = \Phi(k, n, \lfloor \density n\rfloor)$. In the execution of Algorithm~\ref{alg:main} with inputs $\Phi$ and $\varepsilon$, for any subset of variables $\mathcal{V}'$, any assignment $\Lambda \colon \mathcal{V}' \to \{\mathsf{F}, \mathsf{T}\}$, any connected component $H$ of $G_{\Phi^{\Lambda}}$ that has size at most $\sizecc$ and any variable $v$ appearing in $H$, it is possible to sample from the marginal distribution of $v$ induced by  $\mu_{\Omega^{\Lambda}}$  in time $O( \log n )$.}
\newcommand{\statelemsample}{
There is an integer $k_0 \ge 3$ such that, for any integers $k \ge k_0$, $b \ge 2k^4$ and any density $\density > 0$, the following holds w.h.p. over the choice of $\Phi = \Phi(k, n, \lfloor \density n\rfloor)$. Let $V$ be a subset of variables and let $\Lambda \colon V \to \{\mathsf{F}, \mathsf{T}\}$ be a partial assignment such that all the connected components in $G_{\Phi^{\Lambda}}$ have size at most $b \log(n)$. Then, there is an algorithm that, for any $S \subseteq \mathcal{V} \setminus V$, samples an assignment from $\restr{\mu_{\Omega^{\Lambda}}}{S}$  in time $O( \lvert S \rvert \log n )$.}
\begin{lemma} \label{lem:sample}
  \statelemsample
\end{lemma}

The method $\operatorname{Sample}(\Phi^\Lambda, S)$  
is used in Algorithm~\ref{alg:main} 
to implement each step of the $\rho$-uniform-block Glauber dynamics on the marked variables. It is also used to extend the assignment of marked variables computed by the Glauber dynamics to a satisfying assignment of $\Phi$. As a design choice, this method returns \textit{error} when the connected components of  $G_{\Phi^{\Lambda}}$ have size larger than $2k^4(1+\xi) \log(n)$. We remark that the probability that $\operatorname{Sample}(\Phi^\Lambda, S)$  returns \textit{error} is very small when running the Glauber dynamics thanks to Lemma~\ref{lem:cc-general}. We can now introduce Algorithm~\ref{alg:main}, which has two parameters $\theta \in (0,1)$ and $\xi \ge 1$ as in Theorem~\ref{thm:sampling}.

\begin{algorithm}[H]
  \begin{algorithmic}[1]

    \caption{The approximate sampling algorithm for satisfying assignments of random $k$-CNF formulae.} \label{alg:main}

    \REQUIRE A $k$-CNF formula $\Phi = (\mathcal{V},\mathcal{C})$ with $n$ variables 

    \STATE Compute the sets of bad/good variables and bad/good clauses for $\Phi$ as in Proposition~\ref{prop:bad}.

    \STATE Let $\varepsilon = n^{-\xi}$. Compute a marking $(\mathcal{V}_{\mathrm m}, \mathcal{V}_{\mathrm a}, \mathcal{V}_{\mathrm c})$ for $\Phi$ as in Lemma~\ref{lem:marking} with $p = \varepsilon/4$. This succeeds with probability at least $1 - \varepsilon/4$.  If this does not succeed, the algorithm returns \textit{error}.

    \STATE For each $v \in \mathcal{V}_{\mathrm m}$, sample $X_0(v) \in \{\mathsf{F}, \mathsf{T}\}$ uniformly at random.
    
    \FOR{$t$ from $1$ to $T := \mtdef$}

    \STATE Choose uniformly at random a set of marked variables $S\subseteq \mathcal{V}_{\mathrm m}$ with size $\rho := \lceil 2^{-k-1} \lvert \mathcal{V}_{\mathrm m} \rvert \rceil$.

    \STATE Let $\Lambda_t$ be the assignment $X_{t-1}$ restricted to $\mathcal{V}_{\mathrm m} \setminus S$.

    \STATE $Y \leftarrow \operatorname{Sample}(\Phi^{\Lambda_t}, S)$. \label{line:sample1}

    \STATE $X_t \leftarrow \Lambda_t \cup Y$.
    
    \ENDFOR

    \STATE $Y \leftarrow \operatorname{Sample}(\Phi^{X_T}, \mathcal{V}_{\mathrm a} \cup \mathcal{V}_{\mathrm c})$. \label{line:sample2}
    
    \RETURN $X_T \cup Y$.
  \end{algorithmic}
\end{algorithm}

We remark here that Algorithm~\ref{alg:main} only works for large enough $k$, and this hypothesis will be used several times in our arguments. The quantity $T$ defined in this algorithm corresponds to the mixing time of the $\rho$-uniform-block Glauber dynamics given in Lemma~\ref{lem:mixing-time}. 

\section{Paper outline}

The rest of this work is organised as follows. In Section~\ref{sec:bad} we introduce the procedure for determining bad clauses. In Section~\ref{sec:marking} we prove Proposition~\ref{prop:marginals} on markings of random formulae. In Section~\ref{sec:errors} we prove our technical result on the connected components of $\Phi^\Lambda$, Lemma~\ref{lem:cc-general}. In Section~\ref{sec:sample} we give the method $\operatorname{Sample}$ and prove Lemma~\ref{lem:sample}. In Section~\ref{sec:mixing-time} we prove the results on spectral independence stated in Section~\ref{sec:po:si} of the proof outline. In Section~\ref{sec:alg} we complete the proof of Theorem~\ref{thm:sampling} by combining our mixing time results (Lemma~\ref{lem:mixing-time}), our algorithm to sample from small connected components (Lemma~\ref{lem:sample}) and our result on the size of the connected components of $\Phi^{\Lambda}$ (Lemma~\ref{lem:cc-general}). Finally, in Section~\ref{sec:connectivity} we prove Theorems~\ref{thm:connectivity} and~\ref{thm:looseness} on the geometry of the space of satisfying assignments of $\Phi$.

To help keep track of the notation and definitions introduced in this work, the reader is referred to the tables in Appendix~\ref{ap:notation}.

\section{High-degree and bad variables in random CNF formulae} \label{sec:bad}

As we noted in the introduction, one of the keys to sampling satisfying assignments in the unbounded-degree setting is to   ``sacrifice''  a few variables per clause  (treating them separately in the sampling algorithm) and  to (temporarily) remove a small linear number of clauses that contain  these. The point of this is to ensure that the remaining (``good'') clauses have mostly low-degree variables (at most two bad ones) and also that the rest of the clauses (the ``bad'' ones) form small connected components that interact with the good clauses in a manageable way. 

Recall that, for $r \in (0,1)$, high-degree variables were introduced in Definition~\ref{def:degree} as those variables with at least $\Delta_r := \lceil 2^{k r} \rceil$ occurrences in the formula. In this work we consider two possible values for $r$ here, $r= r_0 - \delta$ and $r= r_1 - \delta$, where $r_0 = \rvalue$, $r_1 = \rone$ and $\delta = \deltadef$. The values $r_0$ and $r_1$ arise as solutions of an optimisation problem in Section~\ref{sec:marking} when we establish the markings that we use in our proofs. The marking used in our algorithmic results requires the more restrictive definition of high-degree variable with $r = r_0-\delta$ than the marking used in our connectivity results with $r = r_1 - \delta$. Subtracting $\delta$ will make our calculations easier without affecting our results.

By standard arguments about random graphs, one can determine that, w.h.p. over the choice of $\Phi$, the number of high-degree variables of $\Phi$ is bounded. We want to identify the clauses of $\Phi$ that have at most $2$ high-degree variables, since clauses with a lot of high-degree variables will interfere with our sampling algorithms. This motivates the following construction. The \emph{bad variables} and \emph{bad clauses} of $\Phi$ are identified by running the process given in Algorithm~\ref{alg:bad}. Here $\bv(r)$ denotes the set of bad variables and $\bc(r)$ denotes the set of bad clauses.

\begin{algorithm}[H]
  \begin{algorithmic}[1]

    \caption{Computing bad variables and bad clauses for $r \in (0,1)$} \label{alg:bad}

    \REQUIRE A $k$-CNF formula $\Phi = (\mathcal{V},\mathcal{C})$

    \STATE $\mathcal{V}_0(r) \leftarrow \text{ the set of high-degree variables, i.e., variables  with at least } \Delta_r = \lceil 2^{r k} \rceil \text{ occurrences in } \Phi$.

    \STATE $\mathcal{C}_0(r) \leftarrow \text{ the set of clauses with at least } 3 \text{ variables in } \mathcal{V}_0(r)$

    \STATE $i \leftarrow 0$
    
    \WHILE{$i = 0$ or $\mathcal{V}_i(r) \ne \mathcal{V}_{i-1}(r)$}

    \STATE $i \leftarrow i + 1$

    \STATE $\mathcal{V}_i(r) \leftarrow \mathcal{V}_{i-1}(r) \cup \var(\mathcal{C}_{i-1}(r))$
    
    \STATE $\mathcal{C}_i(r) \leftarrow \{c \in \mathcal{C} : \lvert \var(c) \cap \mathcal{V}_i(r) \rvert \ge 3 \}$

    \ENDWHILE
    
    \STATE $\bc(r) \leftarrow \mathcal{C}_i(r)$ and $\bv \leftarrow \mathcal{V}_i(r)$
    
    \RETURN $\bv(r), \bc(r)$
  \end{algorithmic}
\end{algorithm}

We define the \emph{good clauses} of $\Phi$ as $\gc(r) = \mathcal{C} \setminus \bc(r)$ and the \emph{good variables} of $\Phi$ as $\gv(r) = \mathcal{V} \setminus \bv(r)$. The sets $\gv(r), \bv(r), \gc(r), \bc(r)$ depend on the parameter $r \in (0,1)$. The value of $r$ here will be $r_0-\delta$ except in Section~\ref{sec:connectivity} where we prove our connectivity results for $r = r_1-\delta$, and in some of the marking results in Section~\ref{sec:marking}. We will use the observations given in Proposition~\ref{prop:bad} several times in this work.

\begin{propbad}
\statepropbad
\end{propbad}
\begin{proof} \label{prop:bad:proof}
  In this proof we briefly explain the implementation of Algorithm~\ref{alg:bad}. First, for each clause $c$ we keep track of the number of bad variables in $\var(c)$, denoted $\operatorname{bad}(c)$. We also have a stack of  bad variables $S_{\mathcal{V}}$ that are yet to be processed by the algorithm. At the start of the algorithm, we set $S_{\mathcal{V}} \leftarrow \mathcal{V}_0$. While $S_{\mathcal{V}}$ is non-empty, we take the variable $v$ on the top of the stack and increase $\operatorname{bad}(c')$ by $1$ for those clauses $c'$ where $v$ appears. If any of these updates gives $\operatorname{bad}(c') \ge 3$, we add $\var(c')$ to the stack $S_{\mathcal{V}}$, set the variables in $\var(c')$  as bad and set the clause $c'$ as bad. At the end of this process, $S_{\mathcal{V}}$ is empty and we have found all the bad variables and bad clauses of $\Phi$. As every variable is added to the stack at most once and the list $\operatorname{bad}(\cdot)$ is updated at most $m k$ times (once per literal in $\Phi$), the running time is  $O(n + m k)$.
\end{proof}

In our work we need a variation of  result of~\cite{galanis2019counting} that controls the number of bad clauses in connected subgraphs of $G_\Phi$. We state this result in Lemma~\ref{lem:bad} and prove it in Appendix~\ref{sec:ap:bad}.

\newcommand{\statelembad}{Let $r \in (0, 1/(2\log 2)]$. There is a positive integer $k_0$ such that for any integer $k \ge k_0$, $\Delta_r = \lceil 2^{r k} \rceil$, and any density $\alpha$ with $\alpha \le \Delta_r/k^3$, the following holds w.h.p. over the choice of $\Phi = \Phi(k, n, \lfloor \alpha n \rfloor)$. For every connected set of clauses $Y$ in $G_\Phi$ such that $\lvert \var(Y) \rvert \ge  2k^4  \log n$, we have $\lvert Y \cap \bc(r) \rvert \le  \lvert Y \rvert / k$.}

\begin{lemma}[{Modified version of~\cite[Lemma 8.16]{galanis2019counting}}] \label{lem:bad}
\statelembad
\end{lemma}

We also need a bound on the number of bad clauses of $\Phi$, which is also proved in Appendix~\ref{sec:ap:bad}.

\newcommand{\statelembadall}{Let $r \in (0, 1/(2\log 2)]$. There is a positive integer $k_0$ such that for any integer $k \ge k_0$, $\Delta_r = \lceil 2^{r k} \rceil$, and any density $\alpha$ with $\alpha \le \Delta_r/k^3$, the following holds w.h.p. over the choice of $\Phi = \Phi(k, n, \lfloor \alpha n \rfloor)$. We have $\lvert \bc(r) \rvert \le  2(\alpha/ \Delta_r)n / 2^{k^{10}}$ and $\lvert \bv(r) \rvert \le  2(k+1)(\alpha/ \Delta_r)n / 2^{k^{10}}$.}

\begin{lemma}[{Modified version of~\cite[Lemma 8.12]{galanis2019counting}}] \label{lem:bad:2}
\statelembadall
\end{lemma}

Lemmas~\ref{lem:bad} and~\ref{lem:bad:2} guarantee that, w.h.p. over the choice of $\Phi$, bad clauses are a minority among all the clauses of $\Phi$. This will be used to show that bad clauses do not affect significantly the behaviour of our sampling algorithm. We point out that the definitions of $\gv(r), \bv(r), \gc(r)$ and $\bc(r)$ given in~\cite{galanis2019counting} have $r = 1/300$ and, in Algorithm~\ref{alg:bad}, use the condition $\lvert \var(c) \cap \mathcal{V}_i(r) \rvert \ge k/10$ instead of $\lvert \var(c) \cap \mathcal{V}_i(r) \rvert \ge 3$ 
 
Hence, our definitions of good clauses and good variables are more restrictive. However, it turns out that, with minor changes, the proof of Lemma~\ref{lem:bad} given in~\cite{galanis2019counting} can be extended to our setting. These changes are explained in Appendix~\ref{sec:ap:bad}. 

\section{Identifying a set of ``marked'' variables with good marginals} \label{sec:marking}
 
A property that is useful for
  sampling satisfying assignments 
is having a high proportion of variables in each good clause such that the marginals of these variables are fairly close to~$1/2$.
That is,  having variables which are roughly equally likely to be true or false in a random
satisfying assignment. The marginals of high-degree variables  do vary.
However, even in the random $k$-SAT model
it turns out that there are  enough variables with marginals near~$1/2$.
Following  the basic approach of Moitra~\cite{moitra19}, we partition the good variables
of a random $k$-CNF formula into types. Here we have three types of variables (instead of two): marked, auxiliary and control variables. The high-level goal is to do this in such a way that each clause has a good proportion of each one of these types of variables. We call this construction a marking, see Definition~\ref{def:distributed-marking} of the proof outline for the precise definition. For such a marking, we will show that as long as the control variables are left unassigned/unpinned, the marginals of the marked and auxiliary variables are all near $1/2$ as a consequence the Lov\'asz local lemma~\cite{erdos1975}. We first set up the notation and results that we need.

It is not difficult to show that in the random $k$-SAT model, 
w.h.p.{}   over the choice of the formula $\Phi$, two distinct clauses share at most $2$ variables (see Lemma~\ref{lem:linearity}). Previous work on counting/sampling satisfying assignments of bounded degree formulae had to analyse subsets of disjoint clauses in order to deal with the fact that small sets of clauses might share most of their variables. The restriction to disjoint subsets imposes further restrictions on the maximum degree of the formula and on the density of the formula in the random $k$-SAT model setting. Here we manage to exploit Lemma~\ref{lem:linearity} to avoid these restrictions.

\begin{lemma} \label{lem:linearity}
 For any $k \ge 3$ and any density $\alpha > 0$ (possibly depending on $k$), the following holds w.h.p. over the choice of the random $k$-CNF formula $\Phi = \Phi(k, n, \lfloor \alpha n\rfloor)$. We have $\lvert \var(c) \rvert \ge k - 1$ and $\lvert \var(c) \cap \var(c')\rvert \le 2$ for all $c, c' \in \mathcal{C}$ with $c \ne c'$.
\end{lemma}
\begin{proof} \label{lem:linearity:proof}
  First, let us prove  that, for $k \ge 3$, w.h.p. over the choice of $\Phi$, $\lvert \var(c) \rvert \ge k -1$ for all $c \in \mathcal{C}$. Let us denote by $\mathcal{R}_c$ the event that a clause $c$ has at least two repetitions among its variables, that is, $\lvert \var(c) \rvert \le k-2$. We claim that $\pr(\mathcal{R}_c) \le q(k) / n^2$, where $q = \binom{k}{3} + k(k-1)(k-2)(k-3)/4$. To prove this statement we note that the probability that a variable  appears at least $3$ times in $c$ is at most $\binom{k}{3} n^{k-2} / n^k$, and the probability that two distinct variables are repeated in $c$ is at most $p(k) n (n-1) n^{k-4}/n^k$ for $p(k) = k(k-1)(k-2)(k-3)/4$.  Hence, by adding up both cases, we find that $\pr(\mathcal{R}_c) \le q(k) / n^2$, and $\pr(\bigcup_{c \in \mathcal{C}}\mathcal{R}_c) \le q(k) m / n^2 \le  q(k) \density /n = O(1/n)$, so the result follows. 

  Let $c, c' \in \mathcal{C}$ with $c \ne c'$. We study $\lvert \var(c) \cap \var(c') \rvert$,
  \begin{equation*}
      \pr \left( \left| \var(c) \cap \var(c') \right| \ge 3 \right) \le \frac{n (n-1) (n-2) n^{2(k-3)} (k(k-1)(k-2))^2 }{n^{2k}} \le \frac{k^6}{n^3}.
  \end{equation*}
  Therefore, the probability that there is a pair of clauses $c, c'$ with $\lvert \var(c) \cap \var(c') \rvert \ge 3$ is bounded from above by  $\frac{m(m-1)}{2} \frac{k^6}{n^3} \le  \frac{\alpha^2}{2} \frac{k^6}{ n} =  O \left( \frac{1}{n}\right)$,   which finishes the proof.
\end{proof}

We will use the asymmetric version of the Lov\'asz local lemma (LLL), proved by Lov\'asz and originally published in~\cite{Spencer1997}. Before stating this result, let us introduce some notation. Let $\mathcal{P}$ be a finite collection of mutually independent random variables. Let $B$ an event that is a function of the random variables in $\mathcal{P}$.
Let $\mathcal{A}$  be a collection of events that are a function of the random variables in $\mathcal{P}$. We define $\Gamma(B)$ as the set of events $A \in \mathcal{A}$ such that $A \ne B$ and $A$ and $B$ are not independent. In this setting, $\pr_P(B)$ is the probability that the event $B$ holds when sampling all the random variables in $\mathcal{P}$.

\begin{theorem}[{Asymmetric Lov\'asz local lemma,~\cite[Theorems 1.1 and 2.1]{Haeupler2011}}]  \label{thm:lll}
    Let $\mathcal{P}$ be a finite collection of mutually independent random variables. Let $\mathcal{A}$  be a collection of events that are a function of the random variables in $\mathcal{P}$. If there exists a function $x : \mathcal{A} \to (0,1)$ such that, for all $A \in \mathcal{A}$, we have
    \begin{equation*}
        \pr_P\left(A\right) \le x(A) \prod_{N \in \Gamma(A) } \left(1 - x(N) \right),
    \end{equation*}
    then $\pr_P\left(\, \bigcap \nolimits_{A \in \mathcal{A}} \overline{A} \, \right) > 0$. Furthermore, for any event $B$ that is a function of the random variables in $\mathcal{P}$, we have
    \begin{equation*}
        \pr_P\left( B \left| \, \bigcap\nolimits_{A \in \mathcal{A}} \overline{A} \right. \right) \le \pr_P\left(B\right) \prod_{A \in \Gamma(B)} \big(1 - x(A) \big)^{-1}.
    \end{equation*}
\end{theorem}

We are going to apply the LLL in Lemma~\ref{lem:marking} to find an $(r_0 - \delta, r_0, r_0, 2r_0)$-marking of $\Phi$ (Definition~\ref{def:distributed-marking}), w.h.p. over the choice of the random formula, for some appropriate $r_0 \in (0, 1)$. Before proving Lemma~\ref{lem:marking}, let us highlight how strong the properties of a marking are. First, the fact that a set of marked variables is $\rho$-distributed (Definition~\ref{def:distributed-marking}) will allow us to find, w.h.p. over the choice of $\Phi$, a good amount of marked variables in any set of clauses, even if the set includes bad clauses, see Lemma~\ref{lem:bound-marked} for a precise statement. This result is an essential ingredient in our proofs. Secondly, as long as the control variables are left unassigned, the marginals of the marked and auxiliary variables will be near $1/2$ as a consequence of the LLL, as we show later in this section (Lemma~\ref{lem:marginals}).   We remark that, in the definition of $\rho$-distributed set of variables, we ask for $\lvert \var(c) \cap V \rvert \ge \rho (k-3)$ instead of $\lvert \var(c) \cap V \rvert \ge \rho k$ to account for the fact that w.h.p. a good clause has at most a repeated variable (Lemma~\ref{lem:linearity}) and at most two bad variables (Proposition~\ref{prop:bad}), which will come up in the proofs presented in this section. First, we need the following definition.

\begin{definition}[$\Phi_{\mathrm{good}}(r)$, $\Phi_{\mathrm{bad}}(r)$]  \label{def:good-phi}
Let $r \in (0,1)$. Let $\Phi = (\mathcal{V}, \mathcal{C})$ be a $k$-CNF formula. Let $\Phi_{\mathrm{good}}(r) = (\gv(r), \gc(r))$ be the CNF formula obtained by taking the good clauses of $\Phi$ and ignoring the bad variables appearing in them. Let $\Phi_{\mathrm{bad}}(r)$ be the $k$-CNF formula with variables $\mathcal{V}_{\mathrm{bad}}(r)$ and clauses $\mathcal{C}_{\mathrm{bad}}(r)$.
\end{definition}

Note that in $G_{\Phi_{\mathrm{good}}(r)}$ two clauses $c_1$ and $c_2$ in $\gc$ are adjacent if and only if $\var(c_1) \cap \var(c_2) \cap \gv \ne \emptyset$.  By definition of good variables, the maximum degree in $G_{{\Phi_{\mathrm{good}}}(r)}$ is at most $k (\Delta_r - 1)$, which will be important when applying the LLL. We also need the following version of Chernoff's bounds.

\begin{lemma}[{Chernoff's bounds -~\cite[Theorem 2.1 and Corollary 4.1]{mulzer2018}}] \label{lem:chernoff}
  Let $n \in \mathbb{N}$, $p \in [0,1]$, and let $X_1, \ldots, X_n$ be $n$ independent random variables with $X_j \in \{0, 1\}$ and $\pr(X_j = 1) = p\,$ for all $j = 1, \ldots, n$. Let $X = \sum_{j = 1}^n X_j$. Then, for any $t \in (p, 1)$ and any $s \in (0, p)$, we have
$\pr \left( X \ge t n  \right) \le e^{-D(t, p) n}$ and $\pr \left( X \le s n  \right) \le e^{-D(s, p) n}$, where, for reals $x,y\in (0,1)$, $D(x,y):=x \log \left( x / y\right) + (1 -x) \log \left( (1-x)/ (1-y) \right)$ is the Kullback-Leibler divergence.
\end{lemma}

We can now state the main result of this section. The Lov\'asz local lemma ideas in the proof of Lemma~\ref{lem:marking} are standard in the literature since the work of Moitra~\cite{moitra19} but the quantities involved are adapted to our setting. 

\begin{lemma} \label{lem:marking}
   There is a positive integer $k_0$ such that for any $k \ge k_0$ and any density $\alpha$  with $\density \le \densitydef$ the following holds w.h.p. over the choice of the random $k$-CNF formula $\Phi = \Phi(k, n, \lfloor \alpha n\rfloor)$:
  \begin{enumerate}
      \item \label{item:marking:1} there exists a partial assignment of bad variables that satisfies all bad clauses;
      \item  \label{item:marking:2} there exists an $(r_0-\delta, r_0, r_0, 2r_0)$-marking of $\Phi$. Furthermore, for any $p \in (0,1)$, such an $(r_0-\delta, r_0, r_0, 2r_0)$-marking can be computed with probability at least $1 - p$ in time  $O(n \log(1/p))$.
  \end{enumerate}
\end{lemma}
      
\begin{proof} \label{lem:marking:proof}
  In this proof we set $r = r_0 - \delta$. We note that for any $k \ge 4$ our density $\density \le \densitydef$ is below the threshold $c_k > 1.3836 \cdot 2^k /k$ established in~\cite[Theorem 1.3]{alan1996}. For densities below this threshold, w.h.p. over the choice of $\Phi$, there is a satisfying assignment for $\Phi$. When $\Phi$ is satisfiable, we claim that there is an assignment of the bad variables that satisfies all bad clauses. Indeed, all the variables in bad clauses are bad (Proposition~\ref{prop:bad}) and, thus, the restriction of a satisfying assignment to $\bv(r)$ must satisfy all the bad clauses. In the rest of this proof we show that assertion~\ref{item:marking:2} also holds.

In view of Lemma~\ref{lem:linearity}, we may assume that $\lvert \var(c) \rvert \ge k -1$ for all $c \in \mathcal{C}$. Let us find the $(r, r_0, r_0, 2r_0)$-marking $(\mathcal{V}_{\mathrm m}, \mathcal{V}_{\mathrm a}, \mathcal{V}_{\mathrm c})$. If all clauses are bad, then we set $\mathcal{V}_{\mathrm c} = \mathcal{V}$, $\mathcal{V}_{\mathrm m} = \emptyset$ and $\mathcal{V}_{\mathrm a} = \emptyset$. This is trivially an $(r, r_0, r_0, 2r_0)$-marking for $\Phi$. In the rest of the proof we assume that there are good variables.  We study the following probability space. For each good variable $v$, we set $v$ as ``marked'' with probability $\beta \in (0,1/2)$, ``auxiliary'' with probability $\beta$ and ``control'' with probability $1-2\beta$. This decision is made independently for each good variable. Each bad variable is set as ``control''. Let $\mathcal{P}$ be the set $\{P_v: v \in \gv(r)\}$, where $P_v$ is the random choice made in this experiment for $v$. Let $\mathcal{V}_{\mathrm m}$ be the set of marked variables, let $\mathcal{V}_{\mathrm a}$ be the set of auxiliary variables, and let $\mathcal{V}_{\mathrm c}$ be the set of control variables obtained by running this experiment. For each clause $c \in \gc(r)$, let $A_c$ be the
event that $c$ has less than $r_0 (k-3)$ marked variables or less than $r_0 (k-3)$ auxiliary variables or less than $2r_0 (k-3)$ good control variables.  
We are going to apply the LLL on the formula $\Phi_{\mathrm{good}}(r)$ so as to show that $\pr( \bigcap\nolimits_{c \in \gc(r)} \overline{A_c}) > 0$. For each $c \in \gc(r)$, in view of Proposition~\ref{prop:bad} and the fact that $\lvert \var(c) \rvert \ge k-1$, we have $\lvert \var(c) \cap \gv(r) \rvert \ge k-3$. Hence, we can apply the Chernoff bound given in Lemma~\ref{lem:chernoff} with $n = \lvert \var(c) \cap \gv(r) \rvert$, $p = \beta$ and $s = r_0$ to obtain, for any choice $V \in \{\mathcal{V}_{\mathrm m}, \mathcal{V}_{\mathrm a}\}$,
  \begin{equation*}
      \pr_{P} \left( \lvert  \var(c) \cap V \rvert < r_0 (k-3) \right) \le e^{-D(r_0, \beta) (k-3)}.
  \end{equation*}
  When $V = \mathcal{V}_{\mathrm c} \setminus \bv$, $n = \lvert \var(c) \cap \gv(r)\rvert$, $p = 1-2\beta$ and $s = 2 r_0$ we obtain
  \begin{equation*}
      \pr_{P} \left( \lvert  \var(c) \cap V \rvert < 2 r_0 (k-3) \right) \le e^{-D(2 r_0, 1-2\beta) (k-3)}.      
  \end{equation*}

  %g[x_, y_] := x Log[x/y] + (1 - x) Log[(1 - x) / (1 - y)];
  % NMaximize[{x, g[x, \[Alpha]/2] >= x * N[Log[2]] && g[2 x, 1 - \[Alpha]] >= x * N[Log[2]] && 0.4 <= \[Alpha] <= 0.7 &&  0.1 <= x <= 0.2  && Element[x, Reals] &&  Element[\[Alpha], Reals] }, {x, \[Alpha]}]
  % y = 0.11775; \[Beta] = 0.571; N[g[y, \[Beta]/2]/Log[2]]; N[g[2 y, 1 - \[Beta]]/Log[2]]
  We have chosen $r_0$ to be as large as possible under the restrictions that $D(r_0, \beta) \ge r_0 \log 2$ and $D(2 r_0, 1-2\beta) \ge r_0 \log 2$. The values $\beta = 0.571027$ and $r_0 = \rvalue$ satisfy these restrictions.
  We conclude that
  \begin{equation*}
      \pr_{P} \left( A_c \right) \le 2 \cdot e^{-D(r_0, \beta) (k-3)} + e^{-D(2 r_0, 1-2\beta) (k-3)} \le 3 \cdot 2^{-r_0 (k-3)}.
  \end{equation*}
  Let $\Delta'= 2^{r_0(k-3)}/(3 e^2 k)$ and let  $x(A_c) = 1 / (k \Delta')$ for all $c \in \gc(r)$. We check that $x$ satisfies the condition of the LLL for $\mathcal{P}$ and $\mathcal{A} = \{A_c : c \in \gc(r)\}$. For $k \ge 43$, $1/(k \Delta') \in (0, 1)$ and thus $x(A_c) \in (0, 1)$ for all $c \in \gc(r)$. We note that $\Gamma(A_c) = \{A_{c'} : c' \in \gc(r), c' \ne c, \var(c') \cap \var(c) \cap \gv(r) \ne \emptyset\}$. The graph $G_{\Phi_{\mathrm{good}}(r)}$, given in Definition~\ref{def:graph-phi}, has maximum degree at most $k(\Delta_{r} - 1)$, so $\lvert \Gamma(A_c)\rvert \le k(\Delta_{r} - 1) \le k \Delta'$, where the latter inequality holds for large enough $k$ as $\Delta_{r} = \lceil 2^{r k} \rceil$ and $r = r_0 - \delta$. Therefore, we have
  \begin{equation}  \label{eq:x-bound}
     x(A_c) \prod_{N \in \Gamma(A_c)} \left( 1 - x(N) \right) \ge  \frac{1}{k \Delta'} \left( 1 - \frac{1}{k \Delta'} \right)^{k \Delta'} \ge \frac{1}{e^2 k \Delta'} = 3 \cdot 2^{-r_0 (k-3)},
  \end{equation}
  where we used $(1-1/z)^{z} \ge e^{-2}$ for all $z \ge 2$ in the second inequality. Thus,
  \begin{equation*}
      x(A_c) \prod_{N \in \Gamma(A_c)} \left( 1 - x(N) \right) \ge 3 \cdot 2^{-r_0 (k-3)} \geq  \pr_{P} \left( A_c \right).
   \end{equation*}
   We conclude that, by the LLL, $\pr_{P}\left(\, \bigcap\nolimits_{c \in \gc(r)} \overline{A_c} \, \right) > 0$, so there exists a partition $(\mathcal{V}_{\mathrm m}, \mathcal{V}_{\mathrm a}, \mathcal{V}_{\mathrm c})$ of the variables of $\Phi$ such that $\bv(r) \subseteq \mathcal{V}_{\mathrm c}$ and each good clause contains at least $r_0 (k-3)$ marked variables, $r_0 (k-3)$ auxiliary variables and $2 r_0 (k-3)$ good control variables. That is, $(\mathcal{V}_{\mathrm m}, \mathcal{V}_{\mathrm a}, \mathcal{V}_{\mathrm c})$ satisfies Definition~\ref{def:distributed-marking} for $r = r_0 - \delta$, $r_{\mathrm m} = r_0$, $r_{\mathrm a} = r_0$, and $r_{\mathrm c} = 2r_0$.  Moreover, with probability at least $1-\delta$, this partition can be computed in $4 n \alpha \Delta'k \log(1/\delta)$ steps with the algorithm of Moser and Tardos~\cite{moser2010}. 
\end{proof}

We now give the marking result that we use in our connectivity results, which holds for densities at most $2^{(r_1-\delta) k}/k^3$, where $r_1 = \rone$. The larger density threshold comes from the fact that the marking result is less strong -- we do not require auxiliary variables nor a high number of good control variables in every clause.

\begin{lemma} \label{lem:marking:r1}
   There is a positive integer $k_0$ such that for any $k \ge k_0$ and any density $\alpha$  with $\density \le 2^{(r_1-\delta) k}/k^3$ the following holds w.h.p. over the choice of the random $k$-CNF formula $\Phi = \Phi(k, n, \lfloor \alpha n\rfloor)$:
  \begin{enumerate}
      \item  there exists a partial assignment of bad variables that satisfies all bad clauses;
      \item  there exists an $(r_1 - \delta, r_1, 0, r_1)$-marking of $\Phi$. Furthermore, for any $p \in (0,1)$, such an $(r_1-\delta, r_1, 0, r_1)$-marking can be computed with probability at least $1 - p$ in time  $O(n \log(1/p))$.
  \end{enumerate}
\end{lemma}
\begin{proof}
    The proof is analogous to that of Lemma~\ref{lem:marking}. Here we explain the main differences. First, we set $r= r_1 -\delta$ instead of $r = r_0 - \delta$. The second difference is that we study the following probability space: each good variable $v$ is set as ``marked'' with probability $\beta$ and ``control'' with probability $1-\beta$. We let $A_c$ be the event that $c$ has less than $r_1 (k-3)$ marked variables or less than $r_1(k-3)$ good control variables. A Chernoff bound as in the proof of Lemma~\ref{lem:marking} gives
    \begin{equation*}
      \pr_{P} \left( A_c \right) \le e^{-D(r_1, \beta) (k-3)} + e^{-D(r_1, 1-\beta) (k-3)} \le 2 \cdot 2^{-r_1 (k-3)},
    \end{equation*}
    where we chose $r_1$ as large as possible so that $D(r_1, \beta) \ge r_1 \log 2$ and $D(r_1, 1-\beta) \ge r_1 \log 2$. The choices $\beta = 1/2$ and $r_1 = \rone$ satisfy these restrictions. We let $\Delta'= 2^{r_1(k-3)}/(3 e^2 k)$ and let  $x(A_c) = 1 / (k \Delta')$ for all $c \in \gc(r)$. It remains to check that we can apply the asymmetric LLL on the formula $\Phi_{\mathrm{good}}(r)$ to conclude that $\pr( \bigcap\nolimits_{c \in \gc(r)} \overline{A_c}) > 0$. This was done in equation \eqref{eq:x-bound} in Lemma~\ref{lem:marking}. We note that the bound given in \eqref{eq:x-bound} also holds in our current setting if we replace $r_0$ by $r_1$. We find that $x(A_c) \prod_{N \in \Gamma(A_c)} \left( 1 - x(N) \right) \ge 3 \cdot 2^{-r_1 (k-3)} \geq  \pr_{P} \left( A_c \right)$ and, thus, there exists a partition $(\mathcal{V}_{\mathrm m}, \mathcal{V}_{\mathrm a}, \mathcal{V}_{\mathrm c})$ of the variables of $\Phi$ such that $\bv(r) \subseteq \mathcal{V}_{\mathrm c}$, $\mathcal{V}_{\mathrm a} = \emptyset$, and each good clause contains at least $r_1 (k-3)$ marked variables and at least $r_1 (k-3)$ good control variables. 
\end{proof}

 In the remaining of this section we bound the marginals of $\mu_\Omega$ (recall that $\mu_\Omega$ is the uniform distribution over the satisfying assignments of the formula $\Phi$, Definition~\ref{def:subformula}) on any marked and auxiliary variable. In fact, we prove the stronger result that the marginal distribution of $\mu_{\Omega}$ on $\mathcal{V}_{\mathrm m} \cup \mathcal{V}_{\mathrm a} $ is $\varepsilon$-uniform, i.e., very close to the uniform distribution, see Definition~\ref{def:uniform}. We give a bound for each one of the markings established in Lemmas~\ref{lem:marking} and~\ref{lem:marking:r1}. 
 Here we write $\Lambda_1 \cup \Lambda_2$ for the combined assignment of $\Lambda_1$ and $\Lambda_2$.

\begin{lemma} \label{lem:marginals}
    Let $\Phi = (\mathcal{V}, \mathcal{C})$ be a satisfiable $k$-CNF formula. The following claims hold.
      \begin{enumerate}
        \item Let $r = r_0 - \delta$ and let $(\mathcal{V}_{\mathrm m}, \mathcal{V}_{\mathrm a}, \mathcal{V}_{\mathrm c})$ be a $(r, r_0, r_0, 2r_0)$-marking of $\Phi$. Then for any satisfying assignment $\Lambda_{\bad}$ of $\Phi_{\bad}(r)$, any assignment $\Lambda\colon S \to \{\mathsf{F}, \mathsf{T} \}$ where $S \subseteq \mathcal{V}_{\mathrm m} \cup \mathcal{V}_{\mathrm a}$, and any $v \in \gv(r) \setminus S$ we have 
        \begin{equation*}
         \max \left\{ \pr_{\mu_{\Omega}}\left(\left. v \mapsto \mathsf{F} \right| \Lambda \cup \Lambda_{\bad} \right), \pr_{\mu_{\Omega}}\left(\left. v \mapsto \mathsf{T} \right| \Lambda \cup \Lambda_{\bad} \right)  \right\}  \le  \frac{1}{2} \exp \left( \frac{1}{k 2^{r_0 k}} \right).
      \end{equation*}
      In particular,  the distribution $\restr{\mu_{\Omega}}{\mathcal{V}_{\mathrm m} \cup \mathcal{V}_{\mathrm a}}$ is $(2^{-r_0 k}/k)$-uniform.
      
    \item Let $r = r_1 - \delta$ and let $(\mathcal{V}_{\mathrm m}, \emptyset, \mathcal{V}_{\mathrm c})$ be a $(r, r_1, \emptyset, r_1)$-marking of $\Phi$. Then, for any satisfying assignment $\Lambda_{\bad}$ of $\Phi_{\bad}(r)$, any assignment $\Lambda\colon S \to \{\mathsf{F}, \mathsf{T} \}$ where $S \subseteq \mathcal{V}_{\mathrm m}$, and any $v \in \gv(r) \setminus S$ we have 
        \begin{equation*}
         \max \left\{ \pr_{\mu_{\Omega}}\left(\left. v \mapsto \mathsf{F} \right| \Lambda \cup \Lambda_{\bad} \right), \pr_{\mu_{\Omega}}\left(\left. v \mapsto \mathsf{T} \right| \Lambda \cup \Lambda_{\bad} \right)  \right\}  \le  \frac{1}{2} \exp \left( \frac{1}{k} \right).
      \end{equation*}    
    In particular, the distribution $\restr{\mu_{\Omega}}{\mathcal{V}_{\mathrm m}}$ is $(1/k)$-uniform.
      \end{enumerate}
\end{lemma}

\begin{proof} \label{lem:marginals:proof}
    We prove each one of the claims separately. The proofs are analogous so for the second claim we only highlight the differences in the proof.
    \begin{enumerate}
        \item Here $r = r_0 - \delta$. Let $\Lambda_{\bad}$ be an assignment of bad variables that satisfies all bad clauses. Let $S \subseteq \mathcal{V}_{\mathrm m} \cup \mathcal{V}_{\mathrm a}$, let $\Lambda$ be an assignment of $S$ to $\{\mathsf{F}, \mathsf{T}\}$, and let $v \in \gv(r) \setminus S$. We note that $\pr_{\mu_{\Omega^{\tau}}}\left( \cdot \right) =  \pr_{\mu_{\Omega}}\left( \cdot \, \vert \tau \right)$ for any assignment $\tau$ of some variables. In light of this observation, we are going to prove that 
        \begin{equation} \label{eq:proof-marginals}
         \max \left\{ \pr_{\mu_{\Omega^{\Lambda \cup \Lambda_{\bad}}}}\left(v \mapsto \mathsf{F}\right), \pr_{\mu_{\Omega^{\Lambda \cup \Lambda_{\bad}}}}\left(v \mapsto \mathsf{T}\right)  \right\}  \le  \frac{1}{2} \exp \left( \frac{1}{k 2^{r_0 k}} \right).
      \end{equation}
  We apply the LLL to the formula $\Phi' := \Phi^{\Lambda \cup \Lambda_{\bad}}$ as follows. Let $\mathcal{V}'$ and $\mathcal{C}'$ be the sets of variables and clauses of $\Phi'$. Note that, $\mathcal{V}' \subseteq \gv(r)$, $\mathcal{C}' \subseteq \gc(r)$ and $G_{\Phi'}$ is a subgraph of $G_{\Phi_{\mathrm{good}}(r)}$ as all bad variables have been assigned a value and all bad clauses have been satisfied. We set $P_v = \sigma(v)$ for all $v \in \mathcal{V}'$, where $\sigma \colon \mathcal{V}' \to \{\mathsf{F}, \mathsf{T}\}$ is chosen uniformly at random from the set of assignments $\mathcal{V}' \to \{\mathsf{F}, \mathsf{T}\}$, and $\mathcal{P} = \{P_v : v \in \mathcal{V}'\}$. We define the set $\mathcal{A}$ as the set containing for all $c \in \mathcal{C}'$ the event $A_c = $ ``the clause $c$ is not satisfied by the random assignment $\sigma$". By the definition of $(\mathcal{V}_{\mathrm m}, \mathcal{V}_{\mathrm a}, \mathcal{V}_{\mathrm c})$, there are at least $2\rh (k-3)$  good control variables in $c$. Since good control variables are not assigned a value by $\Lambda \cup \Lambda_{\bad}$ and, thus, they are in $\mathcal{V}'$, we have $\pr_{P}(A_c) \le 2^{-2\rh (k-3)}$. Recall that $\Delta_{r} = \mdegdef$ (Definition~\ref{def:degree}). Let $\Delta'= 2^{2 r_0(k-3)}/(e^2 k)$ and let $x(A_c) = \frac{1}{k \Delta'}$ for all $c \in \mathcal{C}'$. Let us show that $x$ satisfies the LLL condition in this setting. In view of $\Gamma(A_c) = \{A_{c'} : c' \in \mathcal{C}',  c' \neq c, \var(c) \cap \var(c') \cap \mathcal{V}' \ne \emptyset\}$, which can be identified with a subset of the neighbours of $c$ in $G_{\Phi_{\mathrm{good}}(r)}$, and $\lvert \Gamma(A_c)\rvert \le k \Delta_{r} \le k \Delta'$ for large enough $k$, we find that
  \begin{equation*}
        x(A_c) \prod_{N \in \Gamma(A_c)} \left( 1 - x(N) \right) \ge  \frac{1}{k \Delta'} \left( 1 - \frac{1}{k \Delta'} \right)^{k\Delta'} \ge \frac{1}{e^2 k \Delta'} =  2^{-2 r_0 (k-3)} \ge \pr_{P} \left( A_c \right),
  \end{equation*}
  where we used $(1-1/z)^{z} \ge e^{-2}$ for all $z \ge 2$.
   Let $A = \{v \mapsto \mathsf{T}\} :=  
   \{\sigma\colon \mathcal{V}' \to \{\mathsf{F}, \mathsf{T}\} \text{ with } \sigma(v) = \mathsf{T}\}$. In $\Phi'$, we have $\Gamma(A) = \{A_c : c  \in \mathcal{C}', v \in \var(c) \}$, so $\lvert \Gamma(A) \rvert < \Delta_{r}$. By the LLL, we obtain
   \begin{equation*}
     \pr_P\left( v \mapsto \mathsf{T} \left| \, \bigcap\nolimits_{c \in \mathcal{C}'} \overline{A_c} \right. \right) \le \frac{1}{2} \prod_{N \in \Gamma(A) } \big(1 - x(N) \big)^{-1} \le \frac{1}{2}\left(1 - \frac{1}{k \Delta'} \right)^{-(\Delta_r - 1)}.
   \end{equation*}
   For $x > 1$, we have $(1 - 1/x)^{-1} = 1+1/(x-1)\le \exp(1/(x-1))$. We find that
   \begin{equation*}
     \pr_P\left( v \mapsto \mathsf{T} \left| \, \bigcap\nolimits_{c \in \mathcal{C}'} \overline{A_c} \right. \right) \le \frac{1}{2} \exp \left(\frac{\Delta_r - 1}{k \Delta'- 1}\right) \le \frac{1}{2} \exp \left(\frac{1}{k 2^{r_0 k}}\right),
   \end{equation*}
   where in the latter inequality we used $(p-j)/(q-j) \le p/q$ for all $0 < j < p \le q$ and the fact that $\Delta_r = \mdegdef \le 2^{-r_0k} \cdot 2^{2 r_0(k-3)}/(e^2 k) = 2^{-r_0 k} \Delta'$ for large enough $k$.
   We note that $\pr_{\mu_{\Omega^{\Lambda \cup \Lambda_{\bad}}}}\left( \cdot \right) =  \pr_{P}\left( \cdot \, \vert \bigcap\nolimits_{c \in \mathcal{C}'} \overline{A_c}\, \right)$, which completes the proof of one of the upper bounds of~\eqref{eq:proof-marginals}. The other upper bound is proved analogously by applying the LLL with $A = \{v \mapsto \mathsf{F}\}$. Finally, we conclude that the distribution $\restr{\mu_{\Omega}}{\mathcal{V}_{\mathrm m} \cup \mathcal{V}_{\mathrm a}}$ is $(2^{-r_0 k}/k)$-uniform by the arbitrary choice of $\Lambda_{\bad}$ and the law of total probability, see Definition~\ref{def:uniform}. 
   \item The proof is analogous. The only changes are $r = r_1 - \delta$, $\Delta' = 2^{r_1 (k-3)}/(e^ 2 k)$, and the fact that, since each good clause has at least $r_1(k-3)$ good control variables, we have $\pr_{P}(A_c) \le 2^{-r_1 (k-3)}$. This time we have $x(A_c) \prod_{N \in \Gamma(A_c)} \left( 1 - x(N) \right)\ge \frac{1}{e^2 k \Delta'} \ge \pr_{P} \left( A_c \right)$, which justifies our choice of $\Delta'$. Thus, we can apply the LLL, and the conclusion this time becomes 
      \begin{equation*}
     \pr_P\left( v \mapsto \mathsf{T} \left| \, \bigcap\nolimits_{c \in \mathcal{C}'} \overline{A_c} \right. \right) \le \frac{1}{2} \exp \left(\frac{\Delta_r - 1}{k \Delta'- 1}\right) \le \frac{1}{2} \exp \left(\frac{1}{k}\right),
   \end{equation*}
   where in the latter inequality we used $(p-j)/(q-j) \le p/q$ for all $0 < j < p \le q$ and the fact that $\Delta_r = \lceil 2^{(r_1 - \delta) k} \rceil \le 2^{r_1 (k-3)}/(e^2 k) = \Delta'$ for large enough $k$.
    \qedhere
\end{enumerate}
\end{proof}

The $(1/k)$-uniform property proved in Lemma~\ref{lem:marginals} is remarkably strong: as long as the control variables are left unassigned, the rest of the variables have marginals close to $1/2$, even if some of the marked and auxiliary variables are pinned / have already been assigned a value. This property is used several times in this work and will allow us to prove that, for any pinning of some marked variables, the influences between marked variables are bounded. In the following corollary we extend Lemma~\ref{lem:marginals} to the distributions computed by the Glauber dynamics on the marked variables.

\begin{corollary} \label{cor:marginals:glauber}
     Let $r=r_0 -\delta$. Let $\Phi = (\mathcal{V}, \mathcal{C})$ be a satisfiable $k$-CNF formula that has an $(r, r_0, r_0, 2r_0)$-marking $(\mathcal{V}_{\mathrm m}, \mathcal{V}_{\mathrm a}, \mathcal{V}_{\mathrm c})$. Let $\rho$ be an integer with $1 \le \rho < \lvert \mathcal{V}_{\mathrm m} \rvert$. Let $t$ be a non-negative integer and let $X_t$ be the (random) assignment obtained after running the $\rho$-uniform-block Glauber dynamics on the marked variables for $t$ steps, starting on an assignment $X_0$ that is chosen uniformly at random. Then the probability distribution of $X_t$ is $(2^{-r_0 k}/k)$-uniform.
\end{corollary}

\begin{proof} \label{cor:marginals:glauber:proof}
    Let $\epsilon = (2^{-r_0 k}/k)$. Let $V_1, V_2, \ldots, $ be a possible choice of sets of marked variables to be updated when running the $\rho$-uniform-block Glauber dynamics. We are going to prove that, conditioning on this choice of sets of variables, the probability distribution of $X_t$ is $\epsilon$-uniform. Note that by the law of total probability and the fact that the choice of $V_1, V_2, \ldots$ is arbitrary, this is enough to conclude the result. We carry out the proof by induction on $t$. Let $\pi_t$ be the probability distribution of $X_t$. As $\pi_0$ is the uniform distribution over assignments on $\mathcal{V}_{\mathrm m}$, the claim holds for $t = 0$. Let us now assume that $\pi_{t-1}$ is $\epsilon$-uniform and let us prove that this is also the case for $\pi_t$.  To show the desired uniformity of $\pi_t$ (cf.  Definition~\ref{def:uniform}), consider arbitrary $v \in \mathcal{V}_{\mathrm m}$ and $\Lambda \colon \mathcal{V}_{\mathrm m} \setminus \{v\} \to \{\mathsf{F}, \mathsf{T}\}$, we need to bound $\pr_{\pi_t}\left( \left. v \mapsto \mathsf{F} \right|  \Lambda \right)$ and $\pr_{\pi_t}\left( \left. v \mapsto \mathsf{T} \right|  \Lambda \right)$. We distinguish two cases:
    \begin{itemize}
        \item Case $v \in V_t$. By definition of the Glauber dynamics, the values of $X_t$ on $V_t$ are obtained by sampling from the distribution $\mu_{\Omega}$ conditioned on the restriction of $X_{t-1}$ to $\mathcal{V}_{\mathrm m} \setminus V_t$. Thus, we have $\pr_{\pi_t}\left( \left. v \mapsto \mathsf{F} \right| \Lambda \right) = \pr_{\mu_{\Omega^{\Lambda}}}\left(v \mapsto \mathsf{F} \right)$ since the conditioning involving $\Lambda$  sets all the marked variables other than $v$. As $\restr{\mu_{\Omega}}{\mathcal{V}_{\mathrm m} \cup \mathcal{V}_{\mathrm a}}$ is $\epsilon$-uniform by Lemma~\ref{lem:marginals}, we conclude that 
        $\pr_{\pi_t}\left( \left. v \mapsto \mathsf{F} \right| \Lambda \right) = \pr_{\mu_{\Omega^{\Lambda}}}\left(v \mapsto \mathsf{F} \right) \le \frac{1}{2}\exp(\epsilon)$. The same bound holds for $v \mapsto \mathsf{T}$.
        \item Case $v \not\in V_t$. If $v$ is not updated in steps $1$ through $t$, then $\pr_{\pi_t}\left( \left. v \mapsto \mathsf{F} \right| \Lambda \right) = \pr_{\pi_0}\left( v \mapsto \mathsf{F} \right) = 1/2$. Otherwise, let $j$ be the largest integer with $j < t$ such that $v \in V_j$. Let  $\Lambda_j$ be the restriction of $\Lambda$ to $\mathcal{V}_{\mathrm m} \setminus \bigcup_{i \in \{j+1, j+2, \ldots, t\}} V_i$. By the induction hypothesis, $\pr_{\pi_t}\left( \left. v \mapsto \mathsf{F} \right| \Lambda \right) = \pr_{\pi_j}\left( \left. v \mapsto \mathsf{F}  \right| \Lambda_j \right) \le (1/2) \exp(\epsilon)$. The same bound holds for $v \mapsto \mathsf{T}$.
  \end{itemize}
   As both cases are exhaustive, the proof is concluded.
\end{proof}

Previous work on counting/sampling satisfying assignments of $k$-CNF formulae does not require the use of auxiliary variables, so the marking used is of the form $(\mathcal{V}_{\mathrm m}, \mathcal{V}_{\mathrm c})$. Here auxiliary variables play an essential role in bounding the influences between marked variables as we illustrated in Section~\ref{sec:po}. In order for this approach to be successful, we have to show that a large proportion of the variables are marked. We conclude this section with the following bound on the size of $\mathcal{V}_{\mathrm m}$.

\begin{corollary} \label{cor:size-r-distributed}
 Let $r \in(0,1/(2 \log 2))$. There is an integer $k_0$ such that for any $k \ge k_0$ and any density $\alpha$ with $\density \le \Delta_r / k^3$ the following holds w.h.p. over the choice of the random $k$-CNF formula $\Phi = \Phi(k, n, \lfloor \alpha n\rfloor)$. For any $\rho \in (0,1)$ and any set of good variables $V$ that is $\rho$-distributed we have $\lvert V \rvert \ge (\rho - \delta) (k \alpha / \Delta_r) n$.
\end{corollary}
\begin{proof} \label{cor:size-r-distributed:proof}
    W.h.p. over the choice of $\Phi$, by Lemma~\ref{lem:bad:2} we have $\lvert \bc(r) \rvert \le 2 (\alpha/ \Delta_r) n/2^{k^{10}} \le \alpha n / 4^k$, so $\lvert \gc(r) \rvert \ge  \lvert \mathcal{C} \rvert - \alpha n / 4^k \ge \alpha n -1-\alpha n / 4^k = \alpha n (1-1/4^k) - 1$. Since $V$ is $\rho$-distributed, counting repetitions, there are at least $\rho(k-3) \lvert \gc (r) \rvert$ occurrences of the variables of $V$ in the good clauses of $\Phi$. Each good variable occurs in at most $\Delta_r$ good clauses, so we find that
    \begin{equation*}
        \lvert V \rvert \ge \frac{\rho (k-3) \lvert \gc(r) \rvert}{\Delta_r} \ge \frac{\rho (k-3)}{\Delta_r} \left(\alpha n \left(1-\frac{1}{4^k}\right) - 1\right) \ge \frac{\rho (k-4)}{\Delta_r} (\alpha n -1),
    \end{equation*}
    which is at least $(\rho - \delta) (k \alpha / \Delta_r) n$ for large enough $k$.
\end{proof}

\section{Analysis of the connected components of $\Phi^{\Lambda}$} \label{sec:errors}

In this section we prove Lemma~\ref{lem:cc-general}, which bounds the size of the connected components of $\Phi^{\Lambda}$, where $\Lambda$ is drawn from a $(1/k)$-uniform distribution over an $(r+\delta)$-distributed set of good variables. In order to carry out this proof, we have to understand the structure of logarithmic-sized sets of clauses of the random $k$-CNF formula $\Phi$. Section~\ref{sec:errors:log} is devoted to this purpose. In Section~\ref{sec:errors:nmarked} we apply the results of Section~\ref{sec:errors:log} to obtain a lower bound of the number of marked/auxiliary variables in logarithmic-sized sets of clauses. Finally, in Section~\ref{sec:errors:proof} we complete the proof of Lemma~\ref{lem:cc-general}.

\subsection{Logarithmic-sized sets of clauses in the random \texorpdfstring{$k$}{k}-SAT model} \label{sec:errors:log}

A connected graph $H = (V, E)$ has \emph{tree-excess} $c \in \mathbb{Z}_{\geq 0}$ if $\lvert E \rvert = c + \lvert V \rvert - 1$. It turns out that, w.h.p. over the choice of $\Phi$, small connected sets of clauses of $\Phi$ have tree-excess bounded by a quantity that only depends on $k$ and the density $\alpha$. This property is established in Lemma~\ref{lem:tree-excess} and is essential to our proofs.
 
\begin{lemma} \label{lem:tree-excess}
 Let $k \ge 3$ be an integer. Let $b > 0$ and $\alpha > 0$ be real numbers. W.h.p. over the choice of the random $k$-CNF formula $\Phi = \Phi(k, n, \lfloor \alpha n \rfloor)$, every connected subset of clauses with size at most $b \log(n)$ has tree-excess at most $c := \max\{1,  2 b \log(e k^2 \alpha) \}$.
\end{lemma}
\begin{proof} \label{lem:tree-excess:proof}
  Let $n$ be the number of variables and $m$ be the number of clauses of $\Phi$, so $m/n \le \alpha$. Note that the probability that two clauses of $\Phi$ are not disjoint is at most $k^2 / n$. Let $\ell \in \{1, 2, \ldots, \lfloor b \log(n) \rfloor \}$. We upper bound the probability that there is a connected subset of clauses of size $\ell$ with tree-excess at least $c+1$ by
  \begin{equation} \label{eq:tree-excess}
   \binom{m}{\ell} \ell^{\ell-2} \binom{\ell (\ell-1) / 2}{c+1} \left( \frac{k^2}{n} \right)^{\ell+c}, 
  \end{equation}
  where the factors appearing are the following ones:
  \begin{itemize}
      \item $\binom{m}{\ell}$ is the number of subsets of clauses of size $\ell$;
      \item  $\ell^{\ell-2}$ is the number of trees on $\ell$ labelled vertices;
      \item $\binom{\ell (\ell-1) / 2}{c+1}$ is the number of ways to pick $c+1$ pairs of distinct clauses of a set of size $\ell$;
      \item $\left( k^2 / n \right)^{\ell+c}$ is an upper bound of the probability that all the edges chosen in the two previous items appear in the graph $G_\Phi$.
  \end{itemize}
  We are going to show that the probability given in~\eqref{eq:tree-excess} is $O(n^{-c/4})$, where the hidden constant only depends on $k$. If this holds, by a union bound over $\ell \in \{1, 2, \ldots, \lfloor b \log(n) \rfloor\}$, we would find that the probability that there is a connected subset of clauses of $\Phi$ with size at most $b \log(n)$ and tree-excess at least $c+1$ is $O( b \log(n) n^{-c/4}) = o(1)$. This would complete the proof. Using the inequality $\binom{p}{q} \le (e p / q)^q$ and $m/n \le \alpha$ we can bound~\eqref{eq:tree-excess} by 
  \begin{equation} \label{eq:tree-excess:2}
  \begin{aligned}
   \left( \frac{e m}{\ell} \right)^\ell \ell^{\ell-2}  \left(\frac{e \ell (\ell-1) / 2}{c+1}\right)^{c+1} \left( \frac{k^2}{n} \right)^{\ell+c} & \le \left( \frac{e m}{\ell} \right)^\ell \ell^{\ell-2}  \left(\frac{e \ell^2 / 2}{c+1}\right)^{c+1} \left( \frac{k^2}{n} \right)^{\ell+c} \\
   & = \left(\frac{e}{2c+2}\right)^{c+1} \left( \frac{e m k^2 }{n} \right)^\ell  \left( \frac{k^2 \ell^2}{n} \right)^{c} \\ 
  & \le \left(\frac{e}{2c+2}\right)^{c+1} \left(e k^2 \alpha \right)^\ell  \left( \frac{k^2 \ell^2}{n} \right)^{c}.
  \end{aligned}
  \end{equation}
  Now we distinguish two cases:
  \begin{itemize}
    \item Case when  $e k^2 \alpha \le 1$ . We have $c = 1$ by definition. Thus,~\eqref{eq:tree-excess:2} can be further bounded by 
      \begin{equation*}
          \left(\frac{e}{2c+2}\right)^{c+1}  \left( \frac{k^2 \ell^2}{n} \right)^{c} = O \left( \frac{(\log n)^2}{n} \right) = O\left( n^{-c/4} \right)
      \end{equation*}
      as we wanted.
    \item Case when  $e k^2 \alpha > 1$. Then, as $\ell \le b \log n$ and $b \log(e k^2 \alpha) \le c/2$ by definition, we have
    \begin{equation*}
      \left( e k^2 \alpha \right)^\ell \le \left( e  k^2 \alpha \right)^{b \log n}  = n^{b \log(e k^2 \alpha)} \le n^{c/2}.
    \end{equation*}    
    We conclude that~\eqref{eq:tree-excess:2} can be further bounded by
   \begin{equation*}
     \left(\frac{e}{2c+2}\right)^{c+1}
     \left( \frac{k^2 \ell^2}{\sqrt{n}} \right)^{c} = \left(\frac{e}{2c+2}\right)^{c+1}
     \left( \frac{k^4 \ell^4}{n} \right)^{c/2} = O \left( n^{-c/4} \right)
   \end{equation*}
   as we wanted, where we used $c > 0$. \qedhere 
  \end{itemize}
\end{proof}

Recall that in Lemma~\ref{lem:bad} we established that, in sets of clauses that have at least $2 k^4 \log n$ variables,   the number of bad clauses of $\Phi$ is not too large. We aim to apply Lemma~\ref{lem:bad} to logarithmic-sized sets of clauses. In general, $\lvert Y \rvert$ might be significantly larger than $\lvert \var(Y) \rvert$, so it is not clear how to apply Lemma~\ref{lem:bad}. However, in the random $k$-CNF formula setting the following holds.

\begin{lemma} \label{lem:bound-Y}
  Let $k \ge 3$ be an integer and let $a > 0$ and $\alpha > 0$ be real numbers. W.h.p. over the choice of $\Phi = \Phi(k, n, \lfloor \alpha n \rfloor)$, for every set of clauses $Y$ with $\lvert Y \rvert \ge a \log n$, we have $\lvert \var(Y) \rvert \ge a \log n$.
\end{lemma}
\begin{proof} \label{lem:bound-Y:proof}
Let $\ell := \lceil a \log n \rceil - 1$ and let $m = \lfloor \alpha n \rfloor$.
  We prove the equivalent statement that, w.h.p. over the choice of $\Phi$, for every set of clauses $Y$ with $\lvert \var(Y) \rvert \le \ell$, we have $\lvert Y \rvert \le \ell$.  We note that if there is a set of clauses $Y$ with $\lvert \var(Y) \rvert \le \ell$ and $\lvert Y \rvert > \ell$, then for any subset $Y'$ of $Y$ with $\lvert Y' \rvert = \ell +1$ we have $\lvert \var(Y') \rvert \le \lvert \var(Y) \rvert \le  \ell$. Hence, it suffices to prove
  that there is no set $Y$ of clauses with
  $|\var(Y)| \leq \ell$ and $|Y|=\ell+1$.
 We can assume $n$ is large enough so that $\ell \le e \cdot n$.
  
  Let $\mathcal{E}$ be the event that there is a set of clauses $Y$ of size $\ell+1$ and a set of variables $X$ of size $\ell$ such that all clauses in $Y$ have all variables in $X$. Then by a union bound
  \begin{align*}
    \pr \left( \mathcal{E} \right) \le \binom{m}{\ell+1} \binom{n}{\ell} \left(\frac{\ell }{n} \right)^{(\ell+1) k},
  \end{align*}
  where the first factor is the number of sets $Y$, the second factor is the number of sets $X$ and the third factor is the probability that all variables in the clauses of $Y$ are in $X$. From the well-known bound $\binom{p}{q} \le (e p / q)^{q}$, we obtain
  \begin{align*}
    \pr \left( \mathcal{E} \right) & \le \left(\frac{e m}{\ell+1}\right)^{\ell+1}  \left(\frac{e n}{\ell}\right)^{\ell} \left(\frac{\ell }{n} \right)^{(\ell+1) k} \le \left(\frac{e m}{\ell}\right)^{\ell+1}  \left(\frac{e n}{\ell}\right)^{\ell+1} \left(\frac{\ell }{n} \right)^{(\ell+1) k} \\
    & \le \left(\frac{e \alpha n}{\ell}\right)^{\ell+1}  \left(\frac{e n}{\ell}\right)^{\ell+1} \left(\frac{\ell }{n} \right)^{(\ell+1) k} = \left( e^2 \alpha \frac{\ell^{k - 2} }{n^{k-2}} \right)^{\ell + 1},
  \end{align*}
  which is $O(\log(n) / n)$ because $k \ge 3$ and $\ell = O (\log n)$.
\end{proof}

\subsection{Number of marked variables in logarithmic-sized sets of clauses} \label{sec:errors:nmarked}

Our results on random $k$-CNF formulae can now be combined to give a lower bound on the number of marked / auxiliary variables in logarithmic-sized sets of clauses. We prove this result in a more general setting by considering a set of good variables $V$ that is $r'$-distributed for the formula $\Phi$. The reader can think of $V$ as the set of marked variables or the set of auxiliary variables for one of the markings established in Section~\ref{sec:marking}. 

\begin{lemma} \label{lem:bound-marked}
  Let $r \in (0,1/(2 \log 2)]$, $r' \in (0,1)$ and $\hat\delta \in (0, r)$. There is a positive integer $k_0$ such that, for any integer $k \ge k_0$, any density $\alpha \le \Delta_r / k^3$ and any real number $b$ with $2 k^4  < b$, the following holds w.h.p. over the choice of $\Phi = \Phi(k, n, \lfloor \density n\rfloor)$. Let $V$ be a set of good variables that is $r'$-distributed. Then, for every set of clauses $Y$ that is connected in $G_{\Phi}$ such that $2k^4 \log(n) \le \lvert Y \rvert \le b \log (n)$, we have $\lvert \var(Y) \cap V \rvert \ge (r' - \hat\delta) k \lvert Y \rvert$.
\end{lemma}
\begin{proof} \label{lem:bound-marked:proof}
Let $a=2 k^4$.
  We apply Lemma~\ref{lem:bad}   to find that there is $k_1$ such that for $k \ge k_1$, w.h.p. over the choice of $\Phi$, for every set of clauses $Y$ that is connected in $G_{\Phi}$,
  \begin{equation} \label{eq:cbad}
     \text{if } \ \lvert \var(Y) \rvert \ge a \log(n), \text{ then }  \lvert Y \cap \bc(r) \rvert \le \lvert Y \rvert / k.
  \end{equation}
   We apply Lemma~\ref{lem:bound-Y} with $a = 2k^4$ to find that, w.h.p. over the choice of $\Phi$, for every set of clauses $Y$, we have
  \begin{equation}\label{eq:Y-bound}
  \begin{aligned}
     \text{if } & \ \lvert Y \rvert \ge a \log(n), \text{then } \lvert \var(Y)  \rvert \ge a \log(n).
  \end{aligned}
  \end{equation}
  Finally, for any $b>0$, we apply Lemma~\ref{lem:tree-excess}, obtaining that, w.h.p. over the choice of $\Phi$,  for every set of clauses $Y$ that is connected in $G_{\Phi}$,
   \begin{equation} \label{eq:errors:tree-excess}
      \text{if } \lvert Y \vert \le b \log n, \text{ then }  Y \text{ has tree-excess at most } c= \max \{1, 2 b \log (e  k^2 \alpha)\} = O(1). 
  \end{equation}
  Let $Y$ be a set of clauses that is connected in $G_{\Phi}$ such that $a \log(n) \le \lvert Y \rvert \le b \log (n)$. Then, by~\eqref{eq:Y-bound} and~\eqref{eq:cbad}, we have $\lvert Y \cap \gc(r) \rvert \ge \lvert Y \rvert (1 - 1 / k)$. By definition of $r'$-distributed (Definition~\ref{def:distributed-marking}), each good clause has at least $r' (k-3)$ variables in $V$. As there are at most $\lvert Y \rvert - 1 + c$ edges in $G_\Phi$ joining clauses in $Y$, see~\eqref{eq:errors:tree-excess}, and two distinct clauses only share at most two variables by Lemma~\ref{lem:linearity}, we have
  \begin{equation*}
    \begin{aligned}
     \lvert \var(Y) \cap V \rvert & \ge r'(k - 3) \left(1 - \frac{1}{k}\right) \lvert Y \rvert - 2 (\lvert Y \rvert + c - 1) \\
      & \ge (r'(k - 4) - 2) \lvert Y \rvert - 2(c-1).
    \end{aligned}
  \end{equation*}
  There is $k_0 \ge k_1$ such that for $k \ge k_0$, we find that, for any set of clauses $Y$ that is connected in $G_{\Phi}$ and has $a \log(n) \le \lvert Y \rvert \le b \log (n)$, $\lvert \var(Y) \cap V \rvert \geq  (r' - \hat\delta/2) k \lvert Y \rvert - 2(c-1)$.
  Therefore, using $2(c-1) = O(1)$, for large enough $n$ we conclude that $\lvert \var(Y) \cap V \rvert \ge  (r' - \hat\delta) k \lvert Y \rvert$ and the result follows.
\end{proof}

\subsection{Proof of Lemma~\ref{lem:cc-general}} \label{sec:errors:proof}

We use the following result of~\cite{galanis2019counting} on the number of connected sets of clauses in $G_\Phi$.

\begin{lemma}[{\cite[Lemma 8.6]{galanis2019counting}}] \label{lem:bound-Z}
    Let $\alpha > 0$. W.h.p. over the choice of $\Phi = \Phi(k, n, \lfloor \alpha n \rfloor)$, for any clause $c$, the number of connected sets of clauses in $G_\Phi$ with size $\ell \ge \log n$ containing $c$ is at most $(9 k^2 \alpha)^\ell$.
\end{lemma}

We can now complete the proof of Lemma~\ref{lem:cc-general}. Recall that we will apply this result with $r = r_0-\delta$ or $r = r_1-\delta$, where $\delta = \deltadef$.

\begin{lemccgeneral}
\statelemccgeneral
\end{lemccgeneral}
\begin{proof} \label{lem:cc-general:proof}
 We apply Lemma~\ref{lem:bound-marked} with our choices of $b$ and with $\hat\delta=\delta$ to conclude that, w.h.p. over the choice of $\Phi$, for every connected set of clauses $Z \subseteq \mathcal{C}$ we have
 \begin{equation} \label{eq:num-marked}
     \text{if } \ a \log(n) \le \lvert Z \rvert \le b \log(n), \ \text{ then } \  \lvert \var(Z) \cap V \rvert \ge r k \lvert Z \rvert.
 \end{equation}
 We also need the following result on random $k$-CNF formulae. For each clause $c \in \mathcal{C}$, let 
 \begin{equation*}
     \mathcal{Z}(c, L) = \{Z \subseteq \mathcal{C} : c \in Z, Z \text{ is connected in } G_\Phi, \lvert Z \rvert = L\}.
 \end{equation*}
 Then, w.h.p. over the choice of $\Phi$, Lemma~\ref{lem:bound-Z} shows that, as long as $L \ge \log n$,
 \begin{equation} \label{eq:number-Y}
     \text{for any clause }c\in\mathcal{C} \text{ we have } \left| \mathcal{Z}(c, L)\right| \le (9 k^2 \alpha)^L.
 \end{equation}
The facts that we have just established using Lemma~\ref{lem:bound-marked} and Lemma~\ref{lem:bound-Z} are all the properties of random formulae that we need in this proof. The hypothesis $\alpha \le \Delta_r$ is used when calling Lemma~\ref{lem:bad} in the proof of  Lemma~\ref{lem:bound-marked}.

 Let $L$ be an integer with $a \log n \le L \le b \log n$. First, we are going to fix $S \subseteq V$ with $\lvert S \rvert = \rho$ and study the event $\mathcal{F}$ described in the statement. For $c \in \mathcal{C}$ and $Z \in \mathcal{Z}(c, L)$, we denote by $\mathcal{E}_1(Z, S)$ the event that $Z \subseteq \mathcal{C}^{\Lambda}$, where $\Lambda$ is drawn from $\restr{\mu}{V \setminus S}$, see Definition~\ref{def:marginal}. Recall that $Z \subseteq \mathcal{C}^{\Lambda}$ means that none of the clauses in $Z$ are satisfied by the assignment $\Lambda$ (Definition~\ref{def:subformula}). We note that $\mathcal{F} = \bigcup_{c \in \mathcal{C}, Z \in \mathcal{Z}(c, L)} \mathcal{E}_1(Z, S)$. We are going to show that, for large enough $n$,
 \begin{equation} \label{eq:claim-cc}
   \pr_{S \sim \tau} \left( \ \pr_{\Lambda \sim \restr{\mu}{V \setminus S}} \left(  \bigcup\nolimits_{c \in \mathcal{C}, Z \in \mathcal{Z}(c, L)} \mathcal{E}_1(Z, S) \right) > 2^{- \delta k L} \ \right) \le 2^{- \delta k L},
 \end{equation}
 which is equivalent to the result stated in this lemma. We note that the left-hand side of~\eqref{eq:claim-cc} can be upper bounded by 
 \begin{equation} \label{eq:claim-cc-sum}
 \begin{aligned}
   & \pr_{S \sim \tau} \left(  \exists c \in \mathcal{C}, Z \in \mathcal{Z}(c, L) :   \pr_{\Lambda \sim \restr{\mu}{V \setminus S}} \left( \mathcal{E}_1(Z, S) \right) > \frac{2^{- \delta k L}}{\lvert \mathcal{C} \rvert \cdot \lvert \mathcal{Z}(c, L) \rvert } \right) \le \\
   & \sum_{c \in \mathcal{C}, Z \in \mathcal{Z}(c, L)} \pr_{S \sim \tau} \left(    \pr_{\Lambda \sim \restr{\mu}{V \setminus S}} \left( \mathcal{E}_1(Z, S) \right) > \frac{2^{- \delta k L}}{\lvert \mathcal{C} \rvert \cdot \lvert \mathcal{Z}(c, L) \rvert }  \right).
   \end{aligned}
 \end{equation}
 We are going to show that, for any $c \in \mathcal{C}$ and $Z \in \mathcal{Z}(c, L)$, 
  \begin{equation} \label{eq:individual-Z-bound}
    \pr_{S \sim \tau} \left(  \pr_{\Lambda \sim \restr{\mu}{V \setminus S}} \left( \mathcal{E}_1(Z, S) \right) > \frac{2^{- \delta k L}}{\lvert \mathcal{C} \rvert \cdot \lvert \mathcal{Z}(c, L) \rvert }  \right) \le \left( 2 e k \cdot  2^{-r k} \right)^L.
 \end{equation}
 Before proving~\eqref{eq:individual-Z-bound}, let us complete the proof assuming that this inequality holds. In light of~\eqref{eq:number-Y}, we have $\lvert \mathcal{Z}(c, L) \rvert \le ( 9 k^2 2^{(r-2\delta) k} )^L$. We use the following observation,
 \begin{equation} \label{eq:bound-C}
  \text{for } k > 1/(\delta \log 2) \text{ and for large enough } n \text{, } \ \ \lvert \mathcal{C} \rvert \le n \alpha \le n^{\delta k^5  \log 2} \le 2^{(\delta /2) k L}.
 \end{equation}
 Combining~\eqref{eq:claim-cc-sum},~\eqref{eq:individual-Z-bound} and~\eqref{eq:bound-C}, we conclude that, for large enough $k$, the left-hand size of~\eqref{eq:claim-cc} is bounded above by    
\[
   \sum_{c \in \mathcal{C}, Z \in \mathcal{Z}(c, L)} \left( 2 e k \cdot 2^{-r k } \right)^{L} 
   \le n \density \cdot \left( 9 k^2 2^{(r-2\delta) k } \right)^L \cdot \left( 2 e k \cdot 2^{- r k } \right)^{L} 
   = n \density \left(18 e k^3 2^{- 2\delta k}\right)^L  \le 2^{-\delta k L},
\]
which completes the proof of~\eqref{eq:claim-cc}, and hence the proof of the lemma, subject to~\eqref{eq:individual-Z-bound}.

To prove~\eqref{eq:individual-Z-bound}, we are going to find many~$S$ for which  $\pr_{\Lambda \sim \restr{\mu}{V \setminus S}} ( \mathcal{E}_1(Z, S) ) \le 2^{- \delta k L} / (\lvert \mathcal{C} \rvert \cdot \lvert \mathcal{Z}(c, L) \rvert)$ holds. 
With this in mind, we introduce an event that may occur when sampling $S$:
\begin{equation} \label{eq:E2:def}
\begin{aligned}
  \mathcal{E}_2(Z) := & `` \text{the random set } S \subseteq V \text{ that we select contains fewer} \\ 
  & \text{than }\ell := \lceil \lvert \var(Z) \cap V \rvert / k \rceil \text{ variables in } \var(Z) \cap V".
\end{aligned}
\end{equation}
We will show (in equation~\eqref{eq:E2}) that the event $\mathcal{E}_2(Z)$ holds for most choices of $S$. Before proving this claim, let us assume that  $\mathcal{E}_2(Z)$ holds for $S$ and let us prove that $\pr_{\Lambda \sim \restr{\mu}{V \setminus S}} ( \mathcal{E}_1(Z, S) ) \le 2^{- \delta k L} / (\lvert \mathcal{C} \rvert \cdot \lvert \mathcal{Z}(c, L) \rvert)$. If there are $c_1, c_2 \in Z$ and $v \in \var(c_1) \cap \var(c_2) \cap (V\setminus S)$ such that  $c_1 \ne c_2$ and the literal of $v$ in $c_1$ is the negation of the literal of $v$ in $c_2$, then at least one of $c_1$ and $c_2$ is satisfied by the assignment $\Lambda \colon V \setminus S \to \{\mathsf{F}, \mathsf{T}\}$.  In this case we have $\pr_{\Lambda \sim \restr{\mu}{V \setminus S}} ( \mathcal{E}_1(Z, S) ) = 0$. Let us now consider the complementary case:
 \begin{equation} \label{eq:complementary-case}
 \begin{aligned}
      &\text{for all $c_1, c_2 \in Z$ with $c_1 \ne c_2$ and $v \in \var(c_1) \cap \var(c_2) \cap (V\setminus S)$,} \\ 
      &\text{the literal of $v$ in $c_1$ is the same as the literal of $v$ in $c_2$.}
 \end{aligned}
 \end{equation}
 In this setting, we call $\omega(v)$ the value of $v$ that does not satisfy the clauses in $Z$ that contain $v$. Note that $\omega(v)$ is well-defined by assumption~\eqref{eq:complementary-case}. Let $u_1, u_2, \ldots, u_t$ be the list of variables in $(\var(Z) \cap V) \setminus S$. We denote by $\mathcal{W}_j$ the event that $u_j$ is assigned the value $\omega(u_j)$ by $\Lambda$ when sampling $\Lambda \sim \restr{\mu}{V \setminus S}$.  Then, by definition of $\mathcal{W}_j$, we have 
 \begin{align*}
     \pr_{\Lambda \sim \restr{\mu}{V \setminus S}}\left( \mathcal{E}_1(Z, S) \right)  
     & =  \prod_{j = 1}^t  \pr_{\Lambda \sim \restr{\mu}{V \setminus S}}\left( \mathcal{W}_j \left| \  \bigcap\nolimits_{i = 1}^{j-1} \mathcal{W}_i \right.\right).
 \end{align*}
 As $\mu$ is $(1/k)$-uniform, we find that $\pr_{\Lambda \sim \restr{\mu}{V \setminus S}}( \mathcal{W}_j \vert \  \bigcap\nolimits_{i = 1}^{j-1} \mathcal{W}_i ) \le (1/2) \exp(1/k)$ for all $j \in \{1, 2, \ldots, t\}$. We conclude that
  \begin{align*}
     \pr_{\Lambda \sim \restr{\mu}{V \setminus S}} \left(  \mathcal{E}_1(Z, S)\right)  
     & \le  \left( \frac{1}{2} \exp\big(\frac{1}{k}\big) \right)^{t}.
  \end{align*}
   From~\eqref{eq:num-marked} and the fact that $\mathcal{E}_2(Z)$ holds for $S$, we have $$t = \lvert \var(Z) \cap (V \setminus S) \rvert \ge \lvert \var(Z) \cap V \rvert - \lceil \lvert \var(Z) \cap V\rvert / k \rceil  \ge \lvert \var(Z) \cap V \rvert (1 - 1/k)  -1 \ge r L (k-1) - 1.$$ It follows that
  \begin{align*}
     \pr_{\Lambda \sim \restr{\mu}{V \setminus S}} \left( \mathcal{E}_1(Z, S) \right) 
     & \le \left(\frac{1}{2} \exp\Big(\frac{1}{k}\Big)\right)^{ r (k-1) L-1}  \\
     & \le 2 \left( 2 \cdot 2^{-r k } \exp\Big(\frac{ r (k-1) }{k}\Big) \right)^{L} \\
     & \le \left( 4 e \cdot 2^{-r k } \right)^{L},
 \end{align*}
where we used that  $1/2 \le (1/2) \exp(1/k) < 1$ in the second and third inequality. For large enough $k$, we find that
\begin{equation}
     \left( 4 e \cdot 2^{-r k } \right)^{L} = \left(\frac{ 9 \cdot 4 e k^2 \cdot  \alpha \cdot 2^{-r k }}{9 k^2 \alpha} \right)^{L}   \le  \left(\frac{  9 \cdot 4 e k^2 \cdot  2^{-2\delta k}}{9 k^2 \alpha} \right)^{L} \le  \frac{2^{-(3/2) \delta k L}}{\lvert \mathcal{Z}(c, L) \rvert } \leq  \frac{2^{-\delta k L}}{ \lvert \mathcal{C} \rvert \cdot \mathcal{Z}(c, L) \rvert } ,
\end{equation}
where in the second to last inequality we applied $9 \cdot 4 e k^2  \le 2^{(\delta/2)k}$ and the bound on the size of $\mathcal{Z}(c, L)$ given in~\eqref{eq:number-Y}, and in the last inequality we used \eqref{eq:bound-C}. As $S$ was picked as any subset of $V$ with $\lvert S \rvert = \rho$ such that $\mathcal{E}_2(Z)$ holds, it follows that 
 \begin{equation} \label{eq:cc-by-E2}
   \pr_{S \sim \tau} \left(  \pr_{\Lambda \sim \restr{\mu}{V \setminus S}} \left( \mathcal{E}_1(Z, S) \right) > \frac{2^{- \delta k L}}{\lvert \mathcal{C} \rvert \cdot \lvert \mathcal{Z}(c, L) \rvert } \right) \le  \pr_{S \sim \tau} \left( \overline{\mathcal{E}_2(Z)} \right).
 \end{equation}
In order to prove~\eqref{eq:individual-Z-bound}, which finishes the proof,
we need to show $\pr_{S \sim \tau} \left( \overline{\mathcal{E}_2(Z)} \right)
\leq (2ek \cdot 2^{-rk})^L$.
The probability of $\overline{\mathcal{E}_2(Z)}$ can be bounded as follows. Recall that $\lvert S \rvert = \rho$. If $\rho <  \ell$, then, by the definition of $\mathcal{E}_2(Z)$ in~\eqref{eq:E2:def}, we obtain $\pr_{S \sim \tau} (\mathcal{E}_2(Z) ) = 1$. Otherwise, the number of choices of $S$ (with $\lvert S \rvert = \rho$) such that $\lvert S \cap \var(Z) \cap V \rvert \ge \ell$ is at most $\binom{\lvert \var(Z) \cap V \rvert}{\ell} \binom{\lvert V \rvert - \ell}{\rho -  \ell}$. Hence, we have
 \begin{equation*} \label{eq:prob-E2}
   \begin{aligned}
     \pr_{S \sim \tau} \left( \overline{\mathcal{E}_2(Z)} \right) & \le \binom{\lvert V \rvert }{\rho }^{-1}  \binom{\lvert \var(Z) \cap V \rvert}{\ell}  \binom{\lvert V \rvert  - \ell}{\rho  -  \ell} \\ 
     & = \frac{\rho (\rho -1) \cdots (\rho - \ell + 1)}{\lvert V \rvert  (\lvert V \rvert - 1) \cdots (\lvert V \rvert - \ell + 1)}   \binom{\lvert \var(Z) \cap V \rvert}{\ell} \\ 
     & \le \left(\frac{\rho}{\lvert V \rvert} \right)^{\ell} \left(\frac{e \lvert \var(Z) \cap V \rvert }{\ell} \right)^{\ell} \le \left(\frac{\rho}{\lvert V \rvert} e k \right)^{\ell},
    \end{aligned}
 \end{equation*}
 where we used $\ell := \lceil \lvert \var(Z) \cap V \rvert / k \rceil \ge \lvert \var(Z) \cap V \rvert / k$, $(p - i)/(q -  i)  \le p/q $ for any $0 < i < p < q$ and $\binom{p}{q} \le (e p / q)^q$. Combining this with the hypothesis $\rho \le \lvert V \rvert / 2^{k}$ and the bound $\ell \ge r L$, see~\eqref{eq:num-marked}, we obtain
 \begin{equation} \label{eq:E2}
     \pr_{S \sim \tau} \left( \overline{\mathcal{E}_2(Z)} \right) \le \left( e k 2^{-k} \right)^\ell \le \left( ( e k)^{r} \cdot  2^{-r k} \right)^L \le \left(2  e k \cdot  2^{-r k} \right)^L.
 \end{equation}
 The bound~\eqref{eq:individual-Z-bound} follows from combining~\eqref{eq:cc-by-E2} and~\eqref{eq:E2}, which completes the proof.
\end{proof}

\section{Sampling from small connected components}\label{sec:sample}

In this section we prove Lemma~\ref{lem:sample}. Recall that Lemma~\ref{lem:sample} claims the existence of a procedure to sample from marginals of the uniform distribution on the satisfying assignments of $\Phi^{\Lambda}$ when the connected components of $G_{\Phi^{\Lambda}}$ have small size. Here we make this procedure explicit. Our algorithm exploits the fact that the tree-excess of logarithmic-sized subsets of $G_{\Phi}$ is bounded by a constant depending only on $k$, see Lemma~\ref{lem:tree-excess}, and the fact that when $G_\Phi$ is acyclic, we can exactly count and sample satisfying assignments efficiently via a dynamic programming algorithm (Proposition~\ref{prop:sample-trees}). 

\begin{proposition} \label{prop:sample-trees}
   There is an algorithm that, for any $k$-CNF formula $\Phi = (\mathcal{V}, \mathcal{C})$ such that $G_\Phi$ is a tree, computes the number of satisfying assignments of $\Phi$ in time $O(4^k \lvert \mathcal{C} \rvert)$.
\end{proposition}
\begin{proof} \label{prop:sample-trees:proof}
  We give an algorithm based on dynamic programming. Let us fix a vertex / clause $c$ of $G_\Phi$ as the root and consider the corresponding directed tree structure $T := (G_\Phi, c)$. For any clause $c'$ of $\Phi$, let $T_{c'}$ be the subtree of $T$ hanging from $c'$. For any assignment $\sigma \colon \var(c') \to \{\mathsf{F}, \mathsf{T}\}$, let $\operatorname{sa}(c', \sigma)$ denote the number of satisfying assignments of the formula determined by $T_{c'}$ that extend $\sigma$. Our goal is computing the number of satisfying assignments of $\Phi$, which, under this notation, is equal to
  \begin{equation} \label{eq:tree}
     \operatorname{sa}(\Phi) := \sum_{\sigma \colon \var(c) \to \{\mathsf{F}, \mathsf{T}\}}\operatorname{sa}(c, \sigma).
  \end{equation}
  We do this by computing $\operatorname{sa}(c', \sigma)$ for any clause $c'$ and any assignment $\sigma  \colon \var(c') \to \{\mathsf{F}, \mathsf{T}\}$. Using the tree structure of $T$, we show that $\operatorname{sa}(c', \sigma)$ satisfies a recurrence. There are two cases:
  \begin{enumerate}
    \item $c'$ is a leaf. Then $\operatorname{sa}(c', \sigma) = 1$ if $c'$ is satisfied by $\sigma$ and $0$ otherwise.
    \item $c'$ is not a leaf. Let $T_1, \ldots, T_l$ be the trees hanging from $c'$ in $T$ and let $c_1, \ldots, c_l$ be their roots. Then, since $T_1, \ldots, T_l$ do not share variables as $G_\Phi$ is acyclic, we have 
    \begin{equation*}
    \operatorname{sa}(c', \sigma) = \prod_{j = 1}^l \sum_{\ \ \tau \in A(c_j, \sigma)}\operatorname{sa}(c_j, \tau),
    \end{equation*}
    where $A(c_j, \sigma)$ is the set of assignments of the variables in $\var(c_j)$ that agree with $\sigma$ on $\var(c') \cap \var(c_j)$. 
  \end{enumerate}
  We can apply this recurrence with dynamic programming to compute $\operatorname{sa}(c, \sigma)$ for any assignment $\sigma \colon \var(c) \to \{\mathsf{F}, \mathsf{T}\}$. More explicitly, we compute $\operatorname{sa}(c', \sigma)$ by levels of the tree, starting from the deepest level, where all nodes are leaves, and ending at the root $c$. This involves computing at most $2^k$ entries $\operatorname{sa}(c', \cdot)$ per clause $c'$ of $\Phi$. After computing all the entries appearing in this recurrence, we compute the number of satisfying assignments of $\Phi$, $\operatorname{sa}(\Phi)$, as in equation~\eqref{eq:tree}. The overall procedure takes at most $O(4^k \lvert \mathcal{C} \rvert)$ steps since each entry $\operatorname{sa}(c', \sigma)$ is accessed at most $2^k$ times when computing the corresponding entries for the parent of $c'$, and there are at most $2^k \lvert \mathcal{C}(T) \rvert$ entries.
\end{proof}

In Algorithm~\ref{alg:counting} we give an algorithm based on Proposition~\ref{prop:sample-trees} to count satisfying assignments of a $k$-CNF formula. Recall the folklore fact that if we can count satisfying assignments then we can sample from the marginal of $\mu_\Omega$ on $v$ by counting the satisfying assignments of $\Phi^{v \mapsto \mathsf{F}}$ and $\Phi^{v \mapsto \mathsf{T}}$.

\begin{algorithm}[H]
  \begin{algorithmic}[1]

    \caption{Counting satisfying assignments via trees} \label{alg:counting}

    \REQUIRE a $k$-CNF formula $\Phi = (\mathcal{V},\mathcal{C})$
    \ENSURE The number of satisfying assignments of $\Phi$.

    \STATE Find a spanning forest $T$ of $G_{\Phi}$. 
    
    \STATE Let $\mathcal{V}_{T}$ be the set of variables that gives rise to edges of $G_{\Phi}$ that are not in $T$.

    \STATE $count \leftarrow 0$.

    \FORALL{$\Lambda \colon \mathcal{V}_{T} \to \{\mathsf{F}, \mathsf{T}\}$}
    
    \STATE Note that the graph $G_{\Phi^{\Lambda}}$ is acyclic. Hence, we can count the number of satisfying assignments of $\Phi^{\Lambda}$ in time $O(4^k \lvert \mathcal{C}(\Phi^{\Lambda}) \rvert)$ by applying Proposition~\ref{prop:sample-trees} to each connected component of $G_{\Phi^{\Lambda}}$ and taking the product of the numbers obtained. Let $\operatorname{sa}(\Phi^\Lambda)$ be the result of this computation.
    
    \STATE $count \leftarrow count + \operatorname{sa}(\Phi^\Lambda)$.
    
    \ENDFOR
    
    \RETURN $count$
  \end{algorithmic}
\end{algorithm}

\begin{proposition} \label{prop:counting}
  Let $\Phi = (\mathcal{V}, \mathcal{C})$ be a $k$-CNF formula and let $c$ be the tree-excess of $G_\Phi$. Then Algorithm~\ref{alg:counting} counts the number of satisfying assignments of $\Phi$ in time $O(2^{k(c+2)} \lvert \mathcal{C} \rvert)$.
\end{proposition}
\begin{proof} \label{prop:counting:proof}
  We note that, in the execution of Algorithm~\ref{alg:counting}, we have $\lvert \mathcal{V}_T \rvert \le k c$. Hence, there are at most $2^{k c}$ iterations of the for loop and each one takes $O(4^k \lvert \mathcal{C} \rvert)$ steps, so the running time follows. The fact that the algorithm is correct follows from the correctness of the procedure presented in Proposition~\ref{prop:sample-trees}.
\end{proof}

Even though the running time of Algorithm~\ref{alg:counting} is not polynomial in the size of the formula $\Phi$ (in fact, it is exponential in general), we obtain linear running time when the formulae considered have constant tree-excess. As shown in Lemma~\ref{lem:tree-excess}, this is the case for logarithmic-sized subsets of clauses of random formulae. We can now finish the proof of Lemma~\ref{lem:sample}.

\begin{lemsample}
  \statelemsample
\end{lemsample}
\begin{proof} \label{lem:sample:proof}
  We apply Lemma~\ref{lem:tree-excess}, so, w.h.p. over the choice of $\Phi = \Phi(k, n, \lfloor \alpha n \rfloor)$, any connected set  of clauses in  $G_{\Phi}$ with size at most $b \log(n)$  has tree-excess at most $c = \max\{1, 2 b \log(e \density k^2)\} = O(1)$. First, we give an algorithm for the case $\lvert S \rvert = 1$. Let $\Phi$, $V$ and $\Lambda$ as in the statement, and let $S = \{v\}$. Let $H$ be the connected component of the clauses that contain $v$ in $G_{\Phi^\Lambda}$, and let $\Phi' = (\mathcal{V}', \mathcal{C}')$ be the subformula of $\Phi^{\Lambda}$ with $G_{\Phi'} = H$. The formula $\Phi'$ has size at most $b \log(n)$. Moreover, the graph $G_{\Phi'} = H$ has tree-excess at most $c$ as $H$ is a subgraph of $G_{\Phi}$ with size at most $b \log(n)$. Thus, we can apply Proposition~\ref{prop:counting} to count the number of satisfying assignments of $\Phi'^{v \to \mathsf{F}}$ and $\Phi'^{v \to \mathsf{T}}$ in time $O( 2^{k(c+2)} \lvert \mathcal{C}'\rvert) = O(\log n)$. Let these numbers be $t_0$ and $t_1$ respectively. It is straightforward to use $t_0$ and $t_1$ to sample from the marginal of the distribution $\mu_{\Omega^{\Lambda}}$ for $v$; we only have to sample an integer $t \in [0, t_0 + t_1)$ and output $\mathsf{F}$ if $t < t_0$ and $\mathsf{T}$ otherwise. The whole process takes time $O(\log n)$.
  
  Finally, we argue how to extend this algorithm to the case $\lvert S \rvert > 1$. For this, first, we give an order to the variables in $S$, say $u_1, u_2, \ldots, u_\ell$. We  then call the algorithm described in the paragraph above once for each variable in $u_1, u_2, \ldots, u_\ell$. The inputs of the algorithm in the $j$-th call are the variable $u_j$ and the assignment $\Lambda_j = \Lambda \cup \tau_{j-1}$, where $\tau_{j-1}$ is the assignment obtained in the previous calls for $u_{1}, \ldots, u_{j-1}$. After this process, $\tau_\ell$ is an assignment of all the variables in $S$ that follows the distribution $\restr{\mu_{\Omega^\Lambda}}{S}$. This assignment has been computed in $O(\lvert S \rvert \log n)$ steps as we wanted.
\end{proof}

\section{Mixing time of the Markov chain} \label{sec:mixing-time}

In this section we study the mixing time of the $\rho$-uniform-block Glauber dynamics on the marked variables and prove Lemma~\ref{lem:mixing-time}. As explained in Section~\ref{sec:po:si}, in order to conclude rapid mixing of this Markov chain we apply the spectral independence framework, which has recently been extended to the $\rho$-uniform-block Glauber dynamics~\cite{Chen2021}. Traditionally in path coupling or spectral independence arguments one has to bound a sum of influences by a constant in order to obtain rapid mixing of the single-site Glauber dynamics. However, due to the presence of high-degree variables, an $O(1)$ upper bound seems unattainable in the random $k$-SAT formula setting; in the worst case paths of high-degree variables may significantly affect influences. This seems also to be the case for other random models, such as the hardcore model on random graphs~\cite{bez2021}. Here we show that that sums of influences are at most $\epsilon \log n$ for small $\epsilon$ (Lemma~\ref{lem:si}). Even though this is generally not enough to conclude rapid mixing of the single-site Glauber dynamics, it turns out to be enough to conclude rapid mixing of the $\rho$-uniform-block Glauber dynamics for $\rho = \Theta(n)$. An essential ingredient in our argument is exploiting the auxiliary variables in introduced in Section~\ref{sec:marking}. Therefore, in this section we will work with $r= r_0 -\delta$ and a $(r, r_0, r_0, 2r_0)$-marking $(\mathcal{V}_{\mathrm m}, \mathcal{V}_{\mathrm a}, \mathcal{V}_{\mathrm c})$. Since $r$ is fixed, we drop it from the notation and write, for instance, $\gv$ instead of $\gv(r)$.

This section is divided as follows. In Section~\ref{sec:po:si:pw}, we explain why bounded-degree methods to bound the mixing time of the Glauber dynamics fail to generalise from the bounded-degree $k$-SAT model to the random $k$-SAT model. In Section~\ref{sec:mixing-time:si} we prove Lemmas~\ref{lem:expectation} and~\ref{lem:si}.  In Section~\ref{sec:mixing-time:mt} we prove Lemma~\ref{lem:mixing-time}.

\subsection{Previous work on  the Glauber dynamics for bounded-degree $k$-SAT formulae} \label{sec:po:si:pw}

In this section we explain why previously known arguments for showing rapid mixing of the Glauber dynamics on bounded-degree $k$-SAT formulae do not extend to the random $k$-SAT model. This section is not used in our work and may be skipped by a reader who just wants to understand our approach and result. The best result currently known on bounded-degree formulae is~\cite{jain2021sampling}, where the authors show, for large enough $k$, how to efficiently sample satisfying assignments of $k$-CNF formulae in which their variables have maximum degree  $\hat\Delta \le C \, 2^{0.1742 \cdot k}/ k^3$, where $C > 0$ is a constant that does not depend on $k$.\footnote{In~\cite{jain2021sampling} the maximum degree $\hat\Delta$ of $\Phi$ is defined as the maximum over $c \in \mathcal{C}$ of the number of clauses that share a variable with $c$. Under this definition of $\hat\Delta$, their result holds for $\hat\Delta \le C 2^{0.1742 \cdot k}/ k^2$.} Their result actually holds in the more general setting of atomic constrain satisfaction problems (albeit with a different bound on $\hat\Delta$). As part of their work, they show that the single-site Glauber dynamics on a set of marked variables mixes quickly. Their argument is restricted to atomic CSPs with bounded-degree and strongly exploits the properties of the Glauber dynamics in this setting. They study the optimal coupling of the single-site Glauber dynamics, we refer to~\cite{pcbook} for the definition of coupling of Markov chains. In such a coupling the goal is to show that two copies of the chain starting from truth assignments differing in at least a marked variable (a so-called discrepancy) can be coupled in a small number of steps. Here it is crucial that the marginals of the marked variables are near $1/2$, so the optimal coupling generates new discrepancies with small probability. At this stage, the high-level idea to conclude rapid mixing of the Glauber dynamics is bounding the probability that the dynamics has not coupled by a product of probabilities, each corresponding to the event that a clause is unsatisfied at a certain time, and aggregating over all possible discrepancy sequences. 

The fundamental observation in~\cite{jain2021sampling}, based on the work on monotone $k$-CNF formulae presented in~\cite{hermon2019}, is that if there is an update of a marked variable that generates a discrepancy in the chains, then there is another marked variable where the chains disagree that is connected to the former variable through a path of clauses, where consecutive clauses in the path share at least a variable. Moreover, each one of the clauses in this path is unsatisfied by at least one of the two copies of the chain. As a consequence, from a discrepancy at time $t$ one can find a sequence of discrepancies going back to time $0$, and  these discrepancies are joined by a path of clauses. Thus, the union bound over discrepancy sequences is essentially a union bound over paths of clauses with a particular time structure, where the same clause can be appear in the path several times. Extending this idea to the random $k$-SAT model presents two main issues. First of all,  the number of discrepancy sequences of any given length may be too large due to the presence of bad clauses and the fact that they can repeatedly appear in the  sequence. Moreover, it may be the case that these discrepancy sequences mostly consist of bad clauses, which are always unsatisfied in both chains and, thus, the probability that they are unsatisfied is not small. Interestingly, similar issues arise when directly extending the bounded-degree approach based on the coupling process of~\cite{moitra19, feng2020} to our setting. In~\cite{feng2020} the mixing time argument only succeeds when $\hat\Delta \le 2^{k/20} / (8k)$ and is also based on a union bound over path of clauses that are unsatisfied or contain discrepancies after running a coupling process. However, very importantly, these paths of clauses are simple (clauses are not repeated) and the combinatorial structures appearing in the coupling process are less complex than the discrepancy sequences of~\cite{jain2021sampling}. This allowed the authors of~\cite{galanis2019counting} to exploit the expansion properties of random $k$-CNF formulae to analyse the coupling process of~\cite{moitra19} on the random setting. Here we incorporate novel ideas to the work of~\cite{galanis2019counting} in order to obtain a tighter analysis that leads to nearly linear running time of our sampling algorithm.

We conclude this section by briefly introducing notation and standard results on couplings that we need in the next section. Let $\mu$ and $\nu$ be two distributions over the same space $\widehat\Omega$. A coupling $\tau$ of $\mu$ and $\nu$ is a joint distribution over $\widehat\Omega \times \widehat\Omega$ such that the projection of $\tau$ on the first coordinate is $\mu$ and the projection on the second coordinate is $\nu$.  Recall that the total variation distance of $\mu$ and $\nu$ is defined by 
$d_{\mathrm{TV}}(\mu, \nu) = \frac{1}{2} \sum_{x \in \widehat\Omega} \lvert \mu(x) - \nu(x) \rvert$.
If a random variable $X$ has distribution $\mu$, we also write $d_{\mathrm{TV}}(X, \nu)$ 
to mean $d_{\mathrm{TV}}(\mu, \nu)$. An important property of couplings is the coupling lemma.

\begin{proposition}[Coupling lemma] \label{prop:coupling-lemma}
Let $\tau$ be a coupling of $\mu$ and $\nu$. Then $d_{\TV}(\mu, \nu) \le \pr_{(X,Y) \sim \tau} (X \ne Y)$. Moreover, there exists a coupling that achieves equality.
\end{proposition}

The coupling $\tau$ of $\mu$ and $\nu$ that minimises $\pr_{(X,Y) \sim \tau}(X \ne Y)$ is called \emph{optimal}. 
Let us now assume that $\mu$ and $\nu$ are Bernoulli distributions with parameters $0 \le p \le q \le 1$ respectively, so $\pr_\mu(X=1) = p$ and $\pr_\nu(Y=1) = q$. The \emph{monotone coupling} $\tau$ of $\mu$ and $\nu$ is defined as follows. We pick $U$ uniformly at random in $[0,1]$ and set $X = 1$ only when $U \le p$ and $Y = 1$ only when $U \le q$. For this coupling we have $\pr_{(X,Y) \sim \tau }(X \ne Y) = q-p = d_{\mathrm{TV}}(X, Y)$ and, hence, the monotone coupling is optimal. This optimal coupling will come up in the coupling process when sampling from the marginals of auxiliary variables.

\subsection{Showing spectral independence in the $k$-SAT model} \label{sec:mixing-time:si}

In this section we prove Lemma~\ref{lem:si} about spectral independence and bounding the sum of influences. Throughout this section we fix a $k$-CNF formula~$\Phi$ and a $(r, r_0, r_0, 2r_0)$-marking $(\mathcal{V}_{\mathrm m}, \mathcal{V}_{\mathrm a}, \mathcal{V}_{\mathrm c})$ of $\Phi$ (cf. Definition~\ref{def:distributed-marking} and Proposition~\ref{prop:marginals}).

To briefly overview our approach, we bound the sum of influences of marked variables using the coupling process technique that is by now standard in the literature~\cite{galanis2019counting, moitra19, feng2020}. The new ingredient here is introducing auxiliary variables in the coupling process and exploiting the sparsity properties of logarithmic-sized sets of clauses, which allows us to conclude a $2^{-r_0 k} \log n$ spectral independence bound. The key idea is that if we progressively extend two assignments $X$ and $Y$ on auxiliary variables following the optimal coupling, with high probability over $X$ and $Y$, at some point the formulae $\Phi^X$ and $\Phi^Y$ factorise in small connected components, most of which appear in both  $\Phi^X$ and $\Phi^Y$. Then we can bound influences between marked variables by analysing the connected components where $\Phi^X$ and $\Phi^Y$ differ. 

To carry out the above plan, we begin by showing a bound on the sum of influences between marked variables in terms of a suitable coupling. Recall, given two assignments $\Lambda_1$ and $\Lambda_2$ on disjoint sets of variables, we denote by $\Lambda_1 \cup \Lambda_2$ the combined assignment on the union of their domains. 

\begin{proposition} \label{prop:influence-coupling}
    Let $u \in \mathcal{V}_{\mathrm m}$ and $\Lambda \colon S \to \{\mathsf{F}, \mathsf{T}\}$ with $S \subseteq \mathcal{V}_{\mathrm m} \setminus \{u\}$. Let $(X, Y)$ be a coupling where $X$ follows the distribution $\restr{\mu_{\Omega^{\Lambda \cup u \mapsto \mathsf{T}}}}{\mathcal{V}_{\mathrm m} }$ and $Y$ follows the distribution $\restr{\mu_{\Omega^{\Lambda \cup u \mapsto \mathsf{F}}}}{\mathcal{V}_{\mathrm m}}$. Then
    \begin{equation}  \label{eq:influence-coupling}
    \sum_{v \in \mathcal{V}_{\mathrm m} \setminus (S \cup \{u\})} \left| \mathcal{I}^\Lambda(u \to v) \right| \le \sum_{v \in \mathcal{V}_{\mathrm m} \setminus (S \cup \{u\})} \pr\left(X(v) \ne Y(v)\right).
\end{equation}
\end{proposition}
\begin{proof}
 Let $v \in \mathcal{V}_{\mathrm m}$.
 Then for any $\omega \in \{\mathsf{F}, \mathsf{T}\}$, we have $\pr(v \mapsto \omega \vert \Lambda, u \mapsto \mathsf{T}) = \pr(X(v) = \omega)$ and $\pr(v \mapsto \omega \vert \Lambda, u \mapsto \mathsf{F}) = \pr(Y(v) = \omega)$. Thus, by the coupling lemma (Proposition~\ref{prop:coupling-lemma}),
\begin{equation*} 
    \left| \mathcal{I}^\Lambda(u \to v) \right| = \left|  \pr(X(v) = \mathsf{T}) - \pr(Y(v) = \mathsf{T})  \right| = d_{\TV}(X(v), Y(v)) \le \pr\left(X(v) \ne Y(v)\right),
\end{equation*}
and the proof follows by adding over $v \in \mathcal{V}_{\mathrm m} \setminus (S \cup \{u\})$.
\end{proof}

For two assignments $X$ and $Y$ on a subset of variables $V$, we say that $X$ and $Y$ have a \emph{discrepancy} at $v \in V$ when $X(v) \ne Y(v)$. In~\cite{feng2020} the authors  bound \eqref{eq:influence-coupling} by a constant (independent of $n$) when the underlying formula has bounded degree. However, the argument there breaks under the presence of high-degree variables due to the fact that we cannot control the number of bad clauses in a path of clauses unless the path has length at least $\Omega(\log n)$. Here instead we adapt the coupling process developed in~\cite{galanis2019counting}, which accounts for the presence of bad clauses. The main difference with respect to~\cite{galanis2019counting} is that here we only couple over some auxiliary variables, which significantly simplifies the analysis and leads to tighter bounds on the influences among marked variables. Later we extend this coupling over auxiliary variable to all marked and auxiliary variables so that we can apply Proposition~\ref{prop:influence-coupling}.

To present our  coupling process, let us first describe some of the notation and structures that are used in the process. Let $u \in \mathcal{V}_{\mathrm m}$ and $\Lambda \colon S \to \{\mathsf{F}, \mathsf{T}\}$ with $S \subseteq \mathcal{V}_{\mathrm m} \setminus \{u\}$. We start with two assignments $\widehat{X}$ and $\widehat{Y}$ that have a discrepancy at $u$ and agree with $\Lambda$ on $S$. In the coupling process given in Algorithm~\ref{alg:coupling} we extend $\widehat{X}$ and $\widehat{Y}$ on  auxiliary variables and, at the same time, we identify a set of \emph{failed clauses}, denoted $\faild \cup \failu$, which we define formally in Definition~\ref{def:fail} and will include those clauses containing a discrepancy, i.e., a variable $v$ such that $\widehat{X}(v) \ne \widehat{Y}(v)$.  At each step of the process, we choose a clause $c$ and determine if $c$ is it has become a failed clause. If the clause $c$ is not failed, we extend the coupling $(\widehat{X},\widehat{Y})$ to an auxiliary variable of $c$. For ease of reading, here we present a list of the sets that appear in our algorithm and a definition of $\faild$ and $\failu$. Once a variable/clause is added to one of this sets, it stays there for the rest of the process (except for $\mathcal{C}_{\mathrm {rem}}$, where clauses are only removed, see below).
\begin{definition}[The sets of variables/clauses of the coupling process] \label{def:fail}
~
\begin{enumerate}
    \item $\verticesd$. Set of discrepancies, i.e., set of variables $v$ with $\widehat{X}(v) \ne \widehat{Y}(v)$. 
    \item $\faild$. Set of all clauses containing a variable in $\verticesd$. These are failed clauses.
    \item $\mathcal{V}_{\mathrm {set}}$. Set of variables that are assigned a value in the coupling process. At the start of the process, $\mathcal{V}_{\mathrm {set}} = S \cup \{u\}$.
    \item $\failu$. Set of clauses that have been considered at some step of the coupling process, and  are
    either bad, or have all their auxiliary variables in $\mathcal{V}_{\mathrm {set}}$ and are unsatisfied by at least one of $\widehat{X}$ and $\widehat{Y}$. These are failed clauses. Note that if a clause is in $\failu$ and is not bad, then the clause necessarily comes up in either $\Phi^{\widehat{X}}$ or $\Phi^{\widehat{Y}}$, even if we extend the coupling further on other auxiliary variables.
    \item $\mathcal{C}_{\mathrm {rem}}$. Set of clauses that have unassigned auxiliary variables or have not been explored yet by the coupling process.
\end{enumerate}
\end{definition}
It is important in our arguments that we make sure that the set of failed clauses is connected in $G_\Phi$ at all times. In order to achieve connectivity of failed clauses, at each step of the coupling process we only consider clauses that are adjacent to failed clauses in $G_\Phi$, that is, the clause $c$ chosen above must satisfy $\var(c) \cap (\var(\faild) \cup \var(\failu)) \ne \emptyset$. Our coupling process on auxiliary variables is given in Algorithm~\ref{alg:coupling}. To enhance the clarity of this exposition, we provide a high-level example (Example~\ref{ex:coupling}) that the reader can follow along  either by referring to Algorithm~\ref{alg:coupling} or solely relying on the description given in this paragraph and Definition~\ref{def:fail}.

\begin{algorithm}[H]
  \begin{algorithmic}[1]

    \caption{The coupling process on auxiliary variables}  \label{alg:coupling}

    \REQUIRE A $k$-CNF formula $\Phi = (\mathcal{V},\mathcal{C})$, an $(r, r_0, r_0, 2r_0)$-marking $\mathcal{M} = (\mathcal{V}_{\mathrm m}, \mathcal{V}_{\mathrm a}, \mathcal{V}_{\mathrm c})$, $u \in \mathcal{V}_{\mathrm m}$ and $\Lambda \colon S \to \{\mathsf{F}, \mathsf{T}\}$ with $S \subseteq \mathcal{V}_{\mathrm m} \setminus \{u\}$. 
    
    \ENSURE a pair of assignments $\widehat{X}, \widehat{Y} \colon \mathcal{V}_{\mathrm {set}} \to \{\mathsf{F}, \mathsf{T}\}$ for some set of variables $\mathcal{V}_{\mathrm {set}}$ such that: \\
    $\circ$ \ $S \cup \{u\} \subseteq \mathcal{V}_{\mathrm {set}} \subseteq S \cup \{u\} \cup \mathcal{V}_{\mathrm a}$, \\
    $\circ$ \ $\widehat{X}$ and $\widehat{Y}$ agree with $\Lambda$ on $S$, $\widehat{X}(u) = \mathsf{T}$ and $\widehat{Y}(u) = \mathsf{F}$.

    \STATE We fix two total orders $\le_{\mathcal{V}}$ and $\le_{\mathcal{C}}$ over the variables and clauses of $\Phi$. These are only relevant to have a pre-determined order in which clauses and variables are considered in this algorithm.

    \STATE Initialise $\widehat{X}$ and $\widehat{Y}$ as $\Lambda$, and set $\widehat{X}(u) = \mathsf{T}$ and $\widehat{Y}(u) = \mathsf{F}$.

    \STATE $\mathcal{V}_{\mathrm {set}} \leftarrow S \cup \{u\}$, $\verticesd \leftarrow \{u\}$, $\faild \leftarrow \{ c \in \mathcal{C} : u \in \var(c) \}$,  $\failu \leftarrow \emptyset$, $\mathcal{C}_{\mathrm {rem}} \leftarrow \mathcal{C}$.  \label{line:start}

    \WHILE{$\exists c \in \mathcal{C}_{\mathrm {rem}}: \var(c) \cap (\verticesd \cup \var(\failu)) \ne \emptyset$}

    \STATE Let $c$ be smallest clause according to $\le_{\mathcal{C}}$ with $\var(c) \cap (\verticesd \cup \var(\failu)) \ne \emptyset$.

    \IF{$c$ is a bad clause}

    \STATE Remove $c$ from $\mathcal{C}_{\mathrm {rem}}$ \label{line:rem:bad} and add $c$ to $\failu$. \label{line:un:1}

    \ENDIF

    \IF{$c$ is a good clause and $(\var(c) \cap \mathcal{V}_{\mathrm a}) \setminus \mathcal{V}_{\mathrm {set}} = \emptyset$}

    \STATE Remove $c$ from $\mathcal{C}_{\mathrm {rem}}$ (as all auxiliary variables in $c$ have been set). \label{line:rem:un}

    \IF{$c$ is unsatisfied by at least one of $\widehat{X}$ and $\widehat{Y}$ } 

    \STATE Add $c$ to $\failu$. \label{line:un:2}

    \ENDIF

    \ENDIF
    
    \IF{$c$ is a good clause and $(\var(c) \cap \mathcal{V}_{\mathrm a}) \setminus \mathcal{V}_{\mathrm {set}} \ne \emptyset$ } \label{line:set:start} 

    \STATE Let  $v$ be the smallest  variable in $(\var(c) \cap \mathcal{V}_{\mathrm a}) \setminus \mathcal{V}_{\mathrm {set}}$ (according to $\le_{\mathcal{V}}$). \label{line:v}
    
    \STATE \label{line:sample}
    Extend $\widehat{X}$ and $\widehat{Y}$ to $v$ by sampling from the optimal coupling between the marginal distributions of $\mu_{\Omega^{\widehat{X}}}$ and $\mu_{\Omega^{\widehat{Y}}}$ on $v$, and add $v$ to $\mathcal{V}_{\mathrm {set}}$.

    \IF{$\widehat{X}(v) \ne \widehat{Y}(v)$} 
    
    \STATE Add $v$ to $\verticesd$. Add all clauses containing $v$ to $\faild$. \label{line:discrepancy}

    \ENDIF \label{line:set:end} 
    
    \ENDIF

    \ENDWHILE

    \RETURN $(\widehat{X}, \widehat{Y})$.
    
  \end{algorithmic}
\end{algorithm}

\begin{example}[High-level example of an execution of the coupling process on auxiliary variables.] \label{ex:coupling} Here we assume that $S = \emptyset$, so at the start of Algorithm~\ref{alg:coupling} the assignments $\widehat{X}$ and $\widehat{Y}$ are only defined on a variable $u \in \mathcal{V}_{\mathrm m}$. Moreover, for simplicity of exposition, we assume that each good clause only has $2$ or $3$ auxiliary variables. Note that in our application of the coupling process, good clauses will have a larger number of auxiliary variables (at least $r_0(k-3)$ of them per clause with $k$ large) due to the application of LLL, see Section~\ref{sec:marking} for details. In this example we plot how the graph $\faild \cup \failu$ (as subgraph of $G_\Phi$) could evolve in a few iterations of Algorithm~\ref{alg:coupling}. The variables with names in this example have been assigned values in the coupling process.

\begin{figure}[H]
\centering
\begin{subfigure}[b]{0.48\textwidth}
\centering
\begin{tikzpicture}
\begin{scope}[every node/.style={circle,thick,draw}]
    \node [rectangle split, rectangle split horizontal, rectangle split parts = 3, draw, text centered] (A) at (0,0) {
        $u, \ldots$
        \nodepart{two} $\cdots$ 
        \nodepart{three} $\cdots$
    };
    \node [rectangle split, rectangle split horizontal, rectangle split parts = 3, draw, text centered] (B) at (0,1.5) {
        $u, \ldots$
        \nodepart{two} $\cdots$
        \nodepart{three} $\cdots$
    };
\end{scope}

\begin{scope}[>={Stealth[black]},
              every node/.style={fill=white,circle},
              every edge/.style={draw=black,very thick}]
    \path [-] (A) edge node {$u$} (B);
    \end{scope}
\end{tikzpicture}
\caption{The graph $\faild \cup \failu$ at the start of the coupling process. $\faild$ is the set of all clauses that contain the variable $u$. We assume that there are only two of them in this example and call them $c_1$ and $c_2$. $\failu$ is empty at the start, and we add clauses to $\failu$ as the execution of the coupling process goes on. The set $\mathcal{V}_{\mathrm set}$ contains only $u$, and $\mathcal{C}_{\mathrm rem} = \mathcal{C} \setminus \{c_1, c_2\}$. \newline \newline \newline}
\end{subfigure}
$\quad$
\begin{subfigure}[b]{0.48\textwidth}
\centering
\begin{tikzpicture}
\begin{scope}[every node/.style={circle,thick,draw}]
    \node [rectangle split, rectangle split horizontal, rectangle split parts = 3, draw, text centered] (A) at (0,0) {
        $u, \ldots$
        \nodepart{two} $v_1$, \ldots 
        \nodepart{three} $\cdots$
    };
    \node [rectangle split, rectangle split horizontal, rectangle split parts = 3, draw, text centered] (B) at (0,1.5) {
        $u, \ldots$
        \nodepart{two} $\cdots$
        \nodepart{three} $\cdots$
    };
\end{scope}

\begin{scope}[>={Stealth[black]},
              every node/.style={fill=white,circle},
              every edge/.style={draw=black,very thick}]
    \path [-] (A) edge node {$u$} (B);
    \end{scope}
\end{tikzpicture}
\caption{The algorithm chooses a clause $c$ of $\Phi$ connected to $\faild \cup \failu$ in $G_\Phi$, that is,  $\var(c) \cap (\verticesd \cup \var(\failu)) \ne \emptyset$. The clause in particular that we pick does not matter as long as this connectivity property holds. In this case, the algorithm picks $c_1$. Let $v_1$ be an auxiliary variable of $c_1$ that is not in $\mathcal{V}_{\mathrm set}$. Then we ``couple" $\widehat{X}$ and $\widehat{Y}$ on $v_1$, that is, we sample $(\widehat{X}(v_1), \widehat{Y}(v_1))$ following the optimal coupling between the marginal distributions of $\mu_{\Omega^{\widehat{X}}}$ and $\mu_{\Omega^{\widehat{Y}}}$ on $v_1$. Let us assume that here $\widehat{X}(v_1) = \widehat{Y}(v_1)$. Then $\faild$ does not change and $v_1$ is added to $\mathcal{V}_{\mathrm set}$.}
\end{subfigure}

\vspace{5mm}

\begin{subfigure}[b]{0.44\textwidth}
\centering
\begin{tikzpicture}
\begin{scope}[every node/.style={circle,thick,draw}]
    \node [rectangle split, rectangle split horizontal, rectangle split parts = 3, draw, text centered] (A) at (0,0) {
        $u, \ldots$
        \nodepart{two} $v_1$, $v_2$ 
        \nodepart{three} $\cdots$
    };
    \node [rectangle split, rectangle split horizontal, rectangle split parts = 3, draw, text centered] (B) at (0,1.5) {
        $u, \ldots$
        \nodepart{two} $\cdots$
        \nodepart{three} $\cdots$
    };
\end{scope}

\begin{scope}[>={Stealth[black]},
              every node/.style={fill=white,circle},
              every edge/.style={draw=black,very thick}]
    \path [-] (A) edge node {$u$} (B);
    \end{scope}
\end{tikzpicture}
\caption{We carry out another iteration of the coupling process, choosing again $c_1$ and another auxiliary variable of $c_1$, denoted $v_2$. Again, assume that when we couple $\widehat{X}$ and $\widehat{Y}$ on $v_2$ we sample $(\widehat{X}(v_2), \widehat{Y}(v_2))$ such that $\widehat{X}(v_2) = \widehat{Y}(v_2)$. Thus, the sets $\faild$ and $\failu$ do not change. In the next iteration, note that $\var(c_1) \cap \mathcal{V}_{\mathrm a} \subseteq \mathcal{V}_{\mathrm set}$, so we remove $c_1$ from $\mathcal{C}_{\mathrm rem}$ and add it to $\failu$ if it is not satisfied.}
\end{subfigure}
$\quad$
\begin{subfigure}[b]{0.51\textwidth}
\centering
\begin{tikzpicture}
\begin{scope}[every node/.style={circle,thick,draw}]
    \node [rectangle split, rectangle split horizontal, rectangle split parts = 3, draw, text centered] (A) at (0,0) {
        $u, \ldots$
        \nodepart{two} $v_1$, $v_2$ 
        \nodepart{three} $\cdots$
    };
    \node [rectangle split, rectangle split horizontal, rectangle split parts = 3, draw, text centered] (B) at (0,1.5) {
        $u, \ldots$
        \nodepart{two} $\cdots$
        \nodepart{three} $\cdots$
    };
    \node [rectangle split, rectangle split horizontal, rectangle split parts = 3, draw, text centered] (C) at (5,0) {
        $\cdots$
        \nodepart{two} $v_1, v_3, \ldots$
        \nodepart{three} $\cdots$
    };
    \node [rectangle split, rectangle split horizontal, rectangle split parts = 3, draw, text centered] (D) at (5,1.5) {
        $\cdots$
        \nodepart{two} $v_3, \ldots$
        \nodepart{three} $\cdots$
    };
\end{scope}

\begin{scope}[>={Stealth[black]},
              every node/.style={fill=white,circle},
              every edge/.style={draw=black,very thick}]
    \path [-] (A) edge node {$u$} (B);
    \path [-] (A) edge node {$v_1$} (C);
    \path [-] (C) edge node {$v_3$} (D);
    \end{scope}
\end{tikzpicture}
\caption{Now we pick a different clause $c_3$ with $\var(c_3) \cap (\verticesd \cup \var(\failu)) \ne \emptyset$. This time the clause $c_3$ is not in $\faild$, but it is connected to $\faild \cup \failu$ through the variable $v_1$. We couple $\widehat{X}$ and $\widehat{Y}$ on an auxiliary variable $v_3$ of $c_3$, this time leading to a discrepancy, that is, $\widehat{X}(v_3) \ne \widehat{Y}(v_3)$. Thus, $c_3$ is added to $\faild$, and so is any other clause containing $v_3$. In this example, we assume there are two such clauses.\newline }
\end{subfigure}

\vspace{5mm}

\begin{subfigure}[b]{\textwidth}
\centering
\begin{tikzpicture}
\begin{scope}[every node/.style={circle,thick,draw}]
    \node [rectangle split, rectangle split horizontal, rectangle split parts = 3, draw, text centered] (A) at (0,0) {
        $u, \ldots$
        \nodepart{two} $v_1$, $v_2$ 
        \nodepart{three} $\cdots$
    };
    \node [rectangle split, rectangle split horizontal, rectangle split parts = 3, draw, text centered] (B) at (0,2) {
        $u, \ldots$
        \nodepart{two} $\cdots$
        \nodepart{three} $\cdots$
    };
    \node [rectangle split, rectangle split horizontal, rectangle split parts = 3, draw, text centered] (C) at (5,0) {
        $\cdots$
        \nodepart{two} $v_1, v_3, v_4$
        \nodepart{three} $\cdots$
    };
    \node [rectangle split, rectangle split horizontal, rectangle split parts = 3, draw, text centered] (D) at (5,2) {
        $\cdots$
        \nodepart{two} $v_3, \ldots$
        \nodepart{three} $\cdots$
    };
    \node [rectangle split, rectangle split horizontal, rectangle split parts = 3, draw, text centered] (E) at (10,2) {
        $\cdots$
        \nodepart{two} $v_2, v_4$
        \nodepart{three} $\cdots$
    };
\end{scope}

\begin{scope}[>={Stealth[black]},
              every node/.style={fill=white,circle},
              every edge/.style={draw=black,very thick}]
    \path [-] (A) edge node {$u$} (B);
    \path [-] (A) edge node {$v_1$} (C);
    \path [-] (D) edge node {$v_3$} (C);
    \path [-] (C) edge node {$v_4$} (E);
    \end{scope}
\end{tikzpicture}
\caption{Assume that the coupling process picks again $c_3$, and couples $\widehat{X}$ and $\widehat{Y}$ on the third auxiliary variable $v_4$ of $c_3$. Assume that $\widehat{X}(v_4) = \widehat{Y}(v_4)$, so $\faild$ stays the same. Let us perform another iteration and assume that there is another clause $c_5$ whose auxiliary variables are exactly $v_2$ and $v_4$ and that this is the next clause picked by the coupling process. Since we cannot couple any more variables in $c_5$, we remove it from $\mathcal{C}_\mathrm{rem}$. On top of that, assume that when we coupled $\widehat{X}$ and $\widehat{Y}$ on $v_2$ and $v_4$, both times we sampled the value of $v_2$ and $v_4$ that does not satisfy $c_5$. Then, the clause $c_5$ is going to remain unsatisfied during the rest of the coupling process and is connected to $\faild \cup \failu$ in $G_\Phi$ via $v_4$, so we add it to $\failu$.}
\end{subfigure}
\end{figure}

\begin{figure}[H]
\centering
\begin{subfigure}[b]{\textwidth}
\centering
\begin{tikzpicture}
\begin{scope}[every node/.style={circle,thick,draw}]
    \node [rectangle split, rectangle split horizontal, rectangle split parts = 3, draw, text centered] (A) at (0,0) {
        $u, \ldots$
        \nodepart{two} $v_1$, $v_2$ 
        \nodepart{three} $\cdots$
    };
    \node [rectangle split, rectangle split horizontal, rectangle split parts = 3, draw, text centered] (B) at (0,2) {
        $u, \ldots$
        \nodepart{two} $\cdots$
        \nodepart{three} $\cdots$
    };
    \node [rectangle split, rectangle split horizontal, rectangle split parts = 3, draw, text centered] (C) at (5,0) {
        $\cdots$
        \nodepart{two} $v_1, v_3, v_4$
        \nodepart{three} $\cdots$
    };
    \node [rectangle split, rectangle split horizontal, rectangle split parts = 3, draw, text centered] (D) at (5,2) {
        $\cdots$
        \nodepart{two} $v_3, \ldots$
        \nodepart{three} $\cdots$
    };
    \node [rectangle split, rectangle split horizontal, rectangle split parts = 3, draw, text centered] (E) at (10,2) {
        $\cdots$
        \nodepart{two} $v_2, v_4$
        \nodepart{three} $\cdots$
    };
    \node [rectangle split, rectangle split horizontal, rectangle split parts = 1, draw, text centered] (F) at (10,0) {
        Bad clause containing $v_3$
    };
\end{scope}

\begin{scope}[>={Stealth[black]},
              every node/.style={fill=white,circle},
              every edge/.style={draw=black,very thick}]
    \path [-] (A) edge node {$u$} (B);
    \path [-] (A) edge node {$v_1$} (C);
    \path [-] (A) edge node {$v_2$} (D);
    \path [-] (D) edge node {$v_3$} (C);
    \path [-] (C) edge node {$v_4$} (E.south west);
    \path [-] (C) edge node {$v_3$} (F);
    \path [-] (F.north east) edge node {$v_3$} (D.south east);
    \end{scope}
\end{tikzpicture}
 \setcounter{subfigure}{5}
\caption{In this case the coupling picks $c_6$ with $v_3 \in \var(c_6) \cap (\verticesd \cup \var(\failu)) \ne \emptyset$. In this case, we assume that $c_6$ is a bad clause, and thus, it is removed from $\mathcal{C}_{\mathrm rem}$ and added to $\failu$. Clauses that are neighbours of $c_6$ can now be explored by the coupling process in future iterations.}
\end{subfigure}
\end{figure}

\end{example}

In the rest of this section we fix the inputs of Algorithm~\ref{alg:coupling} unless stated otherwise. In order to motivate our technical results and help with the understanding of this section, we provide the following proof outline, keeping in mind that our purpose is to use the coupling of Algorithm~\ref{alg:coupling} in conjunction with Proposition~\ref{prop:influence-coupling}.
\begin{enumerate}
    \item First, we analyse the sets $\mathcal{V}_{\mathrm {set}}$, $\verticesd$, $\faild$, $\failu$ and $\mathcal{C}_{\mathrm {rem}}$ obtained when the coupling process has finished, prove the connectivity property of $\faild \cup \failu$ and the fact that $\failu \cap \mathcal{C}_{\mathrm rem} = \emptyset$ (Proposition~\ref{prop:coupling:properties}).
    \item Then we analyse the structure of $\Phi^{\widehat{X}}$ and $\Phi^{\widehat{Y}}$, whose clauses turn out to be in  $\mathcal{C}_{\mathrm rem} \cup \failu$ (Lemma~\ref{lem:coupling:structure}) and, in fact, $\Phi^{\widehat{X}}$ and $\Phi^{\widehat{Y}}$ share the connected components of $\mathcal{C}_{\mathrm rem}$. These connected components are not connected to any other clause of $\Phi^{\widehat{X}}$ and $\Phi^{\widehat{Y}}$ since $\mathcal{C}_{\mathrm rem} \cup \failu$.
    \item Thirdly, we extend the coupling $(\widehat{X},\widehat{Y})$ to a coupling $(X,Y)$ on marked and auxiliary variables and use this coupling in conjunction with Proposition~\ref{prop:influence-coupling} to bound sums of influences between marked variables by a constant times the expected size of $\failu$ over the random choices of Algorithm~\ref{alg:coupling}, see Lemma~\ref{lem:coupling:si}.
    \item Finally, we bound the expected size of $\failu$ by $L := \lceil 2k^4 \log n \rceil$. In order to do so, first we bound the probability that a clause fails in the coupling process. The fact that the marginals of auxiliary variables are close to $1/2$ is relevant here. Then we perform a union bound over the number of connected sets of clauses of size $L$. This step turns out to present certain technical difficulties due to the fact that the events that $c_1 \in \faild \cup \failu$ and $c_2 \in \faild \cup \failu$ are not necessarily independent if $c_1$ and $c_2$ share variables. Moreover, a connected set of clauses may include several bad clauses, which always fail if they are ever considered in the coupling process. However, we are able to bypass these issues by applying our bounds on the number of bad clauses in a connected subgraph of $G_\Phi$ (Lemma~\ref{lem:bad}) and the fact if a connected set of clauses is logarithmic-sized, then it has constant tree-excess (Lemma~\ref{lem:tree-excess}), which greatly simplifies the dependency problem mentioned above. 
\end{enumerate}

With this proof outline in mind, we start our analysis of the coupling process. The some of the properties stated in Proposition~\ref{prop:coupling:properties} can be understood from the diagram given in Figure~\ref{fig:venn}.

\begin{proposition}[Properties of the coupling process] \label{prop:coupling:properties}
The coupling process in Algorithm~\ref{alg:coupling} terminates eventually and, at the end of the process, the sets $\mathcal{V}_{\mathrm {set}}$, $\verticesd$, $\faild$, $\failu$ and $\mathcal{C}_{\mathrm {rem}}$ present the following properties:
    \begin{enumerate}
        \item \label{item:coupling:1} We have $S \cup \{u\} \subseteq \mathcal{V}_{\mathrm {set}} \subseteq \mathcal{V}_{\mathrm a} \cup S \cup \{u\}$, $\verticesd = \{ v\in \mathcal{V}_{\mathrm {set}} : \widehat{X}(v) \ne \widehat{Y}(v) \}$, and $\faild$ is the set of clauses containing a variable in $\verticesd$.
        
        \item \label{item:coupling:2} For all $c \in \failu$ we have $\var(c) \cap \mathcal{V}_{\mathrm a} \subseteq \mathcal{V}_{\mathrm {set}}$ and $c$ is unsatisfied by at least one of $\widehat{X}$ and $\widehat{Y}$. 

        \item \label{item:coupling:3}  For all $c \in \mathcal{C}_{\mathrm {rem}}$, we have $\var(c) \cap (\verticesd \cup \var(\failu)) = \emptyset$. 
        
        \item \label{item:coupling:4}  For all $c \in \mathcal{C} \setminus (\mathcal{C}_{\mathrm {rem}} \cup \failu)$, we have $\var(c) \cap (\verticesd \cup \var(\failu)) \ne \emptyset$, $\var(c) \cap \mathcal{V}_{\mathrm a} \subseteq \mathcal{V}_{\mathrm {set}}$ and $c$ is satisfied by $\widehat{X}$ and $\widehat{Y}$.
        
        \item \label{item:coupling:5}  The set $\faild \cup \failu$ is connected in $G_\Phi$. 
    \end{enumerate}
\end{proposition}
\begin{proof}  
Each iteration of the coupling procedure either removes a clause from $\mathcal{C}_{\mathrm {rem}}$, or samples the values $\widehat{X}(v)$ and $\widehat{Y}(v)$ for an auxiliary variable $v$ and adds $v$ to $\mathcal{V}_{\mathrm {set}} \subseteq \mathcal{V}$. As $\mathcal{C}_{\mathrm {rem}}$ and $\mathcal{V}$ are finite, the coupling terminates after a finite number of iterations. We prove the five properties in the statement separately. First, we note that the sets  $\mathcal{V}_{\mathrm {set}}$, $\verticesd$, $\faild$, $\failu$ never decrease in size during the execution of Algorithm~\ref{alg:coupling}, whereas the set $\mathcal{C}_{\mathrm {rem}}$ never increases in size.

\emph{Property~\ref{item:coupling:1}.} Note that at the start of Algorithm~\ref{alg:coupling}   (line~\ref{line:start}) this property holds. The result then follows from the fact that the sets $\mathcal{V}_{\mathrm {set}}$, $\verticesd$ and $\faild$ are only updated from line~\ref{line:set:start} to line~\ref{line:set:end} of Algorithm~\ref{alg:coupling}, and these steps preserve Property~\ref{item:coupling:1}.

\emph{Property~\ref{item:coupling:2}.} This follows from the facts that the set $\failu$ is originally empty, it is only extended in lines~\ref{line:un:1} and~\ref{line:un:2}, and bad clauses do not contain auxiliary variables. 

\emph{Property~\ref{item:coupling:3}.} This property follows from the fact that clauses that satisfy $\var(c) \cap (\verticesd \cup \var(\failu)) \ne \emptyset$ at some point are eventually removed from $\mathcal{C}_{\mathrm {rem}}$ in either line~\ref{line:rem:bad} (if they are bad) or in line~\ref{line:rem:un} (if they are good, once all the auxiliary variables of the clause are in $\mathcal{V}_{\mathrm {set}}$).

\emph{Property~\ref{item:coupling:4}.} Let $c \in \mathcal{C} \setminus (\mathcal{C}_{\mathrm {rem}} \cup \failu)$. Then $c$ has been removed from $\mathcal{C}_{\mathrm {rem}}$ in line~\ref{line:un:1} or line~\ref{line:rem:un}. Note that if $c$ is removed from $\mathcal{C}_{\mathrm {rem}}$ in line~\ref{line:un:1}, then $c \in  \failu$, a contradiction. Thus, $c$ has been removed from $\mathcal{C}_{\mathrm {rem}}$ in line~\ref{line:rem:un}, but it has not been added to $\failu$ in line~\ref{line:un:2}, which proves this property.

\emph{Property~\ref{item:coupling:5}.} We note that at the start of the coupling process (line~\ref{line:start}) $\faild \cup \failu$ is connected. Let us analyse every line of the algorithm where the sets $\faild$ and $\failu$ are enlarged. When it comes to $\faild$, this occurs in line~\ref{line:discrepancy} if this line is executed. Let $c$ be the clause considered in that iteration of the coupling process and let $v$ be the variable of $c$ considered in line~\ref{line:v}. We recall that $\var(c) \cap (\verticesd \cup \var(\failu)) \ne \emptyset$ and $v \in (\var(c) \cap \mathcal{V}_{\mathrm a} ) \setminus \mathcal{V}_{\mathrm {set}}$. In line~\ref{line:discrepancy} we add all to $\faild$ all the clauses containing $v$. Let $C_v$ be the set of such clauses. Since $\emptyset \ne \var(c) \cap (\verticesd \cup \var(\failu)) \subseteq \var(c) \cap \var(\faild \cup \failu)$ and $c \in C_v$, we conclude that $\faild \cup \failu \cup C_v$ is connected as we wanted. When it comes to $\failu$, we add clauses in lines~\ref{line:un:1} and~\ref{line:un:2}. In this case, we add a clause $c$ such that $\var(c) \cap (\verticesd \cup \var(\failu)) \ne \emptyset$, so $\faild \cup \failu \cup \{c\}$ is connected in $G_\Phi$. 
\end{proof}

\begin{figure}[H]
    \centering
\begin{tikzpicture}
% Gray area
\scope
\fill[gray!30] (-2.85,-2) rectangle (4,2);
\fill[white] (-4,0) circle (1.5);
\endscope 

\draw (-4,0) circle (1.5) (-4,0)   node [text=black] {$\failu$}
      (-2.85,-1) rectangle (-1,1) (-2,0) node [text=black] {$\faild$}
      (5,0) circle (2) (5,0)  node [text=black] {$\mathcal{C}_{\mathrm{rem}}$};
\node at (1,0) {Satisfied by $\widehat{X}$ and $\widehat{Y}$};
\end{tikzpicture}
    \caption{Venn diagram of the sets of clauses $\failu, \faild$ and $\mathcal{C}_{\mathrm rem}$. Here the area plotted in gray represents the set of clauses that are satisfied by $\widehat{X}$ and $\widehat{Y}$. Every clause is in one of these four sets. Moreover, $\var(\failu \cup \faild) \cap \var(\mathcal{C}_{\mathrm rem}) = \emptyset$ and, in particular, $\failu \cup \faild$ and $\mathcal{C}_{\mathrm rem}$ are disjoint (Proposition~\ref{prop:coupling:properties}, Property~\ref{item:coupling:3}).}
    \label{fig:venn}
\end{figure}

We can now prove our main result concerning the structure of $\Phi^{\widehat{X}}$ and $\Phi^{\widehat{Y}}$.

\begin{lemma} \label{lem:coupling:structure}
     Let $\widehat{X}$ and $\widehat{Y}$ be the assignments returned by Algorithm~\ref{alg:coupling} and let $\mathcal{C}_{\mathrm {rem}}$ and $\failu$ be as in Proposition~\ref{prop:coupling:properties}. There are sets of clauses $\mathcal{C}_1 \subseteq \mathcal{C}_{\mathrm {rem}}$ and $\mathcal{C}_2, \mathcal{C}_3 \subseteq \failu$ such that $\Phi^{\widehat{X}} = (\mathcal{V} \setminus \mathcal{V}_{\mathrm {set}}, \mathcal{C}_1 \cup \mathcal{C}_2)$ and $\Phi^{\widehat{Y}} = (\mathcal{V} \setminus \mathcal{V}_{\mathrm {set}}, \mathcal{C}_1 \cup \mathcal{C}_3)$, where the variables in $\mathcal{V}_{\mathrm {set}}$ are removed from the clauses in $\mathcal{C}_1, \mathcal{C}_2, \mathcal{C}_3$.
\end{lemma}

\begin{proof}
We determine the set of clauses that are unsatisfied by $\widehat{X}$ or $\widehat{Y}$ with the help of Proposition~\ref{prop:coupling:properties}. We distinguish 3 disjoint cases:
\begin{itemize}
    \item $c \in \mathcal{C}_{\mathrm {rem}}$. Then $\var(c) \cap \verticesd = \emptyset$, so $\widehat{X}$ and $\widehat{Y}$ agree in all the variables in $\var(\mathcal{C}_{\mathrm {rem}}) \cap \mathcal{V}_{\mathrm {set}}$. As a consequence, the restrictions of $\Phi^{\widehat{X}}$ and $\Phi^{\widehat{Y}}$ to $\mathcal{C}_{\mathrm {rem}}$ give rise to the same CNF formula. Note that some of the clauses in $\mathcal{C}_{\mathrm {rem}}$ might be satisfied by both $\widehat{X}$ and $\widehat{Y}$, but they are never satisfied by only one of the two assignments.
    
    \item $c \in \failu$. Then by Property~\ref{item:coupling:2} of Proposition~\ref{prop:coupling:properties}, $c$ unsatisfied by at least one of $\widehat{X}$ and $\widehat{Y}$ and, thus, it appears in at least one of $\Phi^{\widehat{X}}$ and $\Phi^{\widehat{Y}}$. The clause $c$ may contain a variable $v \in \verticesd$, in which case $c$ is a good clause and it appears in exactly one of the formulas $\Phi^{\widehat{X}}$ and $\Phi^{\widehat{Y}}$.
    
    \item $c \in \mathcal{C} \setminus (\mathcal{C}_{\mathrm {rem}} \cup \failu)$. By Proposition~\ref{prop:coupling:properties}, we have $\var(c) \cap (\verticesd \cup \var(\failu)) \ne \emptyset$ and $\var(c) \cap \mathcal{V}_{\mathrm a} \subseteq \mathcal{V}_{\mathrm {set}}$.  Since $c \not \in \failu$, it follows that $c$ is satisfied by both $\widehat{X}$ and $\widehat{Y}$ and, thus, $c$ does not appear in any of the formulae $\Phi^{\widehat{X}}$ and $\Phi^{\widehat{Y}}$. 
\end{itemize}

We conclude that we can write $\mathcal{C}^{\widehat{X}} = \mathcal{C}_1 \cup \mathcal{C}_2$ and  $\mathcal{C}^{\widehat{Y}} = \mathcal{C}_1 \cup \mathcal{C}_3$, where $\mathcal{C}_1 \subseteq \mathcal{C}_{\mathrm {rem}}$ and $\mathcal{C}_2, \mathcal{C}_3 \subseteq \failu$ as we wanted.
\end{proof}

  In order to further analyse the probability distribution of the output of Algorithm~\ref{alg:coupling}, we introduce the following definition.
  
\begin{definition}[run, $\mathcal{R}(\Phi, \mathcal{M}, u, \Lambda)$, $\tau_{\mathcal{R}}(\Phi, \mathcal{M}, u, \Lambda)$, $\mathcal{V}_{\mathrm {set}}(R)$, $\verticesd(R)$, $\failu(R)$, $\faild(R)$, $\mathcal{C}_{\mathrm {rem}}(R)$] \label{def:R}
  A \emph{run} of Algorithm~\ref{alg:coupling} is a sequence of all the random choices $(\widehat{X}(v), \widehat{Y}(v))$ made in line~\ref{line:sample} when executing Algorithm~\ref{alg:coupling}. Let $\mathcal{R}(\Phi, \mathcal{M}, u, \Lambda)$ be the set of all possible runs of Algorithm~\ref{alg:coupling} for the inputs $\Phi, \mathcal{M}, u, \Lambda$ and let $\tau_{\mathcal{R}}(\Phi, \mathcal{M}, u, \Lambda)$ be the probability distribution that Algorithm~\ref{alg:coupling} yields on $\mathcal{R}(\Phi, \mathcal{M}, u, \Lambda)$. Each run $R \in \mathcal{R}(\Phi, \mathcal{M}, u, \Lambda)$ determines the output $(\widehat{X}, \widehat{Y})$ and the sets $\mathcal{V}_{\mathrm {set}}(R), \verticesd(R), \failu(R), \faild(R), \mathcal{C}_{\mathrm {rem}}(R)$ that are computed in Algorithm~\ref{alg:coupling}. 
\end{definition}

With the aim of applying Proposition~\ref{prop:influence-coupling}, we extend the coupling  $(\widehat{X}, \widehat{Y})$ to all marked and auxiliary variables. 

\begin{definition}[The coupling $(X,Y)$]  \label{def:coupling-process}
  Let $R \in \mathcal{R}(\Phi, \mathcal{M}, u, \Lambda)$ and let $(\widehat{X}, \widehat{Y})$ be the corresponding output of the run $R$. Let $\le_{\mathcal{V}}$ be a total order on the variables of $\Phi$ and let $v_1 \le_{\mathcal{V}} v_2\le_{\mathcal{V}} \dots \le_{\mathcal{V}} v_t$ be the variables in $(\mathcal{V}_{\mathrm m} \cup \mathcal{V}_{\mathrm a}) \setminus \mathcal{V}_{\mathrm {set}}$. We extend the assignments $\widehat{X}, \widehat{Y} \colon \mathcal{V}_{\mathrm {set}} \to \{\mathsf{F}, \mathsf{T}\}$ to $v_1, v_2, \ldots, v_t$ inductively (as follows) to obtain a coupling $(X,Y)$ such that $X$ follows the distribution $\restr{\mu_{\Omega^{\Lambda \cup u \mapsto \mathsf{T}}}}{(\mathcal{V}_{\mathrm m} \cup \mathcal{V}_{\mathrm a}) \setminus \mathcal{V}_{\mathrm {set}}}$ and $Y$ follows the distribution $\restr{\mu_{\Omega^{\Lambda \cup u \mapsto \mathsf{F}}}}{(\mathcal{V}_{\mathrm m} \cup \mathcal{V}_{\mathrm a})\setminus \mathcal{V}_{\mathrm {set}}}$. Assume that $X$ and $Y$ are defined on $\mathcal{V}_{\mathrm {set}} \cup \{v_1, v_2, \ldots, v_{j-1}\}$ for $j \in \{1,2,\ldots,t\}$. Then we sample $(X(v_j), Y(v_j))$ from the optimal/monotone coupling of the marginal distributions (on $v_j$) of $\mu_{\Omega^{X}}$ and $\mu_{\Omega^{Y}}$.
\end{definition}

\begin{remark} \label{rem:coupling}
    When $R \in  \mathcal{R}(\Phi, \mathcal{M}, u, \Lambda)$ follows the probability distribution $\tau_{\mathcal{R}}(\Phi, \mathcal{M}, u, \Lambda)$ (Definition~\ref{def:R}), the pair of random assignments $(X,Y)$ of Definition~\ref{def:coupling-process} is a coupling of the distributions $\restr{\mu_{\Omega^{\Lambda \cup u \mapsto \mathsf{T}}}}{\mathcal{V}_{\mathrm m} \cup \mathcal{V}_{\mathrm a}}$ and $\restr{\mu_{\Omega^{\Lambda \cup u \mapsto \mathsf{F}}}}{\mathcal{V}_{\mathrm m} \cup \mathcal{V}_{\mathrm a}}$. 
\end{remark}

In Lemma~\ref{lem:coupling:discrepancy} we bound the probabilities $\pr(X(v) \ne Y(v) \vert R)$ for any $R \in  \mathcal{R}(\Phi, \mathcal{M}, u, \Lambda)$ and $v \in (\mathcal{V}_{\mathrm m}\cup \mathcal{V}_{\mathrm a}) \setminus \mathcal{V}_{\mathrm {set}}(R)$.

\begin{lemma} \label{lem:coupling:discrepancy}
    Let $R \in \mathcal{R}(\Phi, \mathcal{M}, u, \Lambda)$. Let $(X,Y)$ be the coupling of Definition~\ref{def:coupling-process}. Then for any $v \in (\mathcal{V}_{\mathrm m} \cup \mathcal{V}_{\mathrm a}) \setminus \mathcal{V}_{\mathrm {set}}(R)$ we have $\pr(X(v) \ne Y(v) \vert R) \le 2^{-r_0 k + 1} / k$.
\end{lemma}
\begin{proof} 
Let $\widehat{X}$ and $\widehat{Y}$ be the output of Algorithm~\ref{alg:coupling} for the run $R$. Let $v_1, v_2, \ldots, v_t$ be the variables in $(\mathcal{V}_{\mathrm m} \cup \mathcal{V}_{\mathrm a}) \setminus \mathcal{V}_{\mathrm {set}}(R)$ in the order that they are considered in Definition~\ref{def:coupling-process}. Let $j \in \{1, 2, \ldots, t\}$ and let $\Lambda', \Lambda'' \colon \mathcal{V}_{\mathrm {set}}(R) \cup \{v_1, v_2, \ldots, v_{j-1}\} \to \{\mathsf{F}, \mathsf{T}\}$ be two assignments such that $\restr{\Lambda'}{\mathcal{V}_{\mathrm {set}}} = \widehat{X}$ and $\restr{\Lambda''}{\mathcal{V}_{\mathrm {set}}} = \widehat{Y}$. When $X$ agrees with $\Lambda'$ and $Y$ agrees with $\Lambda''$, the values $X(v_j)$ and $Y(v_j)$ are sampled from the optimal/monotone coupling between the marginals on $v_j$ of the distributions $\mu_{\Omega^{\Lambda'}}$ and $\mu_{\Omega^{\Lambda''}}$. Let us denote these marginals by $\nu_X$ and $\nu_Y$ respectively. Thus, by the coupling lemma (Proposition~\ref{prop:coupling-lemma}) and Proposition~\ref{prop:marginals} (or Lemma~\ref{lem:marginals}) on the marginals of marked and auxiliary variables, we have 
\begin{equation*}
    \begin{aligned}
      \pr\left(X(v_j) \ne Y(v_j) \vert \Lambda', \Lambda'' \right) & = d_{\TV}(\nu_X, \nu_Y) = \left|\pr(X(v_j) = \mathsf{T} \vert \Lambda') - \pr(Y(v_j) = \mathsf{T} \vert \Lambda'' )\right| \\       
      & \le \lvert  \pr(X(v_j) = \mathsf{T} \vert \Lambda') - 1/2  \rvert + \lvert 1/2 -  \pr(Y(v_j) = \mathsf{T} \vert \Lambda'') \rvert \\ 
      & \le \exp \left( \frac{1}{k 2^{r_0 k}} \right) - 1.
    \end{aligned}
\end{equation*}
Applying the inequality $e^z \le 1+2z$ for $z \in (0,1)$, we find that  $\pr\left(X(v_j) \ne Y(v_j) \vert \Lambda', \Lambda''  \right)  \le  2^{-r_0 k + 1} / k$. Thus, from the arbitrary choice of $\Lambda', \Lambda''$ and the law of total probability we conclude that the bound  $\pr\left(X(v_j) \ne Y(v_j) \vert R \right) \le  2^{-r_0k + 1} / k$ holds.
\end{proof}

Combining all the results presented up to this stage in the current section allows us relate the sum $ \sum_{v \in \mathcal{V}_{\mathrm m} \setminus (S \cup \{u\}) } \left| \mathcal{I}^\Lambda(u \to v) \right|$ to the coupling process over auxiliary variables. In fact, we bound this sum of influences between marked variables by the expected number of failed clauses in the coupling process on auxiliary variables. Recall that here $r = r_0 - \delta$.

\newcommand{\statelemexpectation}{
   There is an integer $k_0$ such that for any $k \ge k_0$ and any density $\alpha$  with $\density \le 2^{(r_0 - \delta)k}/k^3$ the following holds w.h.p. over the choice of the random $k$-CNF formula $\Phi = \Phi(k, n, \lfloor \alpha n\rfloor)$. Let $(\mathcal{V}_{\mathrm m}, \mathcal{V}_{\mathrm a}, \mathcal{V}_{\mathrm c})$ be an $(r_0 - \delta, r_0, r_0, 2r_0)$-marking of $\Phi$, and let $u \in \mathcal{V}_{\mathrm m}$ and $\Lambda \colon S \to \{\mathsf{F}, \mathsf{T}\}$ with $S \subseteq \mathcal{V}_{\mathrm m} \setminus \{u\}$. Then for a random run $R$ of the coupling process on the auxiliary variables (Algorithm~\ref{alg:coupling}), we have 
       \begin{equation*}
         \sum_{v \in \mathcal{V}_{\mathrm m} \setminus (S \cup \{u\}) } \left| \mathcal{I}^\Lambda(u \to v) \right| \le 2^{-r_0k+1}  \mathbb{E}  \left[\lvert  \failu(R) \rvert \right].
     \end{equation*}
}
\begin{lemma} \label{lem:expectation} \label{lem:coupling:si}
  \statelemexpectation
\end{lemma}
\begin{proof}
  Let $(X,Y)$ be the coupling in Definition~\ref{def:coupling-process} for a (random) run $R \sim \tau_{\mathcal{R}}(\Phi, \mathcal{M}, u, \Lambda)$ of Algorithm~\ref{alg:coupling}. 
  We are going to show that 
  \begin{equation} \label{eq:R}
    \pr(X(v) = Y(v) \vert R) = 1 \text{ for all } v \in V := (\mathcal{V}_{\mathrm m} \cup \mathcal{V}_{\mathrm a}) \setminus (\mathcal{V}_{\mathrm {set}}(R) \cup \var(\failu(R))).
  \end{equation}
  Let $\widehat{X}, \widehat{Y} \colon \mathcal{V}_{\mathrm {set}}(R) \to \{\mathsf{F}, \mathsf{T}\}$ be the output of Algorithm~\ref{alg:coupling} for the run $R$. By Lemma~\ref{lem:coupling:structure} we conclude that we can write $\mathcal{C}^{\widehat{X}} = \mathcal{C}_1 \cup \mathcal{C}_2$ and $\mathcal{C}^{\widehat{Y}} = \mathcal{C}_1 \cup \mathcal{C}_3$, where $\mathcal{C}_1 \subseteq \mathcal{C}_{\mathrm {rem}}(R)$ and $\mathcal{C}_2, \mathcal{C}_3 \subseteq \failu(R)$. Thus, the variables in $V$ (see~\eqref{eq:R} for a definition of $V$) either appear in a clause in $\mathcal{C}_1$ or they are not present in any of the formulae $\Phi^{\widehat{X}}$ and $\Phi^{\widehat{Y}}$. Moreover, by Proposition~\ref{prop:coupling:properties}, we have $\var(c) \cap \var(c') = \emptyset$ for all $c \in \mathcal{C}_{\mathrm {rem}}(R)$ and $c' \in \failu(R)$. We conclude that the distributions $\restr{\mu_{\Omega^{\widehat{X}}}}{V}$ and $\restr{\mu_{\Omega^{\widehat{Y}}}}{V}$ agree -- both are the uniform distribution over the satisfying assignments of the CNF formula $(V, \mathcal{C}_1)$. Let $v_1, v_2, \ldots, v_t$ be the variables in $V$ in the order they are considered in the definition of the coupling $(X,Y)$. By induction on $j \in \{1, 2, \ldots, t\}$, the marginals on $v_j$ in Definition~\ref{def:coupling-process} are the same when coupling $X(v_j)$ and $Y(v_j)$. Thus, we have $X(v_j) = Y(v_j)$ for all  $j \in \{1, 2, \ldots, t\}$.
  
  Since $S\cup\{u\} \subseteq \mathcal{V}_{\mathrm {set}}(R) \subseteq S\cup\{u\} \cup \mathcal{V}_{\mathrm a}$, we have $\mathcal{V}_{\mathrm m} \setminus V = S \cup \{u\} \cup (\mathcal{V}_{\mathrm m} \cap \var(\failu(R)))$. In light of Lemma~\ref{lem:coupling:discrepancy} and \eqref{eq:R}, we find that 
  \begin{equation*}
         \sum_{v \in \mathcal{V}_{\mathrm m} \setminus (S \cup \{u\}) } \pr(X(v) \ne Y(v) \vert R) \le \sum_{v \in \mathcal{V}_{\mathrm m} \cap \var(\failu(R)) } \pr(X(v) \ne Y(v) \vert R) \le \frac{2}{k} 2^{-r_0 k} \lvert  \var(\failu(R)) \rvert.
  \end{equation*}
  From $\lvert  \var(\failu(R)) \rvert \le k \lvert \failu(R) \rvert$ we conclude that
  \begin{equation} \label{eq:fu}
    \sum_{v \in \mathcal{V}_{\mathrm m} \setminus (S \cup \{u\}) } \pr(X(v) \ne Y(v) \vert R) \le 2^{-r_0k+1} \lvert  \failu(R) \rvert.
  \end{equation}
  In the rest of this proof we are going to aggregate \eqref{eq:fu} over $R \in \mathcal{R}(\Phi, \mathcal{M}, u, \Lambda)$ with the aim of applying Proposition~\ref{prop:influence-coupling}. Let $(X,Y)$ be the coupling in Definition~\ref{def:coupling-process} for a (random) run $R \sim \tau_{\mathcal{R}}(\Phi, \mathcal{M}, u, \Lambda)$ of Algorithm~\ref{alg:coupling}. We have
  \begin{equation*}
\begin{aligned}
     \sum_{v \in \mathcal{V}_{\mathrm m} \setminus (S \cup \{u\}) } \pr(X(v) \ne Y(v)) & = \sum_{v \in \mathcal{V}_{\mathrm m} \setminus (S \cup \{u\}) } \sum_{R \in \mathcal{R}(\Phi, \mathcal{M}, u, \Lambda)} \pr(R) \pr(X(v) \ne Y(v) \vert R)  \\ 
     & = \sum_{R \in \mathcal{R}(\Phi, \mathcal{M}, u, \Lambda)} \pr(R) \sum_{v \in \mathcal{V}_{\mathrm m} \setminus (S \cup \{u\}) }   \pr(X(v) \ne Y(v) \vert R) \\
    & \le 2^{-r_0 k+1} \sum_{R \in \mathcal{R}(\Phi, \mathcal{M}, u, \Lambda)} \pr(R) \lvert  \failu(R) \rvert \\ 
    & = 2^{-r_0 k+1} \mathbb{E}  \left[ \lvert  \failu(R) \rvert \right]. 
    \end{aligned}
  \end{equation*}
  Finally, we note that we can indeed apply Proposition~\ref{prop:influence-coupling} to the restriction of $X$ and $Y$ on $\mathcal{V}_{\mathrm m}$ as $(X,Y)$ is a coupling of the distributions $\restr{\mu_{\Omega^{\Lambda \cup u \mapsto \mathsf{T}}}}{\mathcal{V}_{\mathrm m} \cup \mathcal{V}_{\mathrm a}}$ and $\restr{\mu_{\Omega^{\Lambda \cup u \mapsto \mathsf{F}}}}{\mathcal{V}_{\mathrm m} \cup \mathcal{V}_{\mathrm a}}$ (Remark~\ref{rem:coupling}). This finishes the proof.
\end{proof}

In the remainder of this section we bound $\mathbb{E}  \left[ \lvert  \failu(R) \rvert \right]$, which would complete our proof of Lemma~\ref{lem:si} when combined with Lemma~\ref{lem:coupling:si}. In order to do this we exploit the fact that $\failu(R) \cup \faild(R)$ is connected in $G_\Phi$ (Proposition~\ref{prop:coupling:properties}), the local sparsity properties of random CNF formulae and the properties of the marking $(\mathcal{V}_{\mathrm m}, \mathcal{V}_{\mathrm a}, \mathcal{V}_{\mathrm c})$. It is important that the bound on  $\mathbb{E}  \left[ \lvert  \failu(R) \rvert \right]$ is  $\operatorname{poly}(k) \log n$ in order to conclude fast mixing time of the $\rho$-uniform-block Glauber dynamics when applying the spectral independence framework. First, we  bound the probability that some good clauses are failed in Algorithm~\ref{alg:coupling}. At first glance this seems to be a straightforward task thanks to the fact that the marginals of marked and auxiliary variables are close to $1/2$ (see Proposition~\ref{prop:marginals}). However, for any good clauses $c_1$ and $c_2$, the events that $c_1 \in \faild(R) \cup \failu(R)$ and $c_2 \in \faild(R) \cup \failu(R)$ may not be independent; any value given to the variables in $c_1$ may affects the marginals of the variables in $c_2$ and whether these variables are considered by the coupling process or not. However, we show that, as long as $c_1$ and $c_2$ do not share good variables, these dependencies are not very strong and we can indeed bound the probability that $c_1, c_2 \in \faild(R) \cup \failu(R)$ with a careful probability argument that analyses the coupling process step by step, see Lemma~\ref{lem:prob-V}. With this in mind, we introduce the following definitions.

\begin{definition}[$\mathcal{R}_t(\Phi, \mathcal{M}, u, \Lambda)$, $\mathcal{A}_{\le t}$] \label{def:At}
  For a positive integer $t$, we let $\mathcal{R}_t(\Phi, \mathcal{M}, u, \Lambda)$ be the set containing for each $R \in \mathcal{R}(\Phi, \mathcal{M}, u, \Lambda)$ a tuple with the first $\min\{t, \operatorname{length}(R)\}$ entries of the sequence $R$. That is,  $\mathcal{R}_t(\Phi, \mathcal{M}, u, \Lambda)$ is the set containing all possible sequences of the first $t$ choices that Algorithm~\ref{alg:coupling} makes in line~\ref{line:sample}. Note that if $R \in \mathcal{R}(\Phi, \mathcal{M}, u, \Lambda)$ has $\operatorname{length}(R) \le t$, then $R \in \mathcal{R}_t(\Phi, \mathcal{M}, u, \Lambda)$. Each $R_t \in \mathcal{R}_t(\Phi, \mathcal{M}, u, \Lambda)$ determines two partial assignments $\Lambda'$ and $\Lambda''$ of marked and auxiliary variables that correspond to the assignments $\widehat{X}$ and $\widehat{Y}$ after $\operatorname{length}(R_t)$ iterations of line~\ref{line:sample} following $R_t$. Let $\mathcal{A}_{\le t}$  be the $\sigma$-algebra containing all the subsets of $\mathcal{R}_t(\Phi, \mathcal{M}, u, \Lambda)$. 
\end{definition}

 Intuitively, $\mathcal{A}_{\le t}$  contains all the possible events that may occur in the first $t$ iterations of line~\ref{line:sample}, which is the only randomised operation in Algorithm~\ref{alg:coupling}. When bounding the probability that a clause is failed, we will express this event in terms of events concerning the values that $\widehat{X}$ and $\widehat{Y}$ take on its variables. This motivates Definition~\ref{def:dvj}.

\begin{definition}[$D_v(j)$] \label{def:dvj}
  We define the following events 
  for variable $v \in \mathcal{V}_{\mathrm a}$ and a random run $R \sim \tau_{\mathcal{R}}(\Phi, \mathcal{M}, u, \Lambda)$ of Algorithm~\ref{alg:coupling}. Let  $D_{v}(1)$ be the event that $v \in \mathcal{V}_{\mathrm {set}}(R)$ and $\widehat{X}(v) \ne \widehat{Y}(v)$. Let $D_v(2)$ be the event that $v \in \mathcal{V}_{\mathrm {set}}(R)$ and $\widehat{X}(v) = \mathsf{F}$. Let $D_v(3)$ be the event that $v \in \mathcal{V}_{\mathrm {set}}(R)$ and $\widehat{X}(v) = \mathsf{T}$. Let $D_v(4)$ be the event that $v \in \mathcal{V}_{\mathrm {set}}(R)$ and $\widehat{Y}(v) = \mathsf{F}$. Let $D_v(5)$ be the event that $v \in \mathcal{V}_{\mathrm {set}}(R)$ and $\widehat{Y}(v) = \mathsf{T}$. 
\end{definition}

Finally, in order to study the events $D_v(j)$ for $v \in V$ we will have to reason about the first time that a variable in $V$ is added to $\mathcal{V}_{\mathrm {set}}(R)$, which motivates the following definition.

\begin{definition}[$\tau(V)$, $f(V)$] \label{def:tau-V}
For a set of auxiliary variables $V$, we let $\tau(V)$ be the random variable that takes the value $t$ if the first time that a variable in $V$ is added to $\mathcal{V}_{\mathrm {set}}(R)$ in Algorithm~\ref{alg:coupling} is the $t$-th time line~\ref{line:sample} is executed, and we denote by $f(V)$ this variable. We set $\tau(V) = \infty$ if $V \cap \mathcal{V}_{\mathrm {set}}(R) = \emptyset$, in which case $f(V)$ is not defined.
\end{definition}

We now have all the tools that we need to analyse the coupling process step by step.
 
\begin{lemma} \label{lem:prob-V}
    Let $V \subseteq \mathcal{V}_{\mathrm a}$ and let $i_v \in \{1,2,3,4,5\}$ for each $v \in V$. Let $h(1) = 2^{-r_0 k + 1}/k$ and $h(i) = \frac{\exp(1/k)}{2}$ for $i\in \{2,3,4,5\}$.  Then, we have
    \begin{equation*} \label{eq:V-var:0}
        \pr_{R \sim \tau_{\mathcal{R}}(\Phi, \mathcal{M}, u, \Lambda)}\left( \bigcap\nolimits_{v\in V} D_v(i_v)  \right) \le \prod_{v \in V} h(i_v).
    \end{equation*}
\end{lemma}
 \begin{proof}
    We are going to prove, for any positive integer $t$ and $A \in \mathcal{A}_{\le t}$, 
     \begin{equation} \label{eq:V-var}
        \pr\left( \left. \bigcap\nolimits_{v\in V} D_v({i_v}) \right| A, \tau(V) = t  \right) \le \prod_{v \in V} h(i_v).
    \end{equation} 
The lemma will then follow from the arbitrary choice of~$A$ and $t$ and the law of total probability.

    We carry out the proof of \eqref{eq:V-var} by induction on $M = \lvert V \rvert$. 
    Equation \eqref{eq:V-var} holds when $V$ is empty. Let us assume that \eqref{eq:V-var} holds when $\lvert V \rvert < M$. Let $V$ be a set of auxiliary variables with $M = \lvert V \rvert$ and indexes $i_v$ for all $v \in V$, let $t$ be a positive integer and let $A \in \mathcal{A}_{\le t}$. To simplify the notation,   
    for each $w \in V$ we define $A_t(w, V) = A \cap [\tau(V) = t] \cap [f(V) = w]$.     
    Then, we have 
    \begin{equation*}
    \begin{aligned}
        \pr\left( \left. \bigcap\nolimits_{v\in V} D_v({i_v}) \right|  A, \tau(V) = t  \right) \le \sum_{w \in V} & \pr\left( f(V) = w \vert A, \tau(V) = t \right) \cdot \pr\left( \left. D_{w}({i_{w}}) \right| A_t(w, V) \right) \\ & \cdot  \pr\left(\left.   \bigcap\nolimits_{v\in V\setminus\{w\}} D_{v}({i_{v}}) \right| A_t(w, V),  D_{w}({i_{w}})  \right).
    \end{aligned}
    \end{equation*}
    We note that $\tau(V \setminus\{w\}) > t$ when conditioning on $\tau(V) = t$ and $f(V) = w$. Let $A' = A_t(w, V) \cap  D_{w}({i_{w}})$. We have
        \begin{equation*}
    \begin{aligned}
      \pr\left(\left.   \bigcap\nolimits_{v\in V\setminus\{w\}} D_{v}({i_{v}}) \right| A'   \right) = \sum_{j = t+1}^\infty & \pr\left (\left. \tau(V \setminus \{w\}) = j \right| A' \right) \\ & \cdot \pr\left(\left.   \bigcap\nolimits_{v\in V\setminus\{w\}} D_{v}({i_{v}}) \right| A', \tau(V \setminus \{w\}) = j \right) .
    \end{aligned}
    \end{equation*}
    By our induction hypothesis for $V \setminus \{w\}$, the condition $\tau(V \setminus \{w\}) = j$ and the event $A' \in \mathcal{A}_{\le j}$, we find that 
     \begin{equation*}
    \begin{aligned}
      \pr\left(\left.   \bigcap\nolimits_{v\in V\setminus\{w\}} D_{v}({i_{v}}) \right| A'  \right) & \le  \sum_{j = t+1}^\infty  \pr\left (\left. \tau(V \setminus \{w\}) = j \right| A' \right) \prod_{v \in V \setminus \{w\}} h(i_v)  \le \prod_{v \in V \setminus \{w\}} h(i_v).
    \end{aligned}
    \end{equation*}
    As a consequence, we obtain
    \begin{equation*}
    \begin{aligned}
        \pr\left( \left. \bigcap\nolimits_{v\in V} D_v({i_v}) \right|  A, \tau(V) = t  \right) \le \sum_{w \in V}  &  \pr\left( f(V) = w \vert A, \tau(V) = t \right) \cdot  \pr\left( \left. D_{w}({i_{w}}) \right| A_t(w, V) \right) \\ & \cdot  \prod_{v \in V \setminus \{w\}} h(i_v).
    \end{aligned}
    \end{equation*}
    We are going to show that $\pr( D_{w}({i_{w}}) \vert A_t(w, V) ) \le h(i_{w})$.  
    Once we have proved this, the proof of \eqref{eq:V-var} is completed by noting that $\sum_{w \in V}  \pr\left( f(V) = w \vert A, \tau(V) = t \right) = 1$.
    
    Let us now bound $\pr( D_{w}(i_{w}) \vert A_t(w, V) )$. Recall here that $A_t(w, V)$ implies the event $w \in \mathcal{V}_{\mathrm {set}}(R)$. Recall also that $A_t(w, V) \in \mathcal{A}_{\le t}$, see Definition~\ref{def:At}. For each $R_t \in A_t(w, V) \subseteq \mathcal{R}_t(\Phi, \mathcal{M}, u, \Lambda)$, we are going to apply Proposition~\ref{prop:marginals} and the fact that $\widehat{X}(w)$ and $\widehat{Y}(w)$ follow the optimal coupling between two marginal distributions on $w$ of the form $\mu_{\Omega^{\Lambda'}}$ and $\mu_{\Omega^{\Lambda''}}$ for some  assignments $\Lambda', \Lambda''$ on some marked and auxiliary variables that are determined by $R_t$.  Here it is important for applying Proposition~\ref{prop:marginals} that the event $A_t(w, V)$ is in $\mathcal{A}_{\le t}$, so every partial run $R_t \in A_t(w, V)$ only gives information about what has happened in Algorithm~\ref{alg:coupling} before $w$ is added to $\mathcal{V}_{\mathrm {set}}(R)$. Thus, aggregating over all possible runs $R_t \in A_t(w, V)$, we find that
        \begin{equation} \label{eq:marginals-R}
        \begin{aligned}
            \max \left\{ \pr\left(\left.\widehat{X}(w) = \mathsf{F}\right|  A_t(w, V) \right), \pr\left(\left.\widehat{X}(w) = \mathsf{T}\right|  A_t(w, V)\right)  \right\} & \le \frac{1}{2} \exp \left( \frac{1}{k 
2^{r_0 k}      } \right) \\ & \le \frac{1}{2} \exp\big( \frac{1}{k} \big),
        \end{aligned}
        \end{equation}
        where the probability is over the random run $R \sim \tau_{\mathcal{R}}(\Phi, \mathcal{M}, u, \Lambda)$. The bound \eqref{eq:marginals-R} also applies with $\widehat{Y}$ instead of $\widehat{X}$. In particular, we conclude that $\pr( D_{w}(j) \vert A_t(w, V) ) \le \exp(1/k)/2 = h(j)$  for all $j \in \{2,3,4,5\}$. Moreover, using the definition of optimal coupling for two Bernoulli distributions, the probability that $\widehat{X}(w) \ne \widehat{Y}(w)$ can be bounded as 
        \begin{equation*}
        \begin{aligned}
           \pr\left(\left.\widehat{X}(w) \ne \widehat{Y}(w) \right|  A_t(w, V) \right) &  =  \left|\pr\left(\left. \widehat{X}(w) = \mathsf{T}\right|  A_t(w, V) \right)  - \pr\left(\left. \widehat{Y}(w) = \mathsf{T} \right|   A_t(w, V) \right) \right| \\       
           & \le \left| \pr\left(\left. \widehat{X}(w) = \mathsf{T}\right|  A_t(w, V)\right) - 1/2  \right|  + \left|  1/2 -  \pr\left(\left. \widehat{Y}(w) = \mathsf{T}\right|   A_t(w, V) \right) \right|  \\ 
           & \le \exp \left( \frac{1}{k 2^{r_0 k}} \right) - 1.
        \end{aligned}
        \end{equation*}
        Hence, applying the bound $e^z \le 1+2z$ for $z \in (0,1)$ and the definition of the event $D_{v_j}(1)$, we have
        $\pr ( D_{v_j}(1) |  A_t(w, V) )  \le 2 / (k 2^{r_0 k}) = h(1)$. This finishes the proof of \eqref{eq:V-var}. From the arbitrary choice of $A$ and $t$ and the law of total probability, the statement follows.
\end{proof}
 
We can now bound the probability that some good clauses are failed with the help of Lemma~\ref{lem:prob-V}.

\begin{lemma} \label{lem:bound-fail}
  Let $\Phi, u, \Lambda$ be the input of Algorithm~\ref{alg:coupling}. Let $c_1, \ldots, c_\ell \in \gc$ such that the variable $u$ does not appear in any of the clauses in $c_1, \ldots, c_\ell$, and $\var(c_i) \cap \var(c_j) \cap \gv = \emptyset$ for all $1 \le i < j \le \ell$. Then, for $R \sim \tau_{\mathcal{R}}(\Phi, \mathcal{M}, u, \Lambda)$, we have $\pr(c_1, \ldots, c_\ell \in \faild(R) \cup \failu(R))\le 2^{(-r_0 k+4) \ell}$.
\end{lemma}
\begin{proof} 
    Let $c_1, \ldots, c_\ell$ be some good clauses of $\Phi$ as in the statement. The hypothesis that $u$ does not appear in any of these clauses is necessary as if $u \in \var(c)$ then $c \in \faild(R)$ by definition.
    We consider a random run $R \sim \tau_{\mathcal{R}}(\Phi, \mathcal{M}, u, \Lambda)$ of Algorithm~\ref{alg:coupling} and let $\widehat{X}, \widehat{Y}$ be the (random) output of Algorithm~\ref{alg:coupling} for the run $R$. For $j \in \{1,2,\ldots, \ell\}$, let $F_{j}(1)$ be the event that there is $v \in \var(c_j) \cap \mathcal{V}_{\mathrm a}$ such that $v \in \mathcal{V}_{\mathrm {set}}(R)$ and $\widehat{X}(v) \ne \widehat{Y}(v)$,  let $F_{j}(2)$ be the event that $\var(c_j) \cap \mathcal{V}_{\mathrm a} \subseteq \mathcal{V}_{\mathrm {set}}(R)$ and $c_j$ is unsatisfied by $\widehat{X}$, and let $F_{j}(3)$ be the event that $\var(c_j) \cap \mathcal{V}_{\mathrm a} \subseteq \mathcal{V}_{\mathrm {set}}(R)$ and $c_j$ is unsatisfied by $\widehat{Y}$. In light of Proposition~\ref{prop:coupling:properties}, we have  $[c_1, \ldots, c_\ell \in \faild(R) \cup \failu(R)] = \bigcap_{j = 1}^\ell (F_{j}(1) \cup F_{j}(2) \cup F_{j}(3))$. We obtain
    \begin{equation} \label{eq:fail}
       \pr\left(\bigcap\nolimits_{j = 1}^\ell (F_{j}(1) \cup F_{j}(2) \cup F_{j}(3)) \right) \le \sum_{(i_1, i_2, \ldots, i_\ell) \in \{1,2,3\}^\ell} \pr\left( \bigcap\nolimits_{j = 1}^\ell F_{j}(i_j) \right).
    \end{equation}
    We note that $F_{j}(1) = \bigcup_{v \in \var(c_j) \cap \mathcal{V}_{\mathrm a}} D_v(1)$, see Definition~\ref{def:dvj}. Let $(i_1, i_2, \ldots, i_\ell) \in \{1,2,3\}^\ell$, and let $I_1 = \{j : i_j = 1\}$, $I_2 = \{j : i_j = 2\}$ and $I_3 = \{j : i_j = 3\}$. If the event $\bigcap_{j \in I_1} F_{j}(i_j)$ holds, then, for each  $j \in I_1$  there is a  variable $u_j \in \var(c_j) \cap \mathcal{V}_{\mathrm a}$ such that $D_{u_j}(1)$ holds. Thus, for the set of tuples $T = \prod_{j \in I_1} (\var(c_{j}) \cap \mathcal{V}_{\mathrm a})$, where $\prod$ here denotes the cartesian product of sets, we have
    \begin{equation} \label{eq:I1}
        \bigcap_{j \in I_1} F_{j}(i_j) = \bigcup_{(u_1, u_2, \ldots, u_{|I_1|}) \in T} \bigcap_{j\in I_1}  D_{u_j}(1).
    \end{equation}
    Now we explain how we bound $\pr\left(\left(\bigcap_{j \in I_2 \cup I_3}  F_{j}(i_j)\right)\cap \left(\bigcap_{j\in I_1} D_{u_j}(1) \right) \right)$ for a tuple $(u_1, u_2, \ldots, u_{|I_1|}) \in T$. We are going to show that
    \begin{equation}\label{eq:bound-intersection}
        \pr\left( \left(\bigcap\nolimits_{j \in I_2 \cup I_3}  F_{j}(i_j)\right)\cap \left(\bigcap\nolimits_{j\in I_1} D_{u_j}(1) \right)  \right) \le \left(\frac{\exp(1 / k)}{2}\right)^{(k-3)r_0 \lvert I_2 \cup I_3\rvert}  \left(\frac{2}{k 2^{r_0 k}}\right)^{\lvert I_1 \rvert}.
    \end{equation}
    The proof of \eqref{eq:bound-intersection} is not as straightforward as it may seem at first glance due to the dependencies among the events $F_{j}(i_j)$, $D_{u_j}(1)$. The key idea is re-writing the LHS of \eqref{eq:bound-intersection} as in the statement of Lemma~\ref{lem:prob-V}. Indeed we note that for each $j \in I_2$ and for each variable $v \in \var(c_j) \cap \mathcal{V}_{\mathrm a}$, the event $F_j(2)$ implies that there is $i_v \in \{2,3\}$ such that $D_v(i_v)$ holds, concluding $F_j(2) = \bigcap_{v \in \var(c_j) \cap \mathcal{V}_{\mathrm a}} D_v(i_v)$, see Definition~\ref{def:dvj}. Analogously, for each $j \in I_3$ and for each variable $v \in \var(c_j)  \cap \mathcal{V}_{\mathrm a}$, we find $i_v \in \{4,5\}$ such that $F_j(3) = \bigcap_{v \in \var(c_j) \cap \mathcal{V}_{\mathrm a}} D_v(i_v)$. Therefore, we have
    \begin{equation} \label{eq:re-written}
    \left(\bigcap\nolimits_{j \in I_2 \cup I_3} F_{j}(i_j)\right) \cap \left(\bigcap\nolimits_{j \in I_1} D_{u_j}(1) \right) = \bigcap_{v\in V_f} D_v({i_v}),
    \end{equation}    
    where $V_f$ contains exactly all the auxiliary variables in the clauses $c_j$ with $j \in I_2 \cup I_3$ and the variables $u_1, u_2, \ldots, u_{|I_1|}$. Recall now that each good clause contains at least $r_0(k-3)$ auxiliary variables, and, thus, the bound given in \eqref{eq:bound-intersection} follows from \eqref{eq:re-written} and Lemma~\ref{lem:prob-V}. Combining \eqref{eq:bound-intersection}, \eqref{eq:I1} and \eqref{eq:fail}, and counting the number of tuples in $T$, we conclude that 
    \begin{equation*}
    \begin{aligned}
        \pr\left(\bigcap\nolimits_{j = 1}^\ell (F_{j}(1) \cup F_{j}(2) \cup F_{j}(3)) \right) & \le \sum_{(i_1, i_2, \ldots, i_\ell) \in \{1,2,3\}^\ell} k^{\lvert I_1 \rvert} \left(\frac{\exp(1 / k)}{2}\right)^{(k-3)r_0 \lvert I_2 \cup I_3\rvert}  \left( \frac{2}{k 2^{r_0 k}}\right)^{\lvert I_1 \rvert} \\
        & \le \sum_{(i_1, i_2, \ldots, i_\ell) \in \{1,2,3\}^\ell} \left(\frac{e 2^{3r_0}}{2^{k r_0}}\right)^{ \lvert I_2 \cup I_3\rvert}  \left( \frac{2}{ 2^{r_0 k}}\right)^{\lvert I_1 \rvert} \\ 
        & =  \left(  \frac{e 2^{3r_0}}{2^{k r_0}} +   \frac{e 2^{3r_0}}{2^{k r_0}} +  \frac{2}{ 2^{r_0 k}} \right)^{\ell},
    \end{aligned}
    \end{equation*}
    where we used the multinomial theorem. The result now follows from $ 2 e 2^{3r_0} + 2 \le 2^4$.
  % Python code: 
  %  r = 0.117841
  %  2 +  2 * math.exp( 1+ math.log(2) * r * 3 )
  % gives at most 8.94617154714107
\end{proof}

Following~\cite{galanis2019counting} and motivated by Lemma~\ref{lem:bound-fail}, we introduce the combinatorial structure that we use in our proof of Lemma~\ref{lem:si} to bound the expected number of failed clauses. 

\begin{definition}[$G^{\le k}$, $\mathcal{D}_3(G_\Phi, c, \ell)$]
    For a graph $G = (V, E)$ and a positive integer $k$, let $G^{\le k}$ be the graph with vertex set $V$ in which vertices $u$ and $v$ are connected if and only if there is a path from $u$ to $v$ in $G$ of length at most $k$. Given the graph $G_\Phi$, a clause $c$ and a positive integer $\ell$, let $\mathcal{D}_3(G_\Phi, c, \ell)$ be the set of subsets $T \subseteq V(G_\Phi)$ such that the following holds:
    \begin{enumerate}
        \item $\lvert T \rvert = \ell$ and $c \in T$;
        \item for any $c_1, c_2 \in T$, $\var(c_1) \cap \var(c_2) \cap \gv = \emptyset$;
        \item the graph $G_\Phi^{\le 3}[T]$, which is the subgraph of $G_\Phi^{\le 3}$  induced by $T$, is connected;
        \item we have $\lvert T \cap \mathcal{C}_{good} \rvert \ge (1-8/k) \ell$. 
    \end{enumerate}
\end{definition}

In~\cite{galanis2019counting} the authors consider connected sets in $G_\Phi^{\le 4}$ instead of  $G_\Phi^{\le 3}$. Here we manage to perform our union bound on $\mathcal{D}_3(G_\Phi, c, \ell)$ thanks to the fact that the set of failed clauses is connected in our refinement of the coupling process.

\begin{lemma}[{\cite[Corollary 8.19]{galanis2019counting} for $G^{\le 3}$}] \label{lem:bound-le3}
  Let $G = (V, E)$ be a connected graph, let $v \in V$ and let $\ell$ be a positive integer. Let $n_{G, \ell}(v)$ denote the number of connected induced subgraphs of $G$ with size $\ell$ containing $v$. Then, for $\ell' = \min \{3\ell, \lvert V \rvert \}$, we have $n_{G^{\le 3}, \ell}(v) \le 2^{\ell'} n_{G, \ell'}(v)$.
\end{lemma}
\begin{proof}
  Let $T$ be a connected subgraph of $G^{\le 3}$ with size $\ell$ containing $v$. We claim that, for all positive $\ell$, we can find a connected subset $H$ of $G$ with size $\ell' = \min \{3\ell, \lvert V \rvert \}$ containing $T$. The proof is straightforward by induction on $\ell$, see~\cite[Lemma 8.18]{galanis2019counting} for the analogous result on $G^{\le 4}$. We note that there are at most $\binom{\ell'}{\ell -1} \le 2^{\ell'}$ subsets $T$ of $H$ containing $v$ that could be mapped to $H$ by the previous construction. Hence, we conclude that $n_{G^{\le 3}, \ell}(v) \le 2^{\ell'} n_{G, \ell'}(v)$ as we wanted.
\end{proof}

\begin{lemma}[{\cite[Lemma 7.9]{galanis2019counting} for $\mathcal{D}_3(G_\Phi, c, \ell)$}] \label{lem:bound-I}
    Let $\ell$ be  an integer which is at least $\log n$. W.h.p. over the choice of $\Phi$, every clause $c \in \gc$ has the property that the size of $\mathcal{D}_3(G_\Phi, c, \ell)$ is at most $(18 k^2 \alpha)^{3 \ell}$.
\end{lemma}
\begin{proof}
    This follows from bounding the number of connected sets of size $\ell$ in $G_\Phi^{\le 3}$ that contain $c$ by combining Lemmas~\ref{lem:bound-Z} and~\ref{lem:bound-le3}. 
\end{proof}

We have now all the tools that we need to bound the expected number of failed clauses in the coupling process given in Algorithm~\ref{alg:coupling} and complete the proof of Lemma~\ref{lem:si}.

\begin{lemsi}
   \statelemsi
\end{lemsi}
\begin{proof} \label{lem:si:proof}
  Let $u \in \mathcal{V}_{\mathrm m}$ and $\Lambda \colon S \to \{\mathsf{F}, \mathsf{T}\}$ with $S \subseteq \mathcal{V}_{\mathrm m} \setminus \{u\}$. First of all, we apply Lemma~\ref{lem:coupling:si} to bound $\sum_{v \in \mathcal{V}_{\mathrm m} \setminus (S \cup \{u\}) } \left| \mathcal{I}^\Lambda(u \to v) \right|$  by $2^{-r_0 k+1}  \mathbb{E}  \left[\lvert  \failu(R) \rvert \right]$, where $R \sim \tau_{\mathcal{R}}(\Phi, \mathcal{M}, u, \Lambda)$. In the rest of this proof we show that $\pr(\lvert \failu(R) \rvert \ge 2k^4 \log n) \le O(1/n)$ and, thus, for large enough $n$, $\mathbb{E} \left[\lvert  \failu(R) \rvert \right] = \sum_{R \in \mathcal{R}(\Phi, \mathcal{M}, u, \Lambda)} \pr(R) \lvert  \failu(R) \rvert \le 4 k^4 \log n$. Putting all this together, and using the fact that $8k^4 \le 2^{\delta k}$ for large enough $k$ (here $\delta = \deltadef$) we would obtain the bound $\sum_{v \in \mathcal{V}_{\mathrm m} \setminus (S \cup \{u\}) } \left| \mathcal{I}^\Lambda(u \to v) \right| \le 8 \cdot 2^{-r_0 k} k^4 \log n \le 2^{-(r_0-\delta) k} \log n$ and, thus, the result would follow.

So to finish we just need to show that, w.h.p.{} over the choice of~$\Phi$, $\pr(\lvert \failu(R) \rvert \ge 2k^4 \log n) \le O(1/n)$.  Let $L = \lceil 2k^4 \log n \rceil$ and let $\ell = \lceil 0.5 k^4 \log n \rceil$. First, we are going to show that, w.h.p. over the choice of $\Phi$, the following holds:
  \begin{equation} \label{eq:ind}
  \begin{aligned}  
      \text{if } Z \subseteq \mathcal{C} \text{ is connected and } \lvert Z \rvert = L  & \text{, then } \exists c \in Z \cap \gc \text{ and } T \in \mathcal{D}_3(G_\Phi, c, \ell) \text{ with } T \subseteq Z.
  \end{aligned}
  \end{equation}
  In order to prove \eqref{eq:ind}, we are going to find a large independent set of $Z \cap \gc$, and we are going to extend it with some clauses in $Z \cap \bc$ to obtain $T \in \mathcal{D}_3(G_\Phi, c, \ell)$. We need three results that hold w.h.p. over the choice of $\Phi$: Lemmas~\ref{lem:bad},~\ref{lem:bound-Y} and~\ref{lem:tree-excess}. We note that we can apply Lemma~\ref{lem:bad} for $r= r_0 - \delta$ as our density satisfies $\alpha \le 2^{r_0 k / 3} / k^3 \le \lceil 2^{(r_0 - \delta)k} \rceil /k^3 = \Delta_r / k^3$, where $\delta = \deltadef$.  For $Z$ as in \eqref{eq:ind} we have $\lvert Z \rvert \ge 2k^4 \log n$, so by Lemma~\ref{lem:bound-Y} with $a = 2k^4$, we find that $\lvert \var(Z) \rvert \ge 2k^4 \log n$ and, thus, in light of Lemma~\ref{lem:bad}, we conclude that  $\lvert Z \cap \gc \rvert \ge (1-1/k) \lvert Z \rvert$ and $\lvert Z \cap \bc \rvert \le \lvert Z \rvert / k$. From Lemma~\ref{lem:tree-excess} with $b = 4k^4$, w.h.p. over the choice of $\Phi$, all connected sets of clauses with size at most $4k^4 \log n$ have tree-excess at most $t := \max\{1,  8 k^4 \log(e k^2 \alpha) \}$. 
  Thus, we can find $U \subseteq Z \cap \gc$ such that $U$ is a forest (disjoint union of trees) and $\lvert U \rvert \ge (1-1/k) \lvert Z \rvert - t$. In particular, $U$ is bipartite, so there is $I \subseteq U$ such that $\var(c) \cap \var(c') = \emptyset$ for all $c, c' \in I$ and $\lvert I \rvert \ge \lvert U \rvert / 2 \ge  (1-1/k)L/2 - t/2\ge \tfrac{1}{2} k^4 \log n$, where the last inequality holds for large enough $n$. Let $I'$ be an independent set of  $Z \cap \gc$ with the largest possible size. Then we have shown that $\lvert I' \rvert \ge \ell = \lceil \tfrac{1}{2} k^4 \log n \rceil$. 
  
  We claim that the set $T' :=  I' \cup (Z \cap \bc)$ is connected in $(G_\Phi[Z])^{\le 3}$, where $G_\Phi[Z]$ is the subgraph of $G_\Phi$ induced by $Z$.  Assume for contradiction that $T'$ is not connected in $(G_\Phi[Z])^{\le 3}$. In this case, we can write $T' = S_1 \cup S_2$ such that for all $c_1 \in S_1$ and $c_2 \in S_2$, the shortest path between $c_1$ and $c_2$ through clauses in $Z$ has length at least $4$. Let $(c_1,c_2) \in S_1 \times S_2$ be the pair with the shortest path in $Z$, and let this path be $c_1 = e_1, e_2, \ldots, e_j = c_2$. Then $j \ge 5$ and $e_2, \ldots, e_{j-1} \in Z \setminus T'$. Moreover, we find that $\var(e_3) \cap \var(c) = \emptyset$ for all $c \in T'$ as otherwise $e_1, e_2, \ldots, e_j$ would not be the shortest path between $S_1$ and $S_2$. Moreover, since $T'$ contain all bad clauses in $Z$, we conclude that $e_3$ is a good clause. 
  It follows that $I' \cup \{e_3\}$ is an independent set of good clauses of $Z$, which contradicts the fact that $I'$ has largest possible size among such sets.  
  
  Finally, as $\lvert T' \rvert \ge \ell$, we can find a good clause $c$ and a subset $T$ of $T'$ with size $\ell$ such that $c \in T$, $T$ is connected in $G_\Phi^{\le 3}$ and $\lvert T \cap \bc \rvert \le \lvert Z \cap \bc \rvert \le L /k  \le 8 \ell / k $. We conclude that $T \in \mathcal{D}_3(G_\Phi, c, \ell)$. This finishes the proof of \eqref{eq:ind}. 
  
  In the rest of the proof we use \eqref{eq:ind} to bound $\pr(\lvert \failu(R) \rvert \ge L)$. Recall   that the set of failed clauses $\faild(R) \cup \failu(R)$ is connected (Proposition~\ref{prop:coupling:properties}). If $\lvert\failu(R) \rvert  \ge L$, then there is $Z \subseteq \faild(R) \cup \failu(R)$ with $\lvert Z \rvert = L$ such that $Z$ is connected in $G_\Phi$, and, thus, we can find $c$ and $T$ as in \eqref{eq:ind}. We have shown that the event $\lvert\failu(R) \rvert  \ge L$  is contained in the event that there is a good clause $c$ and $T \in \mathcal{D}_3(\Phi, c, \ell)$ such that $T \subseteq \faild(R) \cup \failu(R)$. As a consequence, 
  we have
\begin{equation*}
\begin{aligned}    
    \pr\left[ \lvert \failu(R) \rvert  \ge L \right] & \le \sum_{c \in \mathcal{C}} \  \sum_{T \in \mathcal{D}_3(\Phi, c, \ell)}  \pr \left[ T \subseteq \faild(R) \cup \failu(R) \right] \\
    & \le \sum_{c \in \mathcal{C}} \  \sum_{T \in \mathcal{D}_3(\Phi, c, \ell)}  \pr \left[ T \cap \gc \subseteq \faild(R) \cup \failu(R) \right].
\end{aligned}
\end{equation*}
We note that for any $T\in \mathcal{D}_3(\Phi, c, \ell)$ there is at most one good clause $c'$ that contains the marked variable $u$. Thus, by definition of $\mathcal{D}_3(\Phi, c, \ell)$, there are at least $(1-8/k) \ell - 1$ good clauses in $T$ that do not contain the variable $u$. 
Hence, we can apply Lemma~\ref{lem:bound-I} on the size of $\mathcal{D}_3(\Phi, c, \ell)$ and Lemma~\ref{lem:bound-fail} on the probability of good clauses (that do not share good variables) failing to further obtain 
\begin{equation*}
\begin{aligned}    
    \pr\left[ \lvert \failu(R) \rvert  \ge L \right]
    & \le m \left(18 k^2 \alpha\right)^{3 \ell} 2^{-(r_0 k-4) [ (1-8/k) \ell - 1 ]}.
\end{aligned}
\end{equation*} 
In what follows it is essential that $\alpha \le 2^{r_0 k / 3} / k^3$, and this is the only proof in this paper where we need this bound on the density -- other proofs only require the less restrictive bounds $\alpha \le 2^{(r_0 - \delta)k} / k^3$ or $\alpha \le 2^{(r_0 - 3\delta)k} / k^3$. Thus, we conclude that
\begin{equation*}
\begin{aligned}    
    \pr\left[ \lvert \failu(R) \rvert  \ge L \right]
    & \le m \left(18 \frac{2^{r_0 k /3}}{k}\right)^{3 \ell} 2^{-(r_0 k-4) (1-8/k) \ell} \ 2^{r_0 k-4} =  m \left(\frac{18^3 }{ k^3} 2^{8 r_0  + 4(1-8/k)} \right)^{\ell} 2^{r_0 k - 4}.
\end{aligned}
\end{equation*}
Finally, for large enough $k$ we find that  $\pr\left[ \lvert \failu(R) \rvert  \ge L \right] \le m e^{-\ell} 2^{r_0 k} \le  m  n^{-0.5 k^4}  2^{r_0 k} = O(1/n)$ as we wanted. 
\end{proof}

\subsection{Mixing time of the $\rho$-uniform-block Glauber dynamics} \label{sec:mixing-time:mt}

Finally, we combine the results in this section with Lemma~\ref{lem:block-glauber} to complete the proof of Lemma~\ref{lem:mixing-time}.

\begin{remark} \label{rem:marginally-bounded}
  The distribution $\restr{\mu_{\Omega}}{\mathcal{V}_{\mathrm m}}$ on assignments of the marked variables (Definition~\ref{def:marginal}) is $b$-marginally bounded for $b = 1 - (1/2)\exp(1/k)$ by Proposition~\ref{prop:marginals} (or, equivalently, Lemmas~\ref{lem:marking} and~\ref{lem:marginals}). Since $\exp(1/k) \le 1 + 2/k$, we have $b \ge 1/2 - 1/k \ge 1/e$ for $k \ge 8$.
\end{remark}

\begin{lemmt}
 \statelemmt
\end{lemmt}
\begin{proof}  \label{lem:mixing-time:proof}
  In view of Lemma~\ref{lem:si}, as $\alpha \le 2^{0.039 k} \le 2^{r_0 k / 3}/k^3$ for large enough $k$, w.h.p. over the choice of $\Phi$, the distribution $\restr{\mu_\Omega}{\mathcal{V}_{\mathrm m}}$ is $\eta$-spectrally independent for $\eta = 2^{-(r_0-\delta) k} \log n$. Moreover, this distribution is $b$-marginally bounded for $b = 1/e$ when $k \ge 8$. We are going to apply Lemma~\ref{lem:block-glauber} with $V = \mathcal{V}_{\mathrm m}$, $\mu = \restr{\mu_\Omega}{\mathcal{V}_{\mathrm m}}$, $M = \lvert \mathcal{V}_{\mathrm m} \rvert$ and $\kappa = 2^{-k-1}$. First, we check that the hypothesis $M \ge \frac{2}{\kappa} (4\eta / b^2 + 1)$ of Lemma~\ref{lem:block-glauber} holds. By Corollary~\ref{cor:size-r-distributed} with $r = r_0 - \delta$ and $V = \mathcal{V}_{\mathrm m}$, we have $M \ge (r_0 - \delta) (k \alpha/\Delta_r) n = \Omega(n)$, so  $M \ge \frac{2}{\kappa} (4\eta / b^2 + 1)$ holds for large enough $n$ as $\frac{2}{\kappa} (4\eta / b^2 + 1) = O(\log n)$. Hence, we can apply Lemma~\ref{lem:block-glauber} to obtain
    \begin{equation*}
    \begin{aligned}
      T_{\mix}(\rho, \varepsilon/2) & \le \left\lceil C_\rho \frac{M}{\rho} \left( \log \log \frac{1}{\mu_{\min}} + \log \frac{2}{\varepsilon^2} \right) \right\rceil,
    \end{aligned}
  \end{equation*}
  where $\rho = \lceil \kappa M \rceil$ and $C_\rho = \left(2/\kappa \right)^{4 \eta / b^2 + 1}$. We have
  \begin{equation*}
      C_\rho = \exp \left( (\log 2) (k+2) \Big(\frac{4 \eta}{b^2} + 1\Big)  \right)\le 2^{k+2} \exp \left( \frac{(\log 2)(\log n) (k+2) 4 e^2 }{2^{(r_0-\delta) k}}  \right),
  \end{equation*}
  so there exists a function $k_0(\theta) = \Theta(\log(1/\theta))$ such that when $k \ge k_0(\theta)$, we have $C_\rho \le 2^{k+ 2} n^\theta$. In light of Remark~\ref{rem:marginally-bounded}, we have $\mu_{\min} \ge b^M$, so $\log \log (1/ \mu_{\min}) \le \log (M \log(1/b) )= \log M$ as $b =1/e$. Thus, we conclude that
  \begin{equation*}
      T_{\mix}(\rho,\varepsilon/2) \le \left\lceil 2^{2k+ 3} n^\theta  \left( \log M + \log \frac{2}{ \varepsilon^2} \right)  \right\rceil \le \left\lceil 2^{2k+3} n^\theta \log \frac{2n}{\varepsilon^2}\right\rceil. \qedhere
  \end{equation*}
\end{proof}

\section{Proof of Theorem~\ref{thm:sampling}} \label{sec:alg}

In this section we complete the proof of Theorem~\ref{thm:sampling}. The proofs in this section do not present any challenging steps. In fact, they amount to combining the main technical results that have already been proved in this work. We start by showing that the calls to the method  $\operatorname{Sample}$ in Algorithm~\ref{alg:main}  are unlikely to ever return error, that is, the connected components of $G_{\Phi^{\Lambda}}$ have size at most $2k^4(1+\xi)\log(n)$ almost every time the method is called. As pointed out in our proof outline, this is a straightforward consequence of Lemma~\ref{lem:cc-general} and the fact that the probability distribution of the output of the Glauber dynamics is $(1/k)$-uniform (Corollary~\ref{cor:marginals:glauber}).

\newcommand{\statelemcc}{ 
 Let $\theta \in (0,1)$. There is an integer $k_0 \ge 3$ such that, for any integers $k \ge k_0$, $\xi \ge 1$ and any density $\density \le 2^{(r_0 - 3\delta)k}/k^3$, the following holds w.h.p. over the choice of $\Phi = \Phi(k, n, \lfloor \density n\rfloor)$. In the execution of Algorithm~\ref{alg:main} with input $\Phi$, with probability at least $1 - n^{-3\xi}$ over the random choices made by Algorithm~\ref{alg:main}, every time that the algorithm calls the method $\operatorname{Sample}(\Phi^\Lambda, S)$, the connected components of $G_{\Phi^{\Lambda}}$ have size at most $2k^4(1+\xi)\log(n)$.}

\begin{lemma}\label{lem:cc}
    \statelemcc
\end{lemma}
\begin{proof} \label{lem:cc:proof}
Let $\varepsilon = n^{-\xi}$ and let $T = \mtdef$ be the mixing time established in Lemma~\ref{lem:mixing-time}. Note that  $\density \le 2^{(r_0 - 3\delta)k}/k^3 \le 2^{(r_0 - \delta)k}/k^3$, so we an indeed compute the marking $(\mathcal{V}_{\mathrm m}, \mathcal{V}_{\mathrm a}, \mathcal{V}_{\mathrm c})$ in Algorithm~\ref{alg:main} with the help of Lemma~\ref{lem:marking}. We need $\density \le 2^{(r_0 - 3\delta)k}/k^3$ so that we can apply Lemma~\ref{lem:cc-general} with $r= r_0 - \delta$.
Algorithm~\ref{alg:main}  calls the method $\operatorname{Sample}$ exactly $T+1$ times in total: $T$ times in line~\ref{line:sample1} (when simulating the $\rho$-uniform-block Glauber dynamics) and one  time in line~\ref{line:sample2} to extend the assignment $X_T$ of marked variables to all variables. 
 
 Let $t \in \{0,1,\ldots, T\}$ and let $\pi_t$ be the probability distribution of $X_t$, where $X_t$ is the state of the $\rho$-block-uniform Glauber dynamics on the marked variables described in Algorithm~\ref{alg:main} after $t$ steps. Recall that $\rho = \lceil 2^{-k-1} \lvert \mathcal{V}_{\mathrm m} \rvert \rceil$ and that $X_0$ is chosen uniformly at random. First, we focus on the case $t < T$. We are going to apply Lemma~\ref{lem:cc-general} with $r=r_0 - \delta$, $a = 2k ^4$, $b = 2 a (1+\xi)$, $V = \mathcal{V}_{\mathrm m}$, $\mu = \pi_t$ and this choice of $\rho$. The set $\mathcal{V}_{\mathrm m}$ is $r_0$-distributed by the definition of $(r_0 - \delta, r_0, r_0, 2r_0)$-marking (Definition~\ref{def:distributed-marking}). Moreover, $\pi_t$ is $(1/k)$-uniform by Corollary~\ref{cor:marginals:glauber}, and we have $\rho \le \lvert \mathcal{V}_{\mathrm m} \rvert/2^k$. Hence, we can indeed apply Lemma~\ref{lem:cc-general}.  Consider the following  experiment described in Lemma~\ref{lem:cc-general} for $L = \lceil a (1+\xi) \log n \rceil$, which satisfies $a \log n \le L \le b \log n$. First, draw $S \subseteq \mathcal{V}_{\mathrm m}$ from the uniform distribution $\tau$ over subsets of $\mathcal{V}_{\mathrm m}$ with size $\rho$. Then, sample an assignment $\Lambda_{t+1}$ from  $\restr{\pi_t}{\mathcal{V}_{\mathrm m} \setminus S}$, the marginal of $\pi_t$ on $\mathcal{V}_{\mathrm m} \setminus S$. Denote by $\mathcal{F}$ the event that that there is a connected set of clauses $Y$ of $\Phi$ with $\lvert Y \rvert \ge  L$  such that all clauses in $Y$ are unsatisfied by $\Lambda_{t+1}$. Then we have 
\begin{equation} \label{eq:prob-error}
    \pr_{S \sim \tau} \left( \pr_{\Lambda_{t+1} \sim \restr{\pi_t}{\mathcal{V}_{\mathrm m} \setminus S}} \left( \mathcal{F} \right) \le 2^{- \delta k L}  \right) \ge 1 - 2^{-\delta k L}.
\end{equation}
Alternatively, this experiment is the same as first sampling an assignment $X_t$ of all variables in $\mathcal{V}_{\mathrm m}$ from $\pi_t$, and then restricting the assignment to a random set $S \sim \tau$, obtaining $\Lambda_{t+1}$. Note that this exact experiment occurs before calling the method $\operatorname{Sample}$ in the $t$-th step of the $\rho$-uniform-block Glauber dynamics in Algorithm~\ref{alg:main}. Thus, in light of \eqref{eq:prob-error}, the probability that in  step~$t+1$ of the execution of Algorithm~\ref{alg:main} the graph $G_{\Phi^{\Lambda_{t+1}}}$ has a connected component with size at least $L$ is at most $2^{-\delta k L}+2^{-\delta k L}$, where the first $2^{-\delta k L}$ comes from the probability of choosing a wrong set $S \sim \tau$ in \eqref{eq:prob-error} and the second $2^{-\delta k L}$ comes from the bound on the probability of the event $\mathcal{F}$ once we have chosen $S$. We have shown that with probability at least $1- {2\cdot} 2^{- \delta k L}$, 
all the connected components of the graph $G_{\Phi^{\Lambda_{t}}}$ appearing in 
step~$t+1$ of
the execution of Algorithm~\ref{alg:main} have size less than $L$. We have $ {2 \cdot} 2^{- \delta k L} \le {2 \cdot} n^{- \delta k a (1+\xi) \log 2} \le n^{-5 \xi}$ for large enough $k$, so the probability that $\operatorname{Sample}$ returns error at step $t+1$ is at most $n^{-5\xi}$. The case $t = T$ is analogous, the only difference here is that we call $\operatorname{Sample}$ on $\Phi^{X_T}$, where $X_T \sim \pi_T$ is an assignment of all marked variables, so we apply Lemma~\ref{lem:cc-general} with $\rho = 0$ instead of $\rho = \lceil 2^{-k-1} \lvert \mathcal{V}_{\mathrm m} \rvert \rceil$.

Finally, we carry out a union bound over $t \in \{0, 1, \ldots, T\}$, so the probability that any of the calls to $\operatorname{Sample}$ returns error is at most $(T+1) n^{-5\xi} \le n^{-3\xi}$ for large enough $n$ as $T = O(n^\theta \log n) = O(n \log n)$.
\end{proof}

Once we have established Lemmas~\ref{lem:mixing-time},~\ref{lem:sample}, and~\ref{lem:cc},  the proof of Theorem~\ref{thm:sampling} follows as below.

\begin{thmsampling} 
  \statethmsampling
\end{thmsampling}

\begin{proof} \label{thm:sampling:proof}
Let $k_0(\theta) = \Theta(\log (1/\theta))$ be large enough so that, for all integers $k \ge k_0(\theta)$, $\xi \ge 1$ and all densities $\alpha \le 2^{0.039 \cdot k}$, the conclusions of 
Lemmas~\ref{lem:marking},~\ref{lem:mixing-time},~\ref{lem:sample}, and~\ref{lem:cc}  hold w.h.p. over the choice of the random $k$-CNF formula $\Phi = \Phi(k, n, \lfloor \alpha n \rfloor)$. These lemmas are enough to analyse Algorithm~\ref{alg:main} and tackle this proof. We analyse the distribution $\mu_{alg}$ of the output of Algorithm~\ref{alg:main}. This distribution outputs either a satisfying assignment of the input formula $\Phi$ or $\textit{error}$. 
Let $\epsilon = n^{-\xi}$. Let $\mathcal{E}$ be the event that running Algorithm~\ref{alg:main} outputs $\textit{error}$. This happens with probability at most $\varepsilon/4$ when computing the marking $(\mathcal{V}_{\mathrm m}, \mathcal{V}_{\mathrm a}, \mathcal{V}_{\mathrm c})$ in line 2 of the algorithm,
and in lines 7 and 10 if the method $\operatorname{Sample}(\hat{\Phi}, S)$ returns error, which occurs when $G_{\hat{\Phi}}$ has a connected component with size more than $2k^4(1+\xi)\log(n)$. In view of Lemma~\ref{lem:cc}, the probability that Algorithm~\ref{alg:main} outputs \emph{error} due to the failure of the method $\operatorname{Sample}$ is at most $n^{-3\xi}=\varepsilon^3$. We conclude that the probability that the algorithm outputs error is bounded above by $\varepsilon / 4 + \varepsilon^3 \le \varepsilon / 2$ .
Let $\mu_{Glauber}$ be the distribution that Algorithm~\ref{alg:main} would output if there were no errors
(that is, the distribution assuming that the method Sample always
outputs from the appropriate distribution). Then $d_{\TV}(\mu_{alg}, \mu_{Glauber})$
is the probability that an error occurs, which is  
at most~$\varepsilon / 2$. Let $\pi_{Glauber}$ be the distribution output by the $\rho$-uniform-block Glauber dynamics on $\mathcal{V}_{\mathrm m}$ after $T$ steps. By Lemma~\ref{lem:mixing-time} on the mixing time of the Glauber dynamics, we have $d_{\TV}(\pi_{Glauber}, \restr{\mu_\Omega}{\mathcal{V}_{\mathrm m}}) \le \varepsilon / 2$. 
As $\mu_{Glauber}$ comes from sampling an assignment $X_T$ from $\pi_{Glauber}$ and then completing $X_T$ to all $\mathcal{V}$ by sampling from $\mu_{\Omega}(\cdot \vert X_T)$, we have $d_{\TV}(\mu_{Glauber}, \mu_{\Omega}) \le d_{\TV}(\pi_{Glauber}, \restr{\mu_\Omega}{\mathcal{V}_{\mathrm m}}) \le \varepsilon/2$. We find that $d_{\TV}(\mu_{alg}, \mu_\Omega) \le d_{\TV}(\mu_{alg}, \mu_{Glauber}) + d_{\TV}(\mu_{Glauber}, \mu_\Omega) \le  \varepsilon / 2 + \varepsilon / 2 = \varepsilon$
  as we wanted. The running time of Algorithm~\ref{alg:main} is now easily obtained by adding up the running times of the following subroutines. The good clauses and good variables are computed in time $O(n + k m) = O(n)$, see Proposition~\ref{prop:bad}. The marking $(\mathcal{V}_{\mathrm m}, \mathcal{V}_{\mathrm a}, \mathcal{V}_{\mathrm c})$ is computed with probability at least $1-\varepsilon/4$ in time  $O(n \Delta_r k^2 \log(4/\varepsilon)) = O(n \log n)$, see  Lemma~\ref{lem:marking}. Recall that there are $T+1 = O(n^\theta \log (n/\varepsilon^2) ) =  O(n^{\theta} \log n )$ calls to the method $\operatorname{Sample}(\Phi', S)$, and each call takes time $O(\vert S \rvert \log n) = O(n \log n)$ by Lemma~\ref{lem:sample}. We conclude that  the running time of Algorithm~\ref{alg:main} is $O(n^{1+\theta} \log(n)^2 )$. The result now follows by choosing $k_1 = k_0(\theta/2)$, so the running time for $k \ge k_1$ is  $O(n^{1+\theta/2} \log(n)^2 ) = O(n^{1+ \theta})$.
\end{proof}

We have now proved that it is possible to (approximately) sample uniformly at random from the satisfying assignments of $\Phi = \Phi(k, n, \lfloor \alpha n \rfloor)$. At this point, standard techniques can be applied to obtain a randomised approximation scheme for counting the satisfying assignments of $\Phi$. These techniques are based on the self-reducibility of $k$-SAT~\cite{jerrum1986}. The following remark shows how to obtain a randomised approximation scheme that runs in time $O(n^\theta (n/\varepsilon)^2)$ following~\cite[Chapter 7]{feng2020}, where the authors base their counting algorithm on the simulated annealing method~\cite{stefankovic2009,huber2015, kolmogorov2018}. \begin{remark}[Approximate counting for random $k$-SAT formulae] \label{rem:counting}
Let $k_0(\theta)$ be the integer depending on $\theta \in(0,1)$ obtained in Theorem~\ref{thm:sampling}. Let $k_1 = k_0(\theta/2)$, let $k \ge k_1$ be an integer, let $\xi$ be a positive integer and let $\alpha \le 2^{0.039 k}$ be a density. We apply Theorem~\ref{thm:sampling} to obtain an algorithm to sample from the satisfying assignments of $\Phi = \Phi(k, n, \lfloor \alpha n \rfloor )$ within $n^{-4\xi}$ total variation distance from the uniform distribution. This algorithm runs in time $O(n^{1+\theta/2})$ and succeeds w.h.p. over the choice of $\Phi$.

Let $\varepsilon \in (0,1)$ with $\varepsilon \ge n^{-\xi}$. A modified version of the approximate counting algorithm of~\cite[Section 7]{feng2020}, using $O(\varepsilon^{-2} n \log (n / \varepsilon))$ samples from the sampling algorithm above, approximates the number of satisfying assignments of the $k$-CNF formula $\Phi$ with multiplicative error $\varepsilon$, thus, achieving running time $O(n^{\theta/2} (n/\varepsilon)^2 \log(n/\varepsilon)) = O(n^\theta (n/\varepsilon)^2)$. Here we describe these minor modifications.

Let $\Omega_{\bad}$ be the set of assignments  $X \colon \mathcal{V} \to \{\mathsf{F},\mathsf{T}\}$ that satisfy the bad clauses of $\Phi$.  For $X \in \Omega_{\bad}$, we define $F(X)$ to be the set of good clauses that are not satisfied by $X$. For $\kappa \in \mathbb{R}$, we define $w_\kappa(X) = \exp(- \kappa \lvert F(X) \rvert )$ and we define the partition function $Z(\kappa) = \sum_{X \in \Omega_{\bad}} w_\kappa(X)$. The simulated annealing algorithm of~\cite[Section 7]{feng2020} uses $Z(\kappa)$ 
(with $\Omega^*$ from Definition~\ref{def:subformula} in place of $\Omega_{\bad}$) 
to approximate the number of satisfying assignments of~$\Phi$. We note that $Z(0) = \lvert \Omega_{\bad} \rvert$, which can be computed in linear time in $n$ using the exact counting algorithm given in Proposition~\ref{prop:counting}. 
Here one has to use the fact that the connected components of $G_{\Phi'}$ for the formula $\Phi' = (\mathcal{V}, \bc)$ have size at most $2k^4 \log n$, see Lemma~\ref{lem:bad-component} from Appendix~\ref{sec:ap:bad} and Lemma~\ref{lem:bound-Y}, and the fact that these connected component have tree-excess upper bounded as a function of $k$ (Lemma~\ref{lem:tree-excess}). Once one has performed these modifications, the algorithm given in ~\cite[Section 7]{feng2020} applies without any difficulties.
\end{remark}

\section{Proof of Theorems~\ref{thm:connectivity} and~\ref{thm:looseness}} \label{sec:connectivity}

In this section we exploit Lemma~\ref{lem:cc-general} to prove Theorems~\ref{thm:connectivity} and~\ref{thm:looseness} on the connectivity and looseness of the solution space of random $k$-CNF formulae. We recall that the density threshold in Theorems~\ref{thm:connectivity} and~\ref{thm:looseness} is $\alpha \le 2^{0.227 k}$, significantly larger than our algorithmic threshold in Theorem~\ref{thm:sampling}, which is $\alpha \le 2^{0.039 k}$. In order to conclude connectivity for densities up to $2^{0.227 k}$, we let $r_1 = 0.227092$ and consider the threshold $\Delta_{r} = \lceil 2^{r k}\rceil$ for $r = r_1-\delta$ in the definition of high-degree variables instead of $\Delta_{r_0-\delta} = \lceil 2^{(r_0-\delta) k}\rceil$. In all this section we set $r = r_1 - \delta$, so we omit $r$ in the notation and we write $\gv$ instead of $\gv(r)$. We work with an $(r, r_1, 0, r_1)$-marking $(\mathcal{V}_{\mathrm m}, \emptyset, \mathcal{V}_{\mathrm c})$ (the set of auxiliary variables is empty), which we can find w.h.p. over the choice of $\Phi = \Phi(k, n, \lfloor \alpha n \rfloor)$ as in Lemma~\ref{lem:marking:r1}. Let us briefly recall some of the properties of this marking. First of all, by definition, the set $\mathcal{V}_{\mathrm m}$ is $r_1$-distributed and is a subset of $\mathcal{V}_{good}$. Moreover, the distribution $\restr{\mu_\Omega}{\mathcal{V}_{\mathrm m}}$ is $(1/k)$-uniform by Lemma~\ref{lem:marginals}. In light of Lemma~\ref{lem:cc-general} for $r = r_1 - \delta$, these properties allow us to show that, when sampling $\Lambda \sim \restr{\mu_\Omega}{\mathcal{V}_{\mathrm m}}$, the connected components of $\Phi^\Lambda$ are logarithmic in size with probability $1-o(1)$ over the choice $\Lambda \sim \restr{\mu_\Omega}{\mathcal{V}_{\mathrm m}}$. In fact, this is also the case when $\Lambda \sim \restr{\mu_\Omega}{\mathcal{V}_{\mathrm m} \setminus \{v\}}$ for any variable $v$.

\begin{corollary} \label{cor:cc-r1}
    There is an integer $k_0 \ge 3$ such that, for any integer $k \ge k_0$, any density $\density \le \alpha_1 := 2^{(r_1-3\delta) k}$, the following holds w.h.p. over the choice of $\Phi = \Phi(k, n, \lfloor \density n\rfloor)$. 

    Let $V$ be a set of good variables of $\Phi$ that is $r_1$-distributed, let $\mu$ be a $(1/k)$-uniform distribution over the assignments $V \to \{\mathsf{F}, \mathsf{T}\}$ and  let $v \in V$. Then, with probability at least $1 - n^{-k}$ over the choice $\Lambda \sim \restr{\mu}{V \setminus \{v\}}$, the connected components of $\Phi^\Lambda$ have size at most $2k^4 \log n$.
\end{corollary}
\begin{proof}
   The result is an application of Lemma~\ref{lem:cc-general} with $r = r_1-\delta$, $b= 4k^4$, $\rho = 1$ and $L = \lceil 2k^4 \log n \rceil$. We need large enough $k_0$ such that $2^{- \delta k L} \le 2^{- \delta 2k^ 5 \log n } \le n^{-k}$ for all $k \ge k_0$. For these parameters, in the setting of Lemma~\ref{lem:cc-general}, the distribution $\tau$ is the uniform distribution over the variables in $V$. The experiment in the statement of Lemma~\ref{lem:cc-general} consists in sampling $v \sim \tau$ and then sampling $\Lambda \sim \restr{\mu}{V \setminus \{v\}}$. Let $\mathcal{F}_v$ be the event, concerning the choice $\Lambda \sim \restr{\mu}{V \setminus \{v\}}$, that there is a connected set of clauses $Y$ of $\Phi$ with $\lvert Y \rvert \ge  \lceil 2k^4 \log n \rceil$ such that all clauses in $Y$ are unsatisfied by $\Lambda$. Then by Lemma~\ref{lem:cc-general} we have $\pr_{v \sim \tau} \left( \pr_{\Lambda \sim \restr{\mu}{V \setminus \{v\}}} \left( \mathcal{F}_v \right) \le 2^{- \delta k L}  \right) \ge 1 - 2^{-\delta k L}$. From $2^{-\delta k L} \le n^{-k}$, we obtain the bound $\pr_{v \sim \tau} \left( \pr_{\Lambda \sim \restr{\mu}{V \setminus \{v\}}} \left( \mathcal{F}_v \right) \le 2^{- \delta k L}  \right) \ge 1 - n^{-k}$. Since $\tau$ is the uniform distribution over the variables in $V$, for $v \sim \tau$, either the event that $\pr_{\Lambda \sim \restr{\mu}{V \setminus \{v\}}} \left( \mathcal{F}_v \right) \le 2^{- \delta k L}$ has probability $1$ or it has probability at most $1- 1/\lvert V \rvert \le 1- 1/n$. The latter option is not possible due to $\pr_{v \sim \tau} \left( \pr_{\Lambda \sim \restr{\mu}{V \setminus \{v\}}} \left( \mathcal{F}_v \right) \le 2^{- \delta k L}  \right) \ge 1 - n^{-k}$ and $k \ge 3$. Thus, we conclude that $\pr_{v \sim \tau} \left( \pr_{\Lambda \sim \restr{\mu}{V \setminus \{v\}}} \left( \mathcal{F}_v \right) \le 2^{- \delta k L}  \right) = 1$, so for any $v \in V$ we have $\pr_{\Lambda \sim \restr{\mu}{V \setminus \{v\}}} \left( \mathcal{F}_v \right) \le 2^{- \delta k L} \le n^{-k}$. That is, for any $v \in V$, with probability at least $1 - n^{-k}$ over the choice of $\Lambda \sim \restr{\mu}{V \setminus \{v\}}$ the connected components of $\Phi^\Lambda$ have size at most $L-1 = \lceil 2k^4 \log n \rceil-1 < 2k^4 \log n$ as we wanted to prove. 
\end{proof}

Our connectivity and looseness results will follow from Corollary~\ref{cor:cc-r1}. In Section~\ref{sec:connectivity:proof} we prove Theorem~\ref{thm:connectivity} and in Section~\ref{sec:looseness} we prove Theorem~\ref{thm:looseness}.

\subsection{Proof of Theorem~\ref{thm:connectivity}} \label{sec:connectivity:proof}

We consider~Algorithm~\ref{alg:find-path}  that receives two satisfying assignments of a $k$-CNF formula $\Phi$ as the input and constructs a path between them. Before introducing this algorithm, recall that the graph $H_\Phi$ is the dependency graph of the variables of $\Phi$ introduced in Definition~\ref{def:graph-H}.

\begin{algorithm}[H]
  \begin{algorithmic}[1]

    \caption{Finding a $(\operatorname{poly}(k) \log n)$-path between two satisfying assignments}\label{alg:find-path}

    \REQUIRE a $k$-CNF formula $\Phi = (\mathcal{V},\mathcal{C})$ with $n$ variables, an $(r, r_1, 0, r_1)$-marking $(\mathcal{V}_{\mathrm m}, \emptyset, \mathcal{V}_{\mathrm c})$ of $\Phi$, and two satisfying assignments $\sigma, \sigma'$.

    \STATE Let $v_1, v_2, \ldots, v_\ell$ be the variables in $\mathcal{V}_{\mathrm m}$.
    
    \STATE $\zeta_0 \leftarrow \sigma$.

    \Comment*[h]{Stage 1: Update the marked variables}\;
    
    \FOR{$i \in [\ell]$}

    \STATE Find $\zeta_i \in \Omega$ with marked variables specified by
  $
  \zeta_i (v_j) = 
  \begin{cases}
  \sigma'(v_j), & j \le i; \\
  \sigma(v_j), & j > i;
  \end{cases}
  $ \newline
  such that $\|\zeta_i - \zeta_{i-1}\|_{1}$ is minimised. \label{alg:step2b}
    
    \ENDFOR

    \STATE $\xi_0 = \zeta_\ell$

    \Comment*[h]{Stage 2: Update the rest of variables}
    
    \STATE Let $\tau' = \restr{\sigma'}{\mathcal{V}_{\mathrm m}}$ and suppose that $H_{\Phi^{\tau'}}$ has connected components $\mathcal{E}_1, \mathcal{E}_2, \ldots, \mathcal{E}_t$.

    \FOR{$i \in [t]$}
	\STATE Let $\xi_i \in \Omega$ be defined as $\xi_i(v) = 
	\begin{cases}
	\sigma'(v), & v \in \left(\mathcal{V} \setminus \bigcup_{j = 1}^t \mathcal{E}_j \right) \cup \left(\bigcup_{j = 1}^{i} \mathcal{E}_j \right); \\
	\zeta_\ell(v), & v \in \bigcup_{j = i+1}^{t} \mathcal{E}_j. 
	\end{cases}
	$
	\label{alg:step3}
    \ENDFOR
    
    \RETURN The path $\sigma = \zeta_0 \leftrightarrow \cdots \leftrightarrow \zeta_\ell = \xi_0 \leftrightarrow \cdots \leftrightarrow \xi_r = \sigma'$.
  \end{algorithmic}
\end{algorithm}

To prove Theorem~\ref{t:robustconn}, it suffices to show that the output of Algorithm~\ref{alg:find-path}  is with high probability a $D$-path in the solution space for $D =  2k^5 \log n$ for the inputs $\sigma \sim \mu_\Omega$ and $\sigma' \sim \mu_\Omega$. We will not actually require $\sigma \sim \mu_\Omega$ and $\sigma' \sim \mu_\Omega$ in the proof; instead we will just use the fact that the restrictions of $\sigma$ and $\sigma'$ on $\mathcal{V}_{\mathrm m}$ follow a $(1/k)$-uniform distribution as guaranteed by Lemma~\ref{lem:marginals}, see the proof of Lemma~\ref{lem:IV} for details.

We need the following two lemmas to establish Theorem~\ref{t:robustconn}.
The first lemma (Lemma~\ref{lem:I}) shows that all the truth assignments $\zeta_i$, $\xi_i$ in the algorithm exist and satisfy the formula (i.e. the algorithm is well-defined), implying our constructed path is indeed a valid path comprising only satisfying assignments. The second lemma  (Lemma~\ref{lem:IV}) shows that w.h.p., two adjacent assignments differ by at most $2k^5 \log n$ variables. This result is an application of Corollary~\ref{cor:cc-r1}.

\begin{lemma}\label{lem:I}
For any $k$-CNF formula $\Phi$ with $n$ variables, any $(r, r_1, 0, r_1)$-marking $(\mathcal{V}_{\mathrm m}, \emptyset, \mathcal{V}_{\mathrm c})$ of $\Phi$, and any two satisfying assignments $\sigma, \sigma'$, Algorithm~\ref{alg:find-path} on these inputs is well-defined in the following sense:
\begin{enumerate}
\item \label{item:I-1} It is always possible to implement Line~\ref{alg:step2b} such that $\zeta_i \in \Omega$. 
\item \label{item:I-2} We have $\xi_i \in \Omega$ for each $i \in [t]$.
\end{enumerate}
\end{lemma}
\begin{proof}
To prove item~\ref{item:I-1}, we are going to show that for any partial assignment $X \colon \mathcal{V}_{\mathrm m} \to \{\mathsf{F}, \mathsf{T}\}$, we have $\Pr_{\mu_\Omega}(X) > 0$ and, thus, can extend $X$ to some satisfying assignment. If this claim holds, then we can indeed compute the satisfying assignments $\zeta_1, \zeta_2, \ldots, \zeta_\ell$ in Algorithm~\ref{alg:find-path}.
Recall that the distribution $\restr{\mu_\Omega}{\mathcal{V}_{\mathrm m}}$ is $(1/k)$-uniform, see Lemma~\ref{lem:marginals}. From the definition of $(1/k)$-uniform distribution, we find that an analogous statement to Proposition~\ref{prop:marginals} holds for our $(r, r_1, 0, r_1)$-marking (here $r = r_1 - \delta$): for any $v \in \gv(r)$, any $V \subseteq \mathcal{V}_{\mathrm m}$ with $v \not \in V$, and any $\Lambda \colon V \to \{\mathsf{F}, \mathsf{T}\}$, we have
\begin{equation*}
    \max \left\{ \pr_{\mu_{\Omega^\Lambda}}\left( \left. v \mapsto \mathsf{F} \right| \Lambda \right), \pr_{\mu_{\Omega}}\left( \left. v \mapsto \mathsf{T} \right| \Lambda \right)  \right\} \le \frac{1}{2} \exp\left(\frac{1}{k}\right).
\end{equation*}
Thus, by induction on the size of a set $S \subseteq \mathcal{V}_{\mathrm m}$, we conclude that any assignment $\Lambda \colon S \to  \{\mathsf{F}, \mathsf{T}\}$ has $\Pr_{\mu_\Omega}(\Lambda) > 0$, finishing the proof of item~\ref{item:I-1}.

Next consider item~\ref{item:I-2}. Let  $\tau' = \restr{\sigma'}{\mathcal{V}_{\mathrm m}}$ as in Algorithm~\ref{alg:find-path}. All clauses that do not appear in $G_{\Phi^{\tau'}}$ are satisfied by the partial assignment $\tau'$. Now consider two satisfying assignments $\Lambda, \Lambda'$ such that $\Lambda(\mathcal{V}_{\mathrm m}) = \Lambda'(\mathcal{V}_{\mathrm m}) = \tau'$. Let $G_{\Phi^{\tau'}}$ have connected components $\mathcal{C}_1, \mathcal{C}_2, \ldots, \mathcal{C}_{t'}$. In particular, $\restr{\Lambda}{{\var(\mathcal{C}_i)}}$ and $\restr{\Lambda'}{\var(\mathcal{C}_i)}$ each satisfy all clauses in $\mathcal{C}_i$. Each clause in $\Phi^{\tau'}$ is in exactly one connected component $\mathcal{C}_i$. Consequently, any assignment $X$ such that $\restr{X}{\mathcal{V}_{\mathrm m}} = \tau'$ and $\restr{X}{\var(\mathcal{C}_i)} \in \{\restr{\Lambda}{\var(\mathcal{C}_i)}, \restr{\Lambda'}{\var(\mathcal{C}_i)}\}$ for all $i \in [t']$ is a satisfying assignment (any variables that do not appear in $\mathcal{V}_{\mathrm m} \cup \left(\bigcup_{i = 1}^{t'} \var(\mathcal{C}_i) \right)$ can be chosen arbitrarily). We note that there are two types of connected components of $H_{\Phi^\tau}$. The first type are those corresponding to $\var(\mathcal{C}_i)$ for some $i \in [t']$. The second type are those connected components with variables in $\mathcal{V} \setminus \left( \mathcal{V}_{\mathrm m} \cup \left(\bigcup_{i = 1}^{t'} \var(\mathcal{C}_i) \right) \right)$. These connected components are singleton and consist of one variable $v$ that does not appear in $\Phi^\tau$ or, equivalently, every clause of $\Phi$ containing $v$ is satisfied by $\tau$. As a consequence, taking $\Lambda = \zeta_\ell$, $\Lambda' = \sigma'$ and $X = \xi_i$ in the argument above, we conclude that $\xi_0, \xi_1, \ldots, \xi_t$ are satisfying assignments by construction in Algorithm~\ref{alg:find-path} and item~\ref{item:I-2} holds.
\end{proof}

\begin{lemma}\label{lem:IV}
  There is an integer $k_0 \ge 3$ such that, for any integer $k \ge k_0$, any density $\density \le 2^{(r_1 - 3\delta) k}$, the following holds w.h.p. over the choice of $\Phi = \Phi(k, n, \lfloor \density n\rfloor)$. In Algorithm~\ref{alg:find-path} with inputs the formula $\Phi$, an $(r, r_1, 0, r_1)$-marking of $\Phi$ and the two satisfying assignments $\sigma$ and $\sigma'$, with probability at least $1-1/n$ over the choices $\sigma \sim \mu_\Omega, \sigma' \sim \mu_\Omega$, we have
\begin{enumerate}
\item $\|\zeta_i - \zeta_{i-1}\|_1 \le 2k^5 \log n$ for all $i \in [\ell]$; \label{item:IV:1}
\item $\|\xi_i - \xi_{i-1}\|_1 \le 2k^5 \log n$ for all $i \in [t]$. \label{item:IV:2}
\end{enumerate}
\end{lemma}
\begin{proof}
Let $\Phi$ and $(\mathcal{V}_{\mathrm m}, \emptyset, \mathcal{V}_{\mathrm a})$ be the first two inputs of Algorithm~\ref{alg:find-path}, and let $v_1, v_2, \ldots, v_\ell$ be the variables in $\mathcal{V}_{\mathrm m}$ in the order considered in Algorithm~\ref{alg:find-path}. Let $\sigma \sim \mu_\Omega$ and $\sigma' \sim \mu_\Omega$. Let $\sigma = \zeta_0 \leftrightarrow \cdots \leftrightarrow \zeta_\ell = \xi_0 \leftrightarrow \cdots \leftrightarrow \xi_r = \sigma'$ be the path between $\sigma$ and $\sigma'$ output by Algorithm~\ref{alg:find-path}. In light of Lemma~\ref{lem:I}, the assignments $\zeta_0, \zeta_1, \ldots, \zeta_\ell, \xi_1, \ldots, \xi_r$ are satisfying assignments of $\Phi$. We also note that the set of marked variables $\mathcal{V}_{\mathrm m}$ is $r_1$-distributed and does not contain bad variables by Definition~\ref{def:distributed-marking}. We are going to apply Corollary~\ref{cor:cc-r1} with $V = \mathcal{V}_{\mathrm m}$ several times in this proof.  In view of Lemma~\ref{lem:marginals}, the distribution $\restr{\mu_\Omega}{\mathcal{V}_{\mathrm m}}$ is $(1/k)$-uniform, and this will be relevant when applying Corollary~\ref{cor:cc-r1}.  We prove that Item~\ref{item:IV:1} holds with probability at least $1-1/(2n)$ and that Item~\ref{item:IV:2} holds with probability $1-1/(2n)$, so the result follows from a union bound.

Item~\ref{item:IV:1}. Let $i \in [\ell]$ and let $\tau_i$ be the restriction of $\zeta_i$ to $\mathcal{V}_{\mathrm m}$. By construction, $\tau_i$ agrees with $\sigma'$ on $v_1, v_2, \dots, v_i$ and it agrees with $\sigma$ on $v_{i+1}, v_{i+2}, \ldots, v_\ell$. Let $\Lambda_i$ denote the restriction of $\tau_i$ on $\mathcal{V}_{\mathrm m} \setminus \{v_i\}$, which agrees with $\zeta_i$ and $\zeta_{i-1}$ on $\mathcal{V}_{\mathrm m} \setminus \{v_i\}$. Recall that, by definition, $\zeta_i$ is the satisfying assignment that extends $\tau_i$ that minimises $\|\zeta_i - \zeta_{i-1}\|_1$, see Algorithm~\ref{alg:find-path}. We consider the connected components of $G_{\Phi^{\Lambda_i}}$, which can be seen as CNF formulae with variables in $\mathcal{V}_{\mathrm c} \cup \{v_i\}$ due to the fact that all marked variables other than $v_i$ are set by $\Lambda_i$. Each one of these connected components are satisfied as CNF formulae by the assignments $\zeta_i$ and $\zeta_{i-1}$. We conclude that $\zeta_i$ and $\zeta_{i-1}$ agree on the variables of all these connected components except for those variables in the connected component of the clauses containing $v_i$, where $\zeta_i$ and $\zeta_{i-1}$ may disagree. Let us denote this connected component by $C_{v_i}$, which is empty when all the clauses containing $v_i$ are satisfied by $\Lambda_i$. We have $\|\zeta_i - \zeta_{i-1}\|_{1} \le  k \lvert C_{v_i} \rvert$, where the factor $k$ comes from the fact that each clause has at most $k$ variables. We now bound the size of $C_{v_i}$. Since the restrictions of $\sigma$ and $\sigma'$ to $\mathcal{V}_{\mathrm m}$ follow $\restr{\mu_\Omega}{\mathcal{V}_{\mathrm m}}$, which is $(1/k)$-uniform, we find, by Definition~\ref{def:uniform}, that $\tau_i$ also follows an $(1/k)$-uniform distribution over the assignments $\mathcal{V}_{\mathrm m} \to \{\mathsf{F}, \mathsf{T}\}$. Let us denote this distribution by $\mu_i$. Then $\Lambda_i \sim \restr{\mu_i}{\mathcal{V}_{\mathrm m} \setminus \{v_i\}}$ and, by Corollary~\ref{cor:cc-r1} with $V = \mathcal{V}_{\mathrm m}$, $\Lambda = \Lambda_i$ and $\mu = \mu_i$, with probability at least $1 - n^{-k}$ over the choice $\Lambda_i \sim \restr{\mu_i}{\mathcal{V}_{\mathrm m} \setminus \{v_i\}}$, the connected component $C_{v_i} \subset G_{\Phi^{\Lambda_i}}$ containing $v_i$ has at most $2k^4 \log n$ clauses. Thus, with probability at least $1 - n^{-k}$, we have $\|\zeta_i - \zeta_{i-1}\|_{1} \le k \lvert C_{v_i} \rvert \le 2k^ 5 \log n$. By a union bound over $i \in [\ell]$ and the fact that $k \ge 3$ and $\ell \le n$, we conclude that, with probability at least $1-1/n^2$, we have $\|\zeta_i - \zeta_{i-1}\|_1 \le  2k^5 \log n$ for all $i \in [\ell]$.

Item~\ref{item:IV:2}. Let $\tau' = \sigma'|_{\mathcal{V}_{\mathrm m}}$ as in Algorithm~\ref{alg:find-path}. By construction, $\xi_0 = \zeta_\ell$ and $\xi_t = \sigma'$ agree with $\tau'$ on $\mathcal{V}_{\mathrm m}$. Since $\sigma' \sim \mu_\Omega$, we have $\tau' \sim \restr{\mu_{\Omega}}{\mathcal{V}_{\mathrm m}}$, which is $(1/k)$-uniform by Lemma~\ref{lem:marginals}. In view of Corollary~\ref{cor:cc-r1} for $V = \mathcal{V}_{\mathrm m}$, $\Lambda = \tau'$ and $\mu = \restr{\mu_{\Omega}}{\mathcal{V}_{\mathrm m}}$, with probability at least $1 - n^{-k}$, all of the connected components of $G_{\Phi^{\tau'}}$, have size at most $2k^4 \log n$. Thus, all the connected components of $H_{\Phi^{\tau'}}$ have size at most $2k^5 \log n$. By construction, see Line~\ref{alg:step3} in Algorithm~\ref{alg:find-path}, the assignments $\xi_{i-1}$ and $\xi_i$ agree on the variables in all the connected components of $H_{\Phi^{\tau'}}$ except for the variables in the $i$-th connected component, where they may disagree. Thus, they disagree on at most $2k^5 \log n$ variables. This gives the desired result.
\end{proof}

We can now complete the proof of Theorem~\ref{t:robustconn}.

\begin{thmconnectivity}
    \statethmconnectivity
\end{thmconnectivity}
\begin{proof}
Since $\alpha \le 2^{0.227 k} \le 2^{(r_1-3\delta) k}/k^3 \le 2^{(r_1 - \delta) k}/k^3$ for large enough $k$, w.h.p. over the choice of $\Phi$, there is an $(r, r_1, 0, r_1)$-marking $(\mathcal{V}_{\mathrm m}, \emptyset, \mathcal{V}_{\mathrm c})$ of $\Phi$, see Lemma~\ref{lem:marking:r1}. We run Algorithm~\ref{alg:find-path} with inputs $\Phi$, and the associated marking $(\mathcal{V}_{\mathrm m}, \emptyset, \mathcal{V}_{\mathrm c})$. W.h.p. over the choice of $\Phi$, Lemma~\ref{lem:IV} holds. Therefore, with probability at least $1 -1/n$ over the choice of two random satisfying assignments $\sigma \sim \mu_\Omega$ and $\sigma' \sim \mu_\Omega$, the output path of Algorithm~\ref{alg:find-path}  is well-defined by Lemma~\ref{lem:I} and satisfies that $\|\zeta_i - \zeta_{i-1}\|_1 \le 2k^5 \log n$ for all $i \in [\ell]$ and $\|\xi_i - \xi_{i-1}\|_1 \le 2 k^5 \log n$ for all $i \in [t]$ by Lemma~\ref{lem:IV}.
Hence, it is a $D$-path in the solution space $\Omega$ for $D = 2k^5 \log n$ as we wanted.
\end{proof}

\subsection{Proof of Theorem~\ref{thm:looseness}} \label{sec:looseness}

We next show $O(\log n)$-looseness for all variables with high probability over $(\Phi, \sigma)$ for random $k$-CNF instances $\Phi$ and uniformly random satisfying assignment $\sigma \in \Omega$. Consequently, in an \textit{algorithmic} regime where $\alpha \ll 2^{ck}$ for some $c < 1$, we resolve a conjecture of~\cite{AC08}. 
Our proof exploits Corollary~\ref{cor:cc-r1} on the size of the connected components of $\Phi^\Lambda$. It is important in our arguments that every variable in the formula is \emph{flippable}.

\begin{definition}
Let $\Phi = \Phi(k, n, m)$ be a random $k$-CNF. A variable $v \in \mathcal{V}$ is \textit{flippable} if there exists a pair of satisfying assignments $(X,Y)$ to $\Phi$, in one of which $X(v) = \mathsf{F}$ and in the other $Y(v) = \mathsf{T}$.
\end{definition}

\begin{lemma}\label{l:mostflippable}
For $\alpha < 2^{k-2}$, with high probability over the choice of $\Phi = \Phi(k, m, n)$, all variables in $\Phi$ are flippable.
\end{lemma}
\begin{proof}
Observe that we can define an NAE-SAT problem based on $\Phi$. By definition, any NAE-satisfying assignment ensures that every clause contains at least one satisfied literal and at least one unsatisfied literal. By Theorem~2 in~\cite{achlioptas2002}, with high probability $\Phi$ is NAE-satisfiable. 
Consequently, we can find some assignment $\sigma$ that NAE-satisfies $\Phi$ with high probability, and then the opposite assignment $\overline{\sigma}$ also NAE-satisfies $\Phi$ by the symmetry of NAE-SAT solutions. In particular, both $\sigma$ and $\overline{\sigma}$ are solutions to the original SAT formula $\Phi$. Observe that for every variable $v \in V$ we have $X(v) = \mathsf{T}$ and $X(v) = \mathsf{F}$ in exactly one of $\sigma, \overline{\sigma}$ and thus, with high probability, every variable in $\Phi$ is flippable.
\end{proof}

\begin{lemma} \label{lem:XY}
For any variable $v \in \mathcal{V}$ and any partial assignment $X \colon \mathcal{V}_{\mathrm m} \setminus \{v\} \to \{\mathsf{F}, \mathsf{T}\}$, we have
\[
\pr_{\mu_\Omega}(v \mapsto \mathsf{F} |X) >0 
\quad\text{and}\quad
\pr_{\mu_\Omega}(v \mapsto \mathsf{T} |X) > 0.
\]
\end{lemma}
\begin{proof}
We prove $\pr_{\mu_\Omega}(v \to \mathsf{F} |X) > 0$; the proof of $\pr_{\mu_\Omega}(v \to \mathsf{T} |X) > 0$ is analogous. We distinguish two cases.

The first case is when $v$ is a good variable.  Lemma~\ref{lem:marginals} gives $\pr_{\mu_\Omega}\left( v \mapsto \mathsf{F} \vert X, \Lambda_{\bad} \right) \ge 1 - \exp(1/k)/2 > 0$ for any satisfying assignment of the bad clauses $\Lambda_{\bad}$. Thus, we have $\pr_{\mu_\Omega}\left( v \mapsto \mathsf{F} \vert X \right) > 0$. 

The second case is when $v$ is a bad variable. By Lemma~\ref{l:mostflippable} there exists a satisfying assignment $\sigma$ with $\sigma(v) = \mathsf{F}$. 
Let $\Lambda_{\bad} = \restr{\sigma}{\bv}$ be the assignment on bad variables and so in particular $\pr_{\mu_\Omega}(\Lambda_{\bad}) > 0$.
Then by Lemma~\ref{lem:marginals} we have $\pr_{\mu_\Omega}\left( X \vert \Lambda_{\bad} \right) \ge (1 - \exp(1/k)/2)^{\lvert \mathcal{V}_{\mathrm m}\rvert} > 0$. This implies that $\pr_{\mu_\Omega}(X, \Lambda_{\bad}) > 0$
and in particular $\pr_{\mu_\Omega}(v \mapsto \mathsf{F}, X) > 0$, so $\pr_{\mu_\Omega}(v \mapsto \mathsf{F} \vert X) > 0$.
\end{proof}

We can now prove Theorem~\ref{thm:looseness} with the help of Corollary~\ref{cor:cc-r1}.

\begin{thmlooseness}
    \statethmlooseness
\end{thmlooseness}
\begin{proof} 
    Note that $2^{0.227 k} \le 2^{(r_1-3\delta)k} \le 2^{(r_1 - \delta)k}/k^3$ for large enough $k$. Thus, w.h.p. over the choice of $\Phi$, there is an $(r, r_1, \emptyset, r_1)$-marking $(\mathcal{V}_{\mathrm m}, \emptyset, \mathcal{V}_{\mathrm c})$ of $\Phi$, see Lemma~\ref{lem:marking}. The distribution $\restr{\mu_\Omega}{\mathcal{V}_{\mathrm m}}$ is $(1/k)$-uniform by Lemma~\ref{lem:marginals}. Hence, Corollary~\ref{cor:cc-r1} holds for $V = \mathcal{V}_{\mathrm m}$ and $\mu = \restr{\mu_\Omega}{\mathcal{V}_{\mathrm m}}$. Let $v$ be a variable of $\Phi$. Let $\sigma \sim \mu_\Omega$ and let $\Lambda$ be the restriction of $\sigma$ to $\mathcal{V}_{\mathrm m} \setminus \{v\}$. Then, with probability at least $1-n^{-k}$, the connected components of $G_{\Phi^{\Lambda}}$ have size at most $2k^4 \log n$. Let $\mathcal{C}^\Lambda_j$ be the connected component containing the variable $v$, which is empty if all clauses containing $v$ are satisfied. Let $\omega$ be the negation of $\sigma(v)$. By Lemma~\ref{lem:XY}, we have $\pr_{\mu_\Omega}(v \mapsto \omega \vert \Lambda) > 0$. Therefore, there is an assignment $Y$ of the variables in $\var(\mathcal{C}^\Lambda_j)$ that satisfies the clauses in $\mathcal{C}^\Lambda_j$ and has $Y(v) = \omega$. We construct the assignment $\sigma'$ that has $\sigma' (v) = \omega$, agrees with $Y$ in $\var(\mathcal{C}^\Lambda_j)$ and agrees with $\sigma$ in the rest of the variables of $\Phi$. In particular, this assignment agrees with $\Lambda$ and satisfies each one of the connected components of $\Phi^\Lambda$. Thus, $\sigma'$ is a satisfying assignment of $\Phi$. Moreover, w.h.p. $\sigma'$ differs with $\sigma$ in at most $2k^5 \log n$ variables (the variables in $\var(\mathcal{C}^\Lambda_j)$).  We have shown that, w.h.p. over the choice of $\Phi$, with probability at least $1 - n^{-k}$ a random assignment $\sigma \sim \mu_{\Omega}$ is $(2k^5 \log n)$-loose, so the statement follows.
\end{proof}

\bibliographystyle{plainurl}
\bibliography{sat}

 \begin{appendices}

 \section{Proof of Lemma~\ref{lem:bad}} \label{sec:ap:bad}
 
 In this section we prove Lemma~\ref{lem:bad}. Recall that this result is~\cite[Lemma 8.16]{galanis2019counting} with a less restrictive bound on the density of the formula and a more restrictive definition of good variables/clauses, see Section~\ref{sec:bad} for details.  Moreover, the obtained upper bound on the number of bad clauses in our version of~\cite[Lemma 8.16]{galanis2019counting} is tighter. The original proof of Lemma~\ref{lem:bad} given in~\cite[Section 8]{galanis2019counting} is split into a sequence of results on random formulae. Here we restate 
 some of these results  ---
 only those
whose statement needs to change  as a consequence of our definition of good variables/clauses and the tighter upper bound. We also explain how these changes affect the proofs
if any modifications are necessary. 

 Let us fix some notation first. The results stated in this section only hold for large enough $k$ unless  we say otherwise. We note that in~\cite{galanis2019counting} the density $\alpha$ is at most $2^{k/300}/k^3$ and $\Delta = 2^{k/300}$, where $\Delta$ is the threshold in the definition of high-degree variables, and the proofs are carried out for these particular values. It turns out that, following the proofs in~\cite[Section 8]{galanis2019counting}, the only properties of $\alpha$ and $\Delta$ needed in order to proof Lemma~\ref{lem:bad} are that, for $r \in (0, 1/(2 \log 2))$, we have $\Delta_r = \lceil 2^{rk} \rceil$ and $\alpha$ is bounded above by $\Delta_r / k^3$ (note the subscript $r$ here to indicate that $\Delta_r$ depends on $r$). First, we need some definitions. For any set of variables $S \subseteq \mathcal{V}$ of $\Phi$, we denote by $\operatorname{HD}(S, r)$ the set of high-degree variables in $S$ (recall that a variable is of high-degree if the degree of $v$ is at least $\Delta_r$). 

\begin{lemma}[{\cite[Lemma 8.1]{galanis2019counting}}]\label{lem:v0}
  Let $r \in (0,1)$. There is a positive integer $k_0$ such that for any integer $k \ge k_0$, $\Delta_r = \lceil 2^{r k} \rceil$, and any density $\alpha$ with $\alpha \le \Delta_r/k^3$, the following holds w.h.p. over the choice of $\Phi = \Phi(k, n, \lfloor \alpha n \rfloor)$. The size of $\mathcal{V}_0(r) := \operatorname{HD}(\mathcal{V}, r)$ is at most $(\alpha/ \Delta_r) n / 2^{k^{10}}$.
\end{lemma}
\begin{proof}
    The proof is the same to that of~\cite[Lemma 8.1]{galanis2019counting}, apart from one change that we highlight here. The degrees of the variables in $\Phi$ have the same distribution as a balls-and-bins experiment with $km$ balls and $n$ bins. Let $D_1, D_2, \ldots, D_n$ be independent variables following the Poisson distribution $\operatorname{Poi}(\mu)$ with parameter $\mu = k \alpha$. The degrees of the variables of $\Phi$ have the same distribution as $\{D_1, D_2, \ldots, D_n\}$ conditioned on the event $\mathcal{E}$ that $D_1 + D_2 + \cdots + D_n = m$, see for instance~\cite[Chapter 5.4]{pcbook}. Let $U = \{i \in [n] : D_i \ge \Delta_r\}$. We want to show that $\pr(\lvert U \rvert > (\alpha/ \Delta_r) n / 2^{k^{10}} \vert \mathcal{E}) = o(1)$. In~\cite[Lemma 8.1]{galanis2019counting} the authors show that $\pr(\lvert U \rvert  >  n / 2^{k^{10}} \vert \mathcal{E}) = o(1)$. Their bound is not tight, but it is enough for their purposes. In fact, one can change $k^{10}$ by any polynomial and the result would still hold for large enough $k$. Here we obtain the extra factor $\alpha/\Delta_r$ by slightly modifying the application of the tail bound $\pr(\operatorname{Poi}(\mu) \ge x) \le e^{-\mu} (e \mu)^x / x^{x}$. For $x = \Delta_r$, instead of using the bound $e^{-\mu} (e \mu)^x / x^{x} \le e^{- \Delta_r} \le 2^{-k^{10}-1}$, which holds for large enough $k$ as $\mu/x \le k^{-2}$ and $\Delta_r$ is exponential in $k$, we use the bound $e^{-\mu} (e \mu)^x / x^{x} \le (e \mu / x) e^{- x + 1}  \le (\alpha / \Delta_r) 2^{-k^{10}-1}$. The rest of the proof is analogous; we have $\mathbb{E}[\lvert U \rvert] \ge n (\alpha / \Delta_r) 2^{-k^{10}-1}$, so by a Chernoff bound we find that $\pr(\lvert U \rvert \ge (\alpha/ \Delta_r) n / 2^{k^{10}}) \le \exp(-\Omega(n))$. From the connection between a balls-and-bins experiment and the Poisson distribution, see~\cite[Theorem 5.7]{pcbook}, we conclude that $\pr(\lvert U \rvert \ge (\alpha/ \Delta_r) n / 2^{k^{10}} \vert \mathcal{E}) \le \exp(-\Omega(n))$ as we wanted. 
\end{proof}
  
 \begin{corollary}[{\cite[Corollary 8.4]{galanis2019counting}}]  \label{cor:bt}
    There is a positive integer $k_0$ such that for any integer $k \ge k_0$ and any density $\alpha$ with $\alpha \le 2^k/(ek^3)$ the following holds w.h.p. over the choice of $\Phi = \Phi(k, n, \lfloor \alpha n \rfloor)$.  For every set of variables $Y$ such that $2 \le \lvert Y \rvert \le n / 2^k$, the number of clauses that contain at least $3$ variables from $Y$ is at most $\lvert Y\rvert$.
 \end{corollary}
 \begin{proof}
   This is a consequence of~\cite[Lemma 35]{galanis2019counting} with $b = 3$ and $t = 2/(b-1) = 1$, whose proof only requires $\alpha \le 2^k/(ek^3)$.
 \end{proof}

Recall that the graph $H_\Phi$ is the dependency graph of the variables of $\Phi$, see Definition~\ref{def:graph-H}.
 
 \begin{lemma}[{\cite[Lemma 8.8]{galanis2019counting}}]\label{lem:hd:upper}
  Let $r \in (0,1)$. There is a positive integer $k_0$ such that for any integer $k \ge k_0$, $\Delta_r = \lceil 2^{r k} \rceil$, and any density $\alpha$ with $\alpha \le \Delta_r/k^3$, the following holds w.h.p. over the choice of $\Phi = \Phi(k, n, \lfloor \alpha n \rfloor)$. Every connected set $U$ of variables in $H_\Phi$ with size at least $2 k^4 \log n$ satisfies that $\lvert \hd(U, r) \rvert \le \frac{1}{2k^3} \lvert U \rvert $. 
 \end{lemma}
 \begin{proof}
   The proof is that of~\cite[Lemma 8.8]{galanis2019counting}, with the difference that $\delta_0= 1/(2 k^3)$ instead of $\delta_0 = 1/21600$, as the exact value of $\delta_0$ does not play a role in the proof as long as, for $\theta_0 = \Delta_r - 2(k+1)$, we have $\delta_0 \theta_0 \log \frac{\theta_0}{k^2 \alpha} \ge 3 \log k +  \log \alpha$, which holds for large enough $k$ when $\delta_0 = \operatorname{poly}(k)$. Moreover, the only restriction on $\alpha$ is that of Corollary~\ref{cor:bt}, and the fact that $\alpha \le \Delta_r /k^3$.
 \end{proof}

  \begin{lemma}[{\cite[Lemma 2.4]{coja2014} and~\cite[Lemma 8.10]{galanis2019counting}}] \label{lem:coja}
  Let $k \ge 3$ be an integer and let $\alpha$ be a positive real number with $\alpha  \le e^{k/2} / (2 e^2 k^2)$. For any $\varepsilon \in [1/n,1)$ (depending on $n$) such that $\varepsilon <  e^{-3k}$ for all $n$, the following holds w.h.p. over the choice of the random formula $\Phi = \Phi(k, n, \lfloor \alpha n \rfloor)$. Let $Z$ be a set of clauses with size at most $\varepsilon n$ and let $c_1, \ldots, c_l \in \mathcal{C} \setminus Z$ be distinct clauses. For $s \in \{1, 2, \ldots, \ell\}$, let $N_s := \var(Z) \cup \bigcup_{j = 1}^{s-1} \var(c_{j})$. If $\left| \var(c_s) \cap N_s \right| \ge 3$ for all $s \in \{1, 2, \ldots, \ell\}$, then $\ell \le \varepsilon n$.
  \end{lemma}
  \begin{proof}
    The proof is almost identical to the proof of~\cite[Lemma 2.4]{coja2014}. There are four differences. First, here, as it is also the case in~\cite[Lemma 44]{galanis2019counting}, $\varepsilon$ can depend on $n$. This will arise later in this proof. Second, the proof of~\cite[Lemma 2.4]{coja2014} is carried out for the condition $\left| \var(c_s) \cap N_s \right| \ge \lambda$, where $\lambda $ is an integer with $\lambda > 4$. Here we set $\lambda = 3$ and impose stricter hypotheses on $\alpha$ and $\varepsilon$ to compensate for a smaller $\lambda$. Their (more relaxed) hypotheses on $\alpha$ and $\varepsilon$ are $\alpha \le 2^k \log 2$, $\varepsilon \le k^{-3}$ and $\varepsilon^\lambda \le (2e)^{-4k} / e$. Third, we substitute the last inequality of~\cite[Equation 4]{coja2014}, which is 
      \begin{equation*}
      \left[ \left(\frac{e m/n}{\varepsilon} \right)^2 \exp(2k) (2k\varepsilon)^\lambda \right]^{\varepsilon n}
      \le  \left[ \left(2 e \right)^{2k}  \varepsilon^{\lambda /2} \right]^{\varepsilon n},
    \end{equation*}
    by the inequality
    \begin{equation} \label{eq:Ns}
    \begin{aligned}
      \left[ \left(\frac{e m/n}{\varepsilon} \right)^2 \exp(2k) (2k\varepsilon)^\lambda \right]^{\varepsilon n}
      & \le  \left[ \left(e m/n \right)^2 \exp(2k) (2k)^3 \varepsilon \right]^{\varepsilon n} \\
      & \le  \left[ \exp(3k -1) \varepsilon \right]^{\varepsilon n},
    \end{aligned}
    \end{equation}
    where we used $\lambda = 3$ and $m / n \le \alpha \le e^{k/2} / (2 e^2 k^2) $. Now, as it is done in~\cite[Lemma 8.10]{galanis2019counting}, we distinguish two cases depending on $\varepsilon$. If $\varepsilon \ge 10 (\log n)/ n$, then using this in conjunction with $\varepsilon <  e^{-3k}$, the right hand size of~\eqref{eq:Ns} is bounded by  $e^{-\varepsilon n} \le 1/n^{10} = o(1/n)$. If $1/n \le \varepsilon < 10 (\log n)/ n$, then, for large enough $n$, the right hand size of~\eqref{eq:Ns} is bounded above by $\exp(3k-1) \varepsilon = o(1)$. The last difference between the proofs is that our argument works for all $k \ge 3$, whereas the bound~\cite[Equation 4]{coja2014}  only holds for large $k$.
  \end{proof}

  The remaining results in this section do not need any changes in their original proofs, other than that every time Corollary~8.4, Lemma~8.8 and Lemmas~8.10-8.16 are invoked in~\cite[Section 8]{galanis2019counting}, we use the version given in this appendix instead. We note that the statements of these results are slightly different to their~\cite[Section 8]{galanis2019counting} versions, and these changes are again due to the fact that we use $\lambda = 3$ instead of $\lambda = k/10$ in the definition of good variables/clauses.

\begin{corollary}[{\cite[Corollary 8.11]{galanis2019counting}}] \label{cor:Ns}
  Let $r \in (0, 1/(2\log 2)]$. There is a positive integer $k_0$ such that for any integer $k \ge k_0$, $\Delta_r = \lceil 2^{r k} \rceil$, and any density $\alpha$ with $\alpha \le \Delta_r/k^3$, the following holds w.h.p. over the choice of $\Phi = \Phi(k, n, \lfloor \alpha n \rfloor)$. Let $Z$ be a set of clauses with size at most $2 n / 2^{k^{10}}$ and let $c_1, \ldots, c_l \in \mathcal{C} \setminus Z$ be distinct clauses. For $s \in \{1, 2, \ldots, \ell\}$, let $N_s := \var(Z) \cup \bigcup_{j = 1}^{s-1} \var(c_{j})$. If $\left| \var(c_s) \cap N_s \right| \ge 3$ for all $s \in \{1, 2, \ldots, \ell\}$, then $\ell \le \lvert Z \rvert$.
\end{corollary}
\begin{proof}
    The proof given in~\cite[Corollary 8.11]{galanis2019counting} also applies here. We note that the density $\alpha$ is at most $ e^{k/2} / (2 e^2 k^2)$ so we can indeed apply Lemma~\ref{lem:coja} when  the proof given in~\cite[Corollary 8.11]{galanis2019counting} invokes~\cite[Lemma 8.10]{galanis2019counting}.
\end{proof}

\begin{lemma}[{\cite[Lemma 8.13]{galanis2019counting}}] \label{lem:hd:lower}
  Let $r \in (0, 1/(2\log 2)]$. There is a positive integer $k_0$ such that for any integer $k \ge k_0$, $\Delta_r = \lceil 2^{r k} \rceil$, and any density $\alpha$ with $\alpha \le \Delta_r/k^3$, the following holds w.h.p. over the choice of $\Phi = \Phi(k, n, \lfloor \alpha n \rfloor)$. For any bad component $S$ of variables, we have $\lvert S \rvert \le 2k \lvert \hd(S, r) \rvert$.
\end{lemma}
\begin{proof}
 The proof given in~\cite[Lemma 8.13]{galanis2019counting} applies using our versions of~\cite[Lemma 8.1, Corollary 8.4 and Corollary 8.11]{galanis2019counting}.
\end{proof}

\begin{lemma}[{\cite[Lemma 8.14]{galanis2019counting}}] \label{lem:bad-component}
  Let $r \in (0, 1/(2\log 2)]$. There is a positive integer $k_0$ such that for any integer $k \ge k_0$, $\Delta_r = \lceil 2^{r k} \rceil$, and any density $\alpha$ with $\alpha \le \Delta_r/k^3$, the following holds w.h.p. over the choice of $\Phi = \Phi(k, n, \lfloor \alpha n \rfloor)$. Every bad component $S$ has size at most $2k^4 \log n$.
\end{lemma}
\begin{proof}
 The proof given in~\cite[Lemma 8.14]{galanis2019counting} applies using our versions of~\cite[Lemma 8.8 and Lemma 8.13]{galanis2019counting}.
\end{proof}

\begin{lemma}[{\cite[Lemma 8.15]{galanis2019counting}}] \label{lem:bad-variables}
Let $r \in (0, 1/(2\log 2)]$. There is a positive integer $k_0$ such that for any integer $k \ge k_0$, $\Delta_r = \lceil 2^{r k} \rceil$, and any density $\alpha$ with $\alpha \le \Delta_r/k^3$, the following holds w.h.p. over the choice of $\Phi = \Phi(k, n, \lfloor \alpha n \rfloor)$. For every connected set of $S$ variables with size at least $2 k^4 \log n$, we have $\lvert S \cap \bv \rvert \le \lvert S \rvert / k^2$.
\end{lemma}
\begin{proof}
  The proof is analogous to that given in~\cite[Lemma 8.15]{galanis2019counting}. The only differences are that we apply Lemma~\ref{lem:hd:upper} instead of~\cite[Lemma 8.8]{galanis2019counting}, we apply Lemma~\ref{lem:hd:lower} instead of~\cite[Lemma 8.13]{galanis2019counting}, and we have $\delta_0 = 1/(2 k^3)$ instead of $\delta_0 = 1 / 21600$.
\end{proof}

\begin{lembad}[{\cite[Lemma 8.16]{galanis2019counting}}]
 \statelembad
\end{lembad}
\begin{proof} \label{lem:bad:proof}
 The same proof applies using our versions of~\cite[Corollary 8.4 and Lemma 8.15]{galanis2019counting}.
\end{proof}

\begin{lembadall}[{\cite[Lemma 8.12]{galanis2019counting}}]
 \statelembadall
\end{lembadall}
\begin{proof} 
 We consider the set of high-degree variables $\mathcal{V}_0(r) = \operatorname{HD}(\mathcal{V},r)$, which w.h.p. over the choice of $\Phi$ has $\lvert \mathcal{V}_0(r) \rvert \le (\alpha/ \Delta_r)  n / 2^{k^{10}}$ by Lemma~\ref{lem:v0}. In view of Corollary~\ref{cor:bt} with $Y = \mathcal{V}_0(r)$, we have $\lvert \mathcal{C}_0(r) \rvert \le \lvert \mathcal{V}_0(r) \rvert \le n / 2^{k^{10}}$, where $\mathcal{C}_0(r)$ is the set of clauses with at least $3$ variables in $\mathcal{V}_0(r)$, see Algorithm~\ref{alg:bad}. From Corollary~\ref{cor:Ns} and the construction of $\bc(r)$ in Algorithm~\ref{alg:bad}, we find that $\lvert \bc(r) \rvert \le  2\lvert \mathcal{C}_0(r) \rvert \le 2\lvert \mathcal{V}_0(r) \rvert \le 2 (\alpha/ \Delta_r) n / 2^{k^{10}}$. By construction of $\bv(r)$, see Algorithm~\ref{alg:bad}, we conclude that $\lvert \bv(r) \rvert \le \lvert \mathcal{V}_0(r) \rvert + k \lvert \bc(r) \rvert \le 2(k+1) (\alpha/ \Delta_r) n / 2^{k^{10}}$.
\end{proof}

\newpage 

\section{Proof of Lemma~\ref{lem:block-glauber}} \label{sec:ap:mt}

In this section we collect the results from~\cite{Chen2021} that one needs to combine to obtain Lemma~\ref{lem:block-glauber} on the mixing time of the $\rho$-uniform-block Glauber dynamics. 

\begin{definition} \label{def:entropy-factorisation}
    Let $\mu$ be a distribution supported on $\Omega \subseteq [q]^V$. Let $f \colon \Omega \to \mathbb{R}_{\ge 0}$. We denote the entropy of $f$ by $\operatorname{Ent}_\mu(f)$, that is, $\operatorname{Ent}_\mu(f) = \mathbb{E}_{\mu}(f \log f )) - \mathbb{E}_{\mu}(f) \log(\mathbb{E}_{\mu}(f))$ when $\mathbb{E}_{\mu}(f) > 0$, and $\operatorname{Ent}_\mu(f) = 0$ when $\mathbb{E}_{\mu}(f) = 0$. For $S \subseteq V$, we denote $\operatorname{Ent}_\mu^S(f) = \mathbb{E}_{\tau \sim \restr{\mu}{V \setminus S}}\operatorname{Ent}_{\mu}(f \vert \, \tau)$, where $\operatorname{Ent}_{\mu}(f \vert \, \tau)$ is the entropy of $f$ conditioning to the event that the assignment drawn from $\mu$ agrees with $\tau$ in $V \setminus S$.
    
    Let $\rho \in \{1, 2, \ldots, n \}$.
    We say that $\mu$ satisfies the $\rho$-uniform block factorisation of entropy (with constant $C_\rho$) if for all $f \colon \Omega \to \mathbb{R}_{\ge 0}$ we have
    \begin{equation*}
        \frac{\rho}{n} \operatorname{Ent}_\mu(f) \le C_\rho \frac{1}{\binom{n}{\rho}} \sum_{S \in \binom{V}{\rho}} \operatorname{Ent}_\mu^S(f).
    \end{equation*}
\end{definition}

One of the main results of~\cite{Chen2021} is showing that $\mu$ satisfies the $\rho$-uniform block factorisation of entropy when the distribution $\mu$ is $\eta$-spectrally independent and $b$-marginally bounded. In the proof of~\cite[Corollary 19]{bez2021} the authors observe that the proof of Lemma~\ref{lem:factorisation-entropy} also holds when $\eta$ depends on $n$ and, in particular, in the case $\eta = \epsilon \log n$.

 \begin{lemma}[ {\cite[Lemma 2.4]{Chen2021}}] \label{lem:factorisation-entropy}
  The following holds for any reals $b, \eta > 0$,  any $\kappa \in (0,1)$ and any integer $n$ with $n \ge \frac{2}{\kappa} (4\eta / b^2 + 1)$.
 
 Let $q \ge 2$ be an integer, let $V$ be a set of size $n$ and let $\mu$ be a distribution over $[q]^V$. If $\mu$ is $b$-marginally bounded and $\eta$-spectrally independent, then $\mu$ satisfies the $\lceil \kappa n \rceil$-uniform block factorisation of entropy with constant $C = (2/\kappa)^{4 \eta / b^2 + 1}$.
 \end{lemma}
 
It turns out that one can bound the mixing time of the $\rho$-uniform-block Glauber dynamics when the target distribution $\mu$ satisfies the $\rho$-uniform block factorisation of entropy.
 
\begin{lemma}[{See, e.g.,~\cite[Lemma~2.6 and Fact 3.5(4)]{Chen2021} or~\cite[Lemma~17]{bez2021}}] \label{lem:block-glauber:mt}
  Let $q \ge 2$, $\rho \ge 1$ be integers and $V$ be a set of size $n \ge \rho + 1$. Let $\mu$ be a distribution supported on $\Omega \subseteq [q]^V$ that satisfies the $\rho$-uniform-block factorisation of entropy with multiplier $C_\rho$. Then, for any $\varepsilon > 0$, the mixing time of the $\rho$-uniform-block Glauber dynamics on $\mu$ satisfies, for $\mu_{\min} = \min_{\Lambda \in \Omega} \mu(\Lambda)$,
  \begin{equation*}
      T_{\mix}(\varepsilon) \le \left\lceil C_\rho \frac{n}{\rho} \left( \log \log \frac{1}{\mu_{\min}} + \log \frac{1}{2 \varepsilon^2} \right) \right\rceil.
  \end{equation*}
\end{lemma}

\begin{proof}[Proof of   Lemma~\ref{lem:block-glauber}] \label{lem:block-glauber:proof} 
The proof of Lemma~\ref{lem:block-glauber} follows directly from combining  Lemmas~\ref{lem:factorisation-entropy} and~\ref{lem:block-glauber:mt}.
\end{proof}

\newpage

\section{Notation and definitions reference} \label{ap:notation}

Here we gather the notation and definitions that are used globally in our work. If some notation is not here, then it is only used in one section of our work (and it is defined in that section).

\subsection{Table of notation}

\begin{table}[H]
    \centering
    \begin{tabular}{l|l|l}
        \textbf{Notation} & \textbf{Description} & \textbf{Reference} \\
        \hline
        $\Phi(k, n, m)$ & A random $k$-CNF formula with $n$ variables and $m$ clauses. & Section~\ref{sec:intro} \\ 
        $\alpha$ & The density of the formula $\Phi$, so $\alpha = m/n$. & Section~\ref{sec:intro}  \\
        $\mathcal{V}$ & The set of variables of $\Phi$.  & Section~\ref{sec:intro} \\
        $\mathcal{C}$ & The set of clauses of $\Phi$. & Section~\ref{sec:intro} \\
        w.h.p. & Stands for ``with high probability". & Section~\ref{sec:intro} \\ 
        $d_{\TV}$ & The total variation distance between two distributions. & Section~\ref{sec:intro} \\
        $\xi$ & Our sampling algorithm has error at most $n^{-\xi}$.  & Theorem~\ref{thm:sampling} \\
        $\Delta_r$ & The high-degree threshold, set to $\mdegdef$. & Definition~\ref{def:degree} \\
        $r_0$, $r_1$, $\delta$ & $r_0 = \rvalue$, $r_1 = 0.227092$ and $\delta = \deltadef$. & Definition~\ref{def:distributed-marking} \\
        $\var(c)$ & The set of variables in a clause $c$. & Section~\ref{sec:po:marking} \\
        $\var(S)$ & The set of variables $\bigcup_{c \in S} \var(c)$. & Section~\ref{sec:po:marking} \\
        $\gc(r)$, $\bc(r)$ & Good and bad clauses, a partition of $\mathcal{C}$. & Section~\ref{sec:bad} \\
        $\gv(r)$, $\bv(r)$ & Good and bad variables, a partition of $\mathcal{V}$. & Section~\ref{sec:bad} \\
        $\mathcal{V}_{\mathrm m}$, $\mathcal{V}_{\mathrm a}$, $\mathcal{V}_{\mathrm c}$ & The sets of marked, auxiliary and control variables. & Definition~\ref{def:distributed-marking} \\
        $\Omega^*$ & The set of all assignments $\mathcal{V} \to \{\mathsf{F}, \mathsf{T} \}$ & Definition~\ref{def:subformula} \\
        $\Omega$ & The set of satisfying assignments of $\Phi$. & Definition~\ref{def:subformula} \\
        $\mu_A$ & The uniform distribution over $A \subseteq \Omega*$. & Definition~\ref{def:subformula} \\
        $\Phi^\Lambda$ & The formula $\Phi$ simplified under $\Lambda$. & Definition~\ref{def:subformula} \\
         $\mathcal{V}^{\Lambda}$, $\mathcal{C}^{\Lambda}$ & The variables and clauses of $\Phi^\Lambda$ & Definition~\ref{def:subformula} \\
        $\Omega^\Lambda$ & The set of satisfying assignments of $\Phi^\Lambda$. & Definition~\ref{def:subformula} \\
        $\restr{\mu}{V}$ & The marginal distribution of $\mu$ on $V$. & Definition~\ref{def:marginal} \\
        $T_{\mix}(\rho, \varepsilon)$ & The mixing time of the $\rho$-uniform-block Glauber dynamics. & Section~\ref{sec:po:si:block} \\
        $\mathcal{I}^\Lambda(u \to v)$ & The influence of $u$ on $v$ (under $\Lambda$). & Section~\ref{sec:po:si:block}, \eqref{eq:influence} \\
        $G_\Phi$ &  The dependency graph of $\mathcal{C}$.  & Definition~\ref{def:graph-phi} \\
        $H_\Phi$ &  The dependency graph of $\mathcal{V}$.  & Definition~\ref{def:graph-H} \\
        $\Phi_{\mathrm{good}}(r)$& The subformula of $\Phi$ with all good variables and good clauses. & Definition~\ref{def:good-phi} \\ 
         $\Phi_{\mathrm{bad}}(r)$ & The subformula of $\Phi$ with all bad variables and bad clauses. & Definition~\ref{def:good-phi} 
    \end{tabular}
\end{table}

\subsection{Table of definitions }
\begin{table}[H]
    \centering
    \begin{tabular}{l|l}
        \textbf{Name} & \textbf{Reference} \\
        \hline
        high-degree & Definition~\ref{def:degree}, page~\pageref{def:degree}\\
        $r$-distributed & Definition~\ref{def:distributed-marking}, page~\pageref{def:distributed-marking} \\
        $(r, r_{\mathrm m}, r_{\mathrm a}, r_{\mathrm c})$-marking & Definition~\ref{def:distributed-marking}, page~\pageref{def:distributed-marking} \\
        $\varepsilon$-uniform & Definition~\ref{def:uniform}, page~\pageref{def:uniform} \\
        $b$-marginally bounded & Section~\ref{sec:po:si:block}, page~\pageref{sec:po:si:block} \\
        $\eta$-spectrally independent & Section~\ref{sec:po:si:block}, page~\pageref{sec:po:si:block} \\
    \end{tabular}
\end{table}
\end{appendices}

\end{document}